\documentclass[12pt]{article} 

\usepackage{a4wide}

\usepackage{graphics,graphicx,epsfig} 
\usepackage{amsfonts, amsmath, amssymb, wasysym}
\usepackage{pstricks, pst-node, pst-tree}
\usepackage{stmaryrd}

\usepackage[thmmarks]{ntheorem}
\theoremsymbol{\ensuremath{\dashv}}
\usepackage{newproof}  

\newtheorem{theorem}{Theorem}
\newtheorem{example}[theorem]{Example}
\newtheorem{definition}[theorem]{Definition}

\newtheorem{proposition}[theorem]{Proposition}

\newtheorem{corollary}[theorem]{Corollary}

\newtheorem{lemma}[theorem]{Lemma}
\newtheorem{Cla}[theorem]{Claim}

\newtheorem{remark}[theorem]{Remark}

\renewcommand{\phi}{\varphi}
\newcommand{\Et}{\bigwedge}
\newcommand{\Vel}{\bigvee}

\newcommand{\Agents}{A}
\newcommand{\Group}{B}
\newcommand{\Atoms}{P}

\newcommand{\States}{S}
\newcommand{\imp}{\rightarrow}

\newcommand{\Imp}{\Rightarrow}

\newcommand{\state}{s}
\newcommand{\agent}{a}
\newcommand{\atom}{p}
\newcommand{\atomb}{q}
\newcommand{\lang}{\mathcal{L}}
\newcommand{\bisrel}{\ensuremath{\mathfrak{R}}}
\newcommand{\T}{\top}
\newcommand{\proves}{\vdash}
\newcommand{\eq}{\leftrightarrow}
\newcommand{\bisim}{\simeq}
\newcommand{\et}{\wedge}
\newcommand{\vel}{\vee}
\newcommand{\agentb}{b}
\newcommand{\stateb}{t}

\newcommand{\G}{\forall}
\newcommand{\F}{\exists}

\newcommand{\B}{\bot}
\newcommand{\M}{\raisebox{-1pt}{$\Diamond$}}
\newcommand{\know}{\raisebox{-1pt}{$\Box$}}
\newcommand{\knows}{\know}
\newcommand{\susp}{\raisebox{-1pt}{$\Diamond$}}
\newcommand{\simul}{\preceq}
\newcommand{\lumis}{\succeq}

\newcommand{\covers}{\nabla}

\newcommand{\amodel}{\ensuremath{\mathsf{M}}}
\newcommand{\arel}{\ensuremath{\mathsf{R}}}
\newcommand{\Actions}{\ensuremath{\mathsf{S}}}
\newcommand{\action}{\ensuremath{\mathsf{s}}}

\newcommand{\actionb}{\ensuremath{\mathsf{t}}}

\newcommand{\pre}{\mathsf{pre}}

\newcommand{\weg}[1]{}
\newcommand{\logicK}{{\mathsf{K}}}

\newcommand{\logicRML}{{\mathsf{RML}}}
\newcommand{\axiomRML}{{\mathbf{RML}}}
\renewcommand{\phi}{\varphi}
\newcommand{\bqall}{{\tilde{\forall}}}
\newcommand{\bqis}{{\tilde{\exists}}}

\newcommand{\Dia}{\Diamond}
\newcommand{\dia}[1]{\langle #1 \rangle}
\newcommand{\union}{\cup}

\newcommand{\Eq}{\Leftrightarrow}

\usepackage{xspace}

\newcommand{\CTL}{{\mathsf{CTL}}}
\newcommand{\CTLfrag}{{\mathsf{CTL^-}}}
\newcommand{\QLTL}{{\mathsf{QPTL}}}
\newcommand{\RCTL}{\mathsf{CTL_\G^-}}

\newcommand{\AGR}{{\mathsf {AG}}} 
\newcommand{\EGR}{{\mathsf {EG}}} 
\newcommand{\EFR}{{\mathsf {EF}}} 
\newcommand{\AFR}{{\mathsf {AF}}} 
\newcommand{\eventually}{{\mathsf {F}}} 
\newcommand{\always}{{\mathsf {G}}} 
\newcommand{\Next}{{\mathsf {X}}} 

 \def\EXPSPACE{{\sc Expspace}}

\def\words{\text{{\sffamily words}}\xspace}
\def\proj{\text{{\sffamily proj}}\xspace}

\newcommand{\statec}{u}   
\newcommand{\stated}{v}

\newcommand{\calG}[1]{\mathcal{G}_{#1}}

\newcommand{\Spoiler}{\text{\tt Spoiler}}
\newcommand{\Duplicator}{\text{\tt Duplicator}}
\newcommand{\refrel}{\ensuremath{\mathfrak{R}}}

\newcommand{\AG}{{\bf A\!G}}       

\newcommand{\translate}[1]{\widetilde{#1}}

\newcommand{\CTLtoLmu}{\tau}
\newcommand{\seqQ}{\chi}
\newcommand{\ad}{\delta}
\newcommand{\weq}[1]{\equiv_{#1}} 

\def\TWOEXPTIME{{\sc 2Exptime}}

\def\code{\text{{\sffamily code}}\xspace}
\def\Traces{\text{{\sffamily Traces}}\xspace}

\def\B{{\cal B}}
\def\N{{\mathbb{N}}}
\newcommand{\PSAA}{{\mathsf{PSAA}}}
\newcommand \tpl[1]{\langle #1 \rangle}
\newcommand{\true}{\texttt{true}}
\newcommand{\false}{\texttt{false}}

\newcommand{\atomc}{r}

\newcommand{\qis}{\overline{\exists}}
\newcommand{\qall}{\overline{\forall}}

\def\LINAEXPTIME{{\sffamily AEXP$_{\text{\sffamily pol}}$}}

\newcommand{\Union}{\bigcup}
\newcommand{\agentc}{c}

\begin{document}

\title{Refinement Modal Logic}

\author{Laura Bozzelli\thanks{Inform\'atica, Universidad Polit\'ecnica de Madrid, Spain, laura.bozzelli@fi.upm.es}, Hans van Ditmarsch\thanks{LORIA, CNRS -- Universit\'e de Lorraine, France, hans.van-ditmarsch@loria.fr}, Tim French\thanks{Computer Science and Software Engineering, University of Western Australia, tim@csse.uwa.edu.au}, James Hales\thanks{Computer Science and Software Engineering, University of Western Australia, james@csse.uwa.edu.au}, Sophie Pinchinat\thanks{IRISA, University of Rennes, Sophie.Pinchinat@irisa.fr}}

\date{\today}

\maketitle

\begin{abstract} 
In this paper we present {\em refinement modal logic}. A refinement is like a bisimulation, except that from the three relational requirements only `atoms' and `back' need to be satisfied. Our logic contains a new operator $\G$ in addition to the standard modalities $\Box$ for each agent. The operator $\G$ acts as a quantifier over the set of all refinements of a given model. As a variation on a bisimulation quantifier, this refinement operator or {\em refinement quantifier} $\G$ can be seen as quantifying over a variable not occurring in the formula bound by it. The logic combines the simplicity of multi-agent modal logic with some powers of monadic second-order quantification. We present a sound and complete axiomatization of multi-agent refinement modal logic. We also present an extension of the logic to the modal $\mu$-calculus, and an axiomatization for the single-agent version of this logic. Examples and applications are also discussed: to software verification and design (the set of agents can also be seen as a set of actions), and to dynamic epistemic logic. We further give detailed results on the complexity of satisfiability, and on succinctness.

\end{abstract}


\tableofcontents

\section{Introduction}

Modal logic is frequently used for modelling knowledge in multi-agent
systems. The semantics of modal logic uses the notion of ``possible
worlds'', between which an agent is unable to distinguish. In dynamic
systems agents acquire new knowledge (say by an announcement, or the
execution of some action) that allows agents to distinguish between worlds
that they previously could not separate. From the agent's point of view,
what were ``possible worlds'' become inconceivable. Thus, a future
informative event may be modelled by a reduction in the agent's
accessibility relation.  In \cite{hvdetal.loft:2009} the {\em future
  event logic} is introduced. It augments the multi-agent logic of
knowledge with an operation $\G \phi$ that stands for ``$\phi$
holds after all informative events'' --- the diamond version $\F \phi$
stands for ``there is an informative event after which $\phi$.'' The proposal was a generalization of a so-called arbitrary public announcement logic with an operator for ``$\phi$ holds after all announcements'' \cite{balbianietal:2008}. The semantics of informative events encompasses action model
execution \`a la Baltag et al.\ \cite{baltagetal:1998}: on finite models,
it can be easily shown that a model resulting from action model
execution is a refinement of the initial model, and for a given
refinement of a model we can construct an action model such that the result of its execution is bisimilar to that refinement. In \cite{hvdetal.felax:2010} an axiomatization of the single-agent version of this logic is presented, and also expressivity and complexity results. These questions were visited in both the context of modal logic, and of the modal $\mu$-calculus.

In the original motivation, the main operator $\F$ had a rather temporal sense --- therefore the `future event' name. However, we have come to realize that the structural transformation that interprets this operator is of much more general use, on many very different kinds of modal logic, namely anywhere where more than a mere model restriction or pruning is required. We have therefore come to call this the refinement operator, and the logic refinement modal logic.

Thus we may consider {\em refinement modal logic} to be a more abstract perspective of future event logic \cite{hvdetal.loft:2009} applicable to other modal logics. To {\em any} other modal logic!
This is significant in that it motivates the application of the new
operator in many different settings. In logics for games \cite{pauly:2001,alur98} or in control theory
\cite{ramadge89,tsitsiklis89}, it may correspond to a player
discarding some moves; for program logics \cite{hareletal:2000} it may
correspond to operational refinement \cite{Morgan94}; and for logics for spatial reasoning it may correspond to sub-space projections \cite{parikhetal:2007}.

\begin{quote}
{\em 
Let us give an example. Consider the following structure. The $\circ$ state is the designated point. The arrows can be associated with a modality.

\psset{border=2pt, nodesep=4pt, radius=2pt, tnpos=a}
\pspicture(0,-.5)(9,.5)
$
\rput(4.5,0){\rnode{3}{\circ}}
\rput(6,0){\rnode{4}{\bullet}}
\rput(7.5,0){\rnode{5}{\bullet}}
\rput(9,0){\rnode{6}{\bullet}}
\ncline{->}{3}{4}
\ncline{->}{4}{5}
\ncline{->}{5}{6}
$
\endpspicture

\noindent E.g., $\Dia\Dia\Dia\Box\bot$ is true in the point. From the point of view of the modal language, this structure is essentially the same structure (it is bisimilar) as

\psset{border=2pt, nodesep=4pt, radius=2pt, tnpos=a}
\pspicture(0,-.5)(9,.5)
$
\rput(0,0){\rnode{0}{\bullet}}
\rput(1.5,0){\rnode{1}{\bullet}}
\rput(3,0){\rnode{2}{\bullet}}
\rput(4.5,0){\rnode{3}{\circ}}
\rput(6,0){\rnode{4}{\bullet}}
\rput(7.5,0){\rnode{5}{\bullet}}
\rput(9,0){\rnode{6}{\bullet}}
\ncline{<-}{0}{1}
\ncline{<-}{1}{2}
\ncline{<-}{2}{3}
\ncline{->}{3}{4}
\ncline{->}{4}{5}
\ncline{->}{5}{6}
$
\endpspicture

\noindent This one also satisfies $\Dia\Dia\Dia\Box\bot$ and any other modal formula for that matter. A more radical structural transformation would be to consider submodels, such as

\psset{border=2pt, nodesep=4pt, radius=2pt, tnpos=a}
\pspicture(0,-.5)(9,.5)
$
\rput(4.5,0){\rnode{3}{\circ}}
\rput(6,0){\rnode{4}{\bullet}}
\rput(7.5,0){\rnode{5}{\bullet}}
\ncline{->}{3}{4}
\ncline{->}{4}{5}
$
\endpspicture

A distinguishing formula between the two is $\Dia\Dia\Box\bot$, which is true here and false above. Can we consider other `submodel-like' transformations that are {\em neither} bisimilar structures {\em nor} strict submodels? Yes, we can. Consider

\psset{border=2pt, nodesep=4pt, radius=2pt, tnpos=a}
\pspicture(0,-.5)(9,.5)
$
\rput(3,0){\rnode{2}{\bullet}}
\rput(4.5,0){\rnode{3}{\circ}}
\rput(6,0){\rnode{4}{\bullet}}
\rput(7.5,0){\rnode{5}{\bullet}}
\ncline{<-}{2}{3}
\ncline{->}{3}{4}
\ncline{->}{4}{5}
$
\endpspicture

It is neither a submodel of the initial structure, nor is it bisimilar. It satisfies the formula $\Dia\Dia\Box\bot\et\Dia\Dia\Dia\Box\bot$ that certainly is false in any submodel. We call this structure a refinement (or `a refinement of the initial structure'), and the original structure a simulation of the latter. Now note that if we consider the three requirements `atoms', `forth', and `back' of a bisimulation, that `atoms' and `back' are satisfied but not `forth', e.g., from the length-three path in the original structure the last arrow has no image. There seems to be still some `submodel-like' relation with the original structure. Look at its bisimilar duplicate (the one with seven states). The last structure is a submodel of that copy. Such a relation always holds: a refinement of a given structure can always be seen as the model restriction of a bisimilar copy of the given structure. This work deals with the semantic operation of refinement, as in this example, in full generality, and also applied to the multi-agent case.}
\end{quote}

Previous works \cite{faginetal:1995,lomuscioetal:1998b} employed a notion of refinement. In
\cite{lomuscioetal:1998b} it was shown that model restrictions were
not sufficient to simulate informative events, and they introduced
{\em refinement trees} for this purpose --- a precursor of the dynamic epistemic logics developed later (for an overview, see \cite{hvdetal.del:2007}). This usage of refinement as a more general operation than model restriction is similar to ours.


In formal methods literature, see e.g.\ \cite{woodcocketal:1996}, {\em refinement of datatypes} is considered such that (datatype) $C$ refines $A$ if $A$ simulates $C$. This usage of refinement as the converse of simulation \cite{aczel:1988,blackburnetal:2001} comes close to ours --- in fact, it inspired us to propose a similar notion, although the correspondence is otherwise not very close. A similar usage of refinement as in \cite{woodcocketal:1996} is found in \cite{alur98alternating,AHLNW-08}. In the theory of modal specifications a refinement preorder is used, known as \emph{modal refinement} \cite{Raclet2007b,ryanetal:2002}. Modal specifications  are deterministic automata equipped with {\emph{may}}-transitions and {\emph{must}}-transitions. A must-transition is available in every component that implements the modal specification, while a may-transition need not be.  This is close to
our definition of refinement, as it also is some kind of submodel quantifier, but the two notions are incomparable, because `must' is a subtype of `may'.

We incorporate implicit quantification over informative events directly into the language using, again, a notion of {\em refinement}; also in our case a refinement is the converse of simulation. Our work is closely related to some recent work on bisimulation quantified modal logics
\cite{dagostinoetal:2008,french:2006}. The refinement operator, seen as refinement quantifier, is weaker than a bisimulation quantifier \cite{hvdetal.loft:2009}, as it is {\em only} based on simulations rather than bisimulations, and as it {\em only} allows us to vary the interpretation of a propositional variable that does not occur in the formula bound by it. Bisimulation quantified modal logic has
previously been axiomatized by providing a provably correct
translation to the modal $\mu$-calculus \cite{dagostinoetal:2005}. This is reputedly a very complicated one. The axiomatization for the refinement operator, in stark contrast, is quite simple and elegant.

\paragraph*{Overview of the paper} 

Section \ref{sec.tech} gives a wide overview of our technical apparatus: modal logic, cover logic, modal $\mu$-calculus, and bisimulation quantified logic.  Section \ref{sec.refinement} introduces the semantic operation of refinement. This includes a game and (modal) logical characterization. Then, in Section \ref{sec-synt-sem}, we introduce two logics with a refinement quantifier that is interpreted with the refinement relation: refinement modal logic and refinement $\mu$-calculus. Section \ref{sec-FEL} contains the axiomatization of that refinement modal logic and the completeness proof. We demonstrate that it is equally expressive as modal logic. We mention results for model classes ${\mathcal KD}45$ and ${\mathcal S}5$. Section \ref{sec-FEL-mu} gives the axiomatization of refinement $\mu$-calculus. Again, we have a reduction here, to standard $\mu$-calculus. In Section \ref{sec-complexity} we show that, although the use of refinement quantification does not change the expressive power of the logics, they do make each logic exponentially more succinct. We give a non-elementary complexity bound for refinement modal $\mu$-calculus.


\section{Technical preliminaries} \label{sec.tech}

Throughout the paper we assume a finite set of agents $\Agents$ and a countable set of propositional variables $\Atoms$ as background parameters when defining the structures and the logics. Agents are named $\agent,\agentb,\agent',\agentb',\dots$, and propositional variables are $\atom,\atomb,\atomc,\atom',\atom'',\atom_1,\atom_2,\dots$. Agent $\agent$ is assumed female, and $\agentb$ male.

\paragraph{Structures}
A {\em model} $M = (\States, R, V)$ consists of a {\em domain} $\States$ of (factual) {\em states} (or {\em worlds}), an {\em accessibility} function $R : \Agents \imp {\mathcal P}(\States \times \States)$, and a {\em valuation} $V: \Atoms \imp {\mathcal P}(\States)$. States are $\state,\stateb,\statec,\stated,\state',\dots,\state_1,\dots$ A pair consisting of a model $M$ (with domain $\States$) and a state $\state\in\States$ is called a {\em pointed model}, for which we write $M_\state$. 
For $R(\agent)$ we write $R_\agent$; accessibility function $R$ can be seen as a set of accessibility relations $R_\agent$, and $V$ as a set of valuations $V(\atom)$. Given two states $\state,\state'$ in the
domain, $R_\agent(\state,\state')$ means that in state $\state$ agent
$\agent$ considers $\state'$ a possibility. We will also use a relation $R_\agent$ simply as a set of pairs $\subseteq \States \times \States$, and use the abbreviation $\state R_\agent = \{\stateb\in\States \mid (\state,\stateb)\in R_\agent\}$. As we will be often required to discuss several models at once, we will use the convention that $M = (\States^M, R^M, V^M)$, $N = (\States^N, R^N, V^N)$, etc. The class of all models (given parameter sets of agents $\Agents$ and propositional variables $\Atoms$) is denoted ${\mathcal K}$. The class of all models where for all agents the accessibility relation is reflexive, transitive and symmetric is denoted ${\mathcal S}5$, and the model class with a serial, transitive and euclidean accessibility relation is denoted ${\mathcal KD}45$. 

The {\em restriction} $M'$ of a model $M$, notation $M' \subseteq M$, is a model $M' = (\States', R', V')$ such that $\States'\subseteq\States$, for each $\agent\in\Agents$, $R'_\agent = R_\agent \cap (\States' \times \States')$, and for each $\atom\in\Atoms$, $V'(\atom) = V(\atom) \cap \States'$.

\paragraph{Multi-agent modal logic}

 
The {\em language $\lang$ of multi-agent modal logic} is inductively defined as
\[ \begin{array}{l}
\phi ::= \atom \ | \ \neg \phi \ | \ (\phi \et \phi) \ | \ \know_\agent \phi \end{array} \]
where $\agent\in\Agents$ and $\atom\in\Atoms$. Without the construct $\know_\agent \phi$ we get the {\em language $\lang_0$ of propositional logic}. 
Standard abbreviations are: 
$\phi \vel \psi$ iff $\neg(\neg\phi\et\neg\psi)$, 
$\phi\imp\psi$ iff $\neg\phi\vel \psi$, $\T$ iff $\atom\vel\neg\atom$, $\bot$ iff $\atom\et\neg\atom$, and $\M_\agent\phi$ iff $\neg \know_\agent \neg \phi$. If there is a single agent only ($|\Agents| = 1$), we may write $\know \phi$ instead of $\know_\agent \phi$. \weg{, and the language is then called $\lang^1$ (similarly for various extensions of our languages, in later sections).} Formula variables are $\phi,\psi,\chi,\phi',\dots,\phi_1,\dots$ and for sets of formulas we write $\Phi,\Psi,\dots$ For a finite set $\Phi$ of $\lang$ formulas we let the {\em cover} operator $\covers_\agent \Phi$ be an abbreviation for $\know_\agent\bigvee_{\phi\in\Phi}\phi\land\bigwedge_{\phi\in\Phi}\susp_\agent \phi$; we note $\bigvee_{\phi\in \emptyset}\phi$ is always false, whilst $\bigwedge_{\phi\in\emptyset} \phi$ is always true. 

Let a finite set of formulas $\Psi = \{\psi_1, \dots, \psi_n \}$ and a formula $\phi$ with possible occurrences of a propositional variable $\atom$ be given. Let $\phi[\psi\backslash\atom]$ denote the substitution of all occurrences of $\atom$ in $\phi$ by $\psi$. Then $\phi[\Psi\backslash\atom]$ abbreviates $\{\phi[\psi_1\backslash\atom], \dots, \phi[\psi_n\backslash\atom] \}$, and similarly $\Vel \phi[\Psi\backslash\atom]$ stands for $\phi[\psi_1\backslash\atom] \vel \dots \vel \phi[\psi_n\backslash\atom]$ and $\Et \phi[\Psi\backslash\atom]$ stands for $\phi[\psi_1\backslash\atom] \et \dots \et \phi[\psi_n\backslash\atom]$. For example, $\Diamond_\agent \Phi$ abbreviates
$\{\Diamond_\agent \phi \mid \phi \in \Phi\}$, and the definition of $\covers_\agent \Phi$, above, is then written as $\know_\agent\bigvee\Phi\land\bigwedge\susp_\agent \Phi$.


We now define the semantics of modal logic.
Assume a model $M = (\States, R, V)$. The interpretation of $\phi\in\lang$ is defined by induction.
\[ \begin{array}{l}
M_\state \models \atom \ \mbox{iff} \  \state \in V_\atom \\

M_\state \models \neg \phi \ \mbox{iff} \  M_\state \not \models \phi \\
M_\state \models \phi \et \psi \ \mbox{iff}  \  M_\state \models \phi \text{ and } M_\state \models \psi \\

M_\state \models \know_\agent \phi \ \mbox{iff} \  \text{for all } \stateb \in \States: (\state, \stateb) \in R_\agent \text{ implies } M_\stateb \models \phi
\end{array} \]
A formula $\phi$ is valid on a model $M$, notation $M \models \phi$, iff for all $\state\in\States$, $M_\state \models \phi$; and $\phi$ is valid iff $\phi$ is valid on all $M$ (in the model class ${\mathcal K}$, given agents $\Agents$ and basic propositions $\Atoms$). The set of validities, i.e., the {\em logic} in the stricter sense of the word, is called $\logicK$.

\paragraph{Cover logic} \label{sec.coverlogic}

The cover operator $\covers$ has also been used as a syntactic primitive in modal logics \cite{dagostinoetal:2005}. It has recently been axiomatized \cite{bilkovaetal:2008}. 
The language $\lang_\covers$ of cover logic is defined as $$ \phi \ ::= \ \atom \mid \neg \phi \mid (\phi \et \phi) \mid \covers_\agent \{ \phi, \dots, \phi \}, $$ where $\atom \in\Atoms$, and $\agent \in \Agents$.
The semantics of $\covers_\agent \Phi$ is the obvious one if we recall
our introduction by abbreviation of the cover operator: 
\begin{quote}
$M_s \models \covers_\agent \Phi$ iff for all $\phi\in\Phi$ there is a $t \in sR_\agent$ such that $M_t \models \phi$, and for all $t \in sR_\agent$ there is a $\phi\in\Phi$ such that $M_t \models \phi$. 
\end{quote}
The set of validities of cover logic is called ${\mathsf K}_\covers$. The conjunction of two cover formulas is again equivalent to a cover formula: $$\covers_\agent \Phi \et \covers_\agent \Psi \ \ \Eq \ \  \covers_\agent ((\Phi\et\Vel\Psi)\union (\Psi\et\Vel\Phi)) \ . $$ The modal box and diamond are definable as $\Box_\agent \phi$ iff $\covers_\agent \emptyset \vel \covers_\agent \{ \phi \}$, and $\Diamond_\agent \phi$ iff $\covers_\agent \{ \phi, \T \}$, respectively. Cover logic ${\mathsf K}_\covers$ is equally expressive as modal logic ${\mathsf K}$ (also in the multi-agent version) \cite{bilkovaetal:2008,kupkeetal:2008}. We use cover operators in the presentation of the axioms.

\paragraph{Modal $\mu$-calculus}

For the modal $\mu$-calculus, apart from the set of propositional variables $\Atoms$ we have another parameter set $X$ of variables to be used in the fixed-point construction. The language $\lang^\mu$ of modal $\mu$-calculus is defined as follows.
\[ \begin{array}{l} \phi ::= x \mid \atom \ | \ \neg \phi \ | \ (\phi \et \phi) \ | \ \know_\agent \phi \ | \ \mu x.\phi \end{array} \]
where $\agent\in\Agents$, $x \in X$, $\atom\in\Atoms$,
and where in $\mu x.\phi$ the variable $x$ only occurs positively (i.e. in the scope of an even number of negations) in the formula $\phi$.
We will refer to a variable $x$ in an expression $\mu x.\phi$ as a {\em fixed-point variable}. The formula $\nu x.\phi$ is an abbreviation for $\neg\mu x.\neg\phi[\neg x\backslash x]$. Here, we extend the notion of substitution to modal $\mu$-calculus by ruling out the substitution of bound variables, i.e., to give the crucial clauses: $(\mu x.\phi)[\psi\backslash x] = \mu x.\phi$ whereas $(\mu x.\phi)[\psi\backslash y] = \mu x. \phi[\psi\backslash y]$.

For the semantics of the $\mu$-calculus, the valuation $V$ of propositional variables is extended to include fixed-point variables. We write $V^{[x \mapsto T]}$ for the operation that changes a given valuation $V$ into one wherein $V(x) = T$ (where $T \subseteq \States$) and the valuation of all other fixed-point and propositional variables remains the same. Given a model $M = (S,R,V)$, we similarly write $M^{[x \mapsto T]}$ for the model $M = (S,R,V^{[x \mapsto T]})$. The semantics of $\mu x.\phi$ (the top-down presentation, not the bottom-up presentation) is now as follows: 
Let $\phi \in \lang^\mu$ and model $M$ be given. \[ M_s \models \mu x.\phi \text{ iff } s \in \bigcap \{ T \subseteq S \mid \{ u \mid M_u^{[x \mapsto T]} \models \phi \} \subseteq T \} \] 

\paragraph{Disjunctive formula}

An important technical definition we require later on is that of a disjunctive formula.
A {\em disjunctive $\lang^\mu$ formula} is specified by the following abstract syntax:
\begin{equation} \label{def:df} 
\phi \ ::= \ x \ | \ (\phi \lor \phi) \ | \ (\phi_0 \et \Et_{\agent\in\Group} \covers_\agent \{\phi, \dots, \phi\}) \ | \ \mu x.\phi \ |\ \nu x.\phi  
\end{equation}
where $x \in X$, $\phi_0 \in \lang_0$ (propositional logic), and $\Group\subseteq\Agents$. To get the {\em disjunctive $\lang$ formula} (of modal logic) we omit the clauses containing $\mu$-calculus variables $x$:
\[ \begin{array}{lcl} 
\phi & ::= & (\phi \lor \phi) \ | \ (\phi_0 \et \Et_{\agent\in\Group} \covers_\agent \{\phi, \dots, \phi\}). \end{array} \]
If the context of the logic is clear, we simply write {\em disjunctive formula} (or {\em df}).
If $\Group=\emptyset$, we have that $\Et_{\agent\in\Group} \covers_\agent \{\phi_1, \dots, \phi_n\} = \top$, as expected.
\begin{subequations} 
\begin{equation}
\text{Every $\lang^\mu$ formula is equivalent to a disjunctive $\lang^\mu$ formula \cite{janin96}.} \label{prop:dnf}
\end{equation}
\begin{equation}
\text{Every $\lang$ formula is equivalent to a disjunctive $\lang$ formula \cite{venema:2012}.} \label{prop.yde}
\end{equation}
\end{subequations}

\paragraph{Bisimulation quantified modal logic}

The language $\lang_\bqall$ is defined as
\[ \begin{array}{l} \phi ::= \atom \ | \ \neg \phi \ | \ (\phi \et \phi) \ | \ \know_\agent \phi \ | \ \bqall \atom \phi \end{array} \]
where $\agent\in\Agents$ and $\atom\in\Atoms$.
We let $\bqis \atom\phi$ abbreviate $\neg \bqall \atom\neg\phi$. We write $\bqall$ and $\bqis$ for the bisimulation quantifiers in order to distinguish them from the {\em refinement} quantifiers $\forall$ and $\exists$, to be introduced later. Given an atom $\atom$ and a formula $\phi$, the expression $\bqis \atom\phi$ means that there exists a denotation of propositional variable $\atom$ such that $\phi$. It is interpreted as follows (restricted bisimulation $\bisim^\atom$ is introduced further below in Definition \ref{def.bisim}):
\[ M_s\models \bqall \atom\phi \text{ iff for all } N_t \text{ such that } N_t \bisim^\atom M_s, N_t\models\phi \] 
In \cite[Lemma 2.43]{french:2006} a bisimulation quantifier characterization of fixed points is given (the details of which are deferred to Section \ref{sec-FEL-mu} on refinement $\mu$-calculus, where they are pertinent), and from \cite{dagostinoetal:2000} we know that bisimulation quantifiers are also expressible in the modal $\mu$-calculus. For more information on the modal $\mu$-calculus, see \cite{dagostinoetal:2005,venema:2012}.

\section{Refinement} \label{sec.refinement}

In this section we define the notion of structural refinement, investigate its properties, give a game characterization in (basic) modal logic, and compare refinement to bisimulation and other established semantic notions in the literature.

\subsection{Refinement and its basic properties}

\begin{definition}[Bisimulation, simulation, refinement] \label{def.bisim}
Let two models $M= (\States, R, V)$ and $M'= (\States', R', V')$ be given. A non-empty relation $\bisrel \subseteq \States
\times \States'$ is a bisimulation if for all $(\state,\state') \in
\bisrel$ and $\agent\in\Agents$:
\begin{description}
\item[atoms] $\state \in V(p)$ iff $\state' \in V'(p)$ for all $p \in \Atoms$;
\item[forth-$\agent$] for all $\stateb \in \States$, if
$R_\agent(\state,\stateb)$, then there is a $\stateb'\in
\States'$ such that $R'_\agent(\state',\stateb')$ and
$(\stateb,\stateb') \in \bisrel$;
\item[back-$\agent$] for all $\stateb' \in \States'$,
if $R'_\agent(\state',\stateb')$, then there is a
$\stateb \in \States$ such that $R_\agent(\state,\stateb)$
and $(\stateb,\stateb') \in \bisrel$.
\end{description}
We write $M \bisim M'$ ($M$ and $M'$ are bisimilar) iff there is a bisimulation between $M$ and $M'$, and we write $M_\state \bisim M'_{\state'}$ ($M_\state$ and $M'_{\state'}$ are bisimilar) iff there is a bisimulation between $M$ and $M'$ linking $\state$ and $\state'$. A {\em restricted bisimulation} ${\mathfrak R}^\atom: M_\state \bisim^\atom M'_{\state'}$ is a bisimulation that satisfies atoms for all variables except $\atom$. A {\em total bisimulation} is a bisimulation such that all states in the domain and codomain occur in a pair of the relation.

A relation ${\mathfrak R}_\Group$ that satisfies {\bf atoms},  {\bf back-$\agent$}, and {\bf forth-$\agent$} for every $\agent\in\Agents\setminus\Group$, and that satisfies {\bf atoms}, and {\bf back-$\agentb$} for every $\agentb\in\Group$, is a $\Group$-{\em refinement}, we say that $M'_{\state'}$ {\em refines} $M_\state$ for group of agents $\Group$, and we write $M_\state \lumis_\Group M'_{\state'}$.\footnote{We will overload the meaning of refinement and also say that $M'_{\state'}$ is a {\em refinement} of $M_\state$} An $\Agents$-{\em refinement} we call a {\em refinement} (plain and simple) and for $\{\agent\}$-{\em refinement} we write $\agent$-{\em refinement}. 

Dually to refinement, we similarly define $\Group$-{\em simulation} ${\mathfrak R}_\Group$. I.e., a relation ${\mathfrak R}_\Group$ that satisfies {\bf atoms},  {\bf back-$\agent$}, and {\bf forth-$\agent$} for every $\agent\in\Agents\setminus\Group$, and that satisfies {\bf atoms}, and {\bf forth-$\agentb$} for every $\agentb\in\Group$, is a $\Group$-{\em simulation}. 

Restricted refinement and restricted simulation are defined similarly to restricted bisimulation.
\end{definition}

The definition of simulation varies slightly from the one given by Blackburn {\em et al.} \cite[p.110]{blackburnetal:2001}. Here we ensure that simulations (and refinements) preserve the interpretations (i.e., the truth and falsity) of atoms, whereas \cite{blackburnetal:2001} has them only preserve the truth of propositional variables in a simulation---and presumably preserve their falsity in a refinement. We prefer to preserve the entire interpretation, as we feel it suits our applications better. For example, in the case where refinement represents information change, we would not wish basic facts to become false in the process. The changes are supposed to be merely of information, and not factual. Another, inessential, difference with \cite{blackburnetal:2001} is that in their case {\bf atoms} and {\bf forth} are required for all modalities (in the similarity type), i.e., they consider ${\mathfrak R}_\Group$ for $\Group = \Agents$ only. 

If ${\mathfrak R}_\Group: M_\state \lumis_\Group M'_{\state'}$ is a $\Group$-refinement, then the converse relation ${\mathfrak R}^-_\Group ::= \{ (\state,\state') \mid (\state',\state) \in {\mathfrak R}_\Group \}$ is a $\Group$-simulation, and if $M'_{\state'}$ refines $M_\state$ then we can also say that $M_\state$ simulates $M'_{\state'}$.

In an epistemic setting a refinement corresponds to the diminishing uncertainty of agents. This means that there is a potential decrease in the number of states and transitions in a model. On the other hand, the number of states as a consequence of refinement may also increase, because the uncertainty of agents over the extent of decreased uncertainty in other agents may still increase. This is perhaps contrary to the concept of program refinement \cite{Morgan94} where detail is added to a specification. However, in program refinement the added detail requires a more detailed state space (i.e., extra atoms) and as such is more the domain of bisimulation quantifiers, rather than refinement quantification. Still, the consequence of program refinement is a more deterministic system which agrees with the notion of diminishing uncertainty.  

\begin{proposition}\label{lem:lumis}
The relation $\lumis_\agent$ is reflexive and transitive (a pre-order), and satisfies the Church-Rosser property.
\end{proposition}
\begin{proof}
Reflexivity follows from the observation that the identity relation satisfies {\bf atoms}, and {\bf back-$\agent$} and {\bf forth-$\agent$} for all agents $\agent$, and therefore also the weaker requirement for refinement. Similarly, given two $\agent$-refinements $\bisrel_1$, and $\bisrel_2$, we can see that their composition, $\{(x,z)\ |\ \text{there is a } y \text{ for which } (x,y)\in \bisrel_1,\ (y,z)\in \bisrel_2\}$ is also an $\agent$-refinement. This is sufficient to demonstrate transitivity. The Church-Rosser property states that if $N_t \lumis_\agent M_s$ and $N_t \lumis_\agent M'_{s'}$, then there is some model $N'_{t'}$ such that $M_s \lumis_\agent N'_{t'}$ and $M'_{s'} \lumis_\agent N'_{t'}$. From Definition~\ref{def.bisim} it follows that $M_s$ and $M'_{s'}$ must be bisimilar to one another with respect to $\Agents-\{\agent\}$.  We may therefore construct such a model $N'_{t'}$ by taking $M_s$ (or $M'_{s'}$) and setting $R_\agent^{N'} = \emptyset$ and $R_\agentb^{N'} = R_\agentb^M$ for all $\agentb\in\Agents-\{\agent\}$. It can be seen that $N'_{t'}$, where $N' = (\States^M,R^{N'}, V^M)$ and $t'=s$, satisfies the required properties.
\end{proof}

An elementary result is the following.
\begin{proposition}\label{lem:compost}
Let $\Group = \{\agent_1,...,\agent_n\}$, and let $M_s$ and $M_t$ be given. Then $M_s (\lumis_{\agent_1} \circ \dots \circ \lumis_{\agent_n}) M_t$ iff $M_s \lumis_\Group M_t$.
\end{proposition}

\begin{example} \label{exxx}
If $N_t \lumis_\agent M_s$ and $M_s \lumis_\agent N_t$, it is not necessarily the case that $M_s\bisim_\agent N_t$. For example, consider the one-agent models $M$ and $N$ where:
\begin{itemize}
\item $S^M = \{1,2,3\}$, $R_\agent^M = \{(1,2),(2,3)\}$ and $V^M(p) = \emptyset$ for all $p\in\Atoms$; and
\item $S^N = \{4,5,6,7\}$, $R_\agent^N = \{(4,5),(5,6),(4,7)\}$ and $V^M(p) = \emptyset$ for all $p\in\Atoms$.
\end{itemize} 
These two models are clearly not bisimilar, although $N_4\lumis_\agent M_1$ via $\{(4,1),(5,2),(6,3)\}$ and $M_1\lumis_\agent N_4$ via $\{(1,4),(2,5),(3,6),(2,7)\}$. See Figure \ref{fig.reffer}.
\end{example}

\begin{figure}[h]
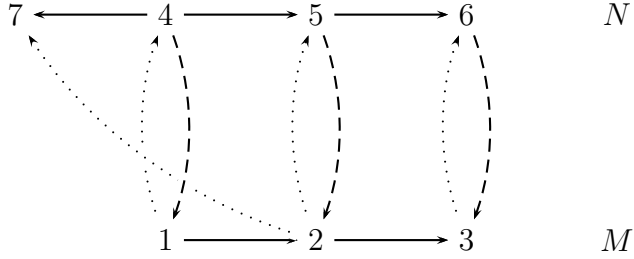

\psset{border=2pt, nodesep=4pt, radius=2pt, tnpos=a}
\pspicture(-.2,-.2)(6.2,3.2)
\rput(2,0){\rnode{1}{1}}
\rput(4,0){\rnode{2}{2}}
\rput(6,0){\rnode{3}{3}}
\rput(0,3){\rnode{7}{7}}
\rput(2,3){\rnode{4}{4}}
\rput(4,3){\rnode{5}{5}}
\rput(6,3){\rnode{6}{6}}
\rput(8,3){\rnode{n}{$N$}}
\rput(8,0){\rnode{m}{$M$}}
\ncline{->}{1}{2}
\ncline{->}{2}{3}
\ncline{->}{4}{5}
\ncline{->}{5}{6}
\ncline{->}{4}{7}
\ncarc[arcangle=20,linestyle=dotted]{->}{1}{4}
\ncarc[arcangle=20,linestyle=dotted]{->}{2}{5}
\ncarc[arcangle=20,linestyle=dotted]{->}{3}{6}
\ncarc[arcangle=20,linestyle=dotted]{->}{2}{7}
\ncarc[arcangle=20,linestyle=dashed]{->}{4}{1} 
\ncarc[arcangle=20,linestyle=dashed]{->}{5}{2} 
\ncarc[arcangle=20,linestyle=dashed]{->}{6}{3}
\endpspicture
\caption{Refinement and simulation, but no bisimulation}
\label{fig.reffer}
\end{figure}

Given that the equivalence $M_s \equiv N_t$ defined by $M_s \simul N_t$ and $M_s \lumis N_t$ is not a bisimulation, an interesting question seems to be {\em what} it then represents. It seems to formalize that two structures are only different in resolvable differences in uncertainty (for the agent of the refinement), but not in hard and necessary facts. So the positive formulas (for that agent) should be preserved under this `equivalence' $\equiv$. Such matters will now be addressed. 

\subsection{Game and logical characterization of refinement}
\label{sec-logicalcharacterization-of-ref}

It is folklore to associate a (infinite duration) two-player game with refinement, in the spirit of \cite{alur98alternating}. 

\begin{definition}[Refinement game]
\label{def-refinementgame}
Let $M_s$ and $N_t$ be two models. We define a turn-based game
$\calG{a}(M_s,N_t)$ between two players \Spoiler\ and \Duplicator\ (male and female, respectively) by
$\calG{a}(M_s, N_t)=(V, E, (s,t))$ where the set of positions $V$ is
partitioned into the positions $V_\Spoiler=S^M \times S^N$ of
\Spoiler\ and the positions $V_\Duplicator = S^M \times
[\{forth,back\}\times (\Agents \union P)] \times S^N$ of \Duplicator. Since the initial position $(s,t) \in V_\Spoiler$, \Spoiler\
starts. The set of moves $E \subseteq V_\Spoiler \times V_\Duplicator
\cup V_\Duplicator \times V_\Spoiler$ is the least set such that the
following pairs belong to $E$ (we take the convention that $\agentb
\neq \agent$, and for convenience, we name those moves with names
similar to the properties of refinement in
Definition~\ref{def.bisim}):
\begin{center}
\scalebox{0.9}{
\begin{tabular}{|c|}
\hline
\Spoiler's moves\\
\begin{tabular}{ll|l}
Move & & Name\\ 
\hline
$((s',t'),(s',(forth,p),t')) $ & whenever $s' \in V^M(p)$ & forth-p?   \\
$((s',t'),(s',(back,p),t')) $ & whenever $t' \in V^N(p)$ & back-p? \\
$((s',t'),(s'',(forth,\agentb),t'))$ & whenever $s'' \in R^M_\agentb(s')$  & forth-b?\\
$((s',t'),(s',(back,\agentb),t''))$ & whenever $t'' \in R^N_\agentb(t')$ & back-b?\\
$((s',t'),(s',(back,\agent),t''))$ & whenever $t'' \in R^N_\agent(t')$ & back-a?
\end{tabular}\\
\hline
\Duplicator's moves\\
\begin{tabular}{ll|l}
Move & & Name\\ 
\hline
$((s',(forth,p),t'),(s',t')) $ & whenever $t' \in V^N(p)$ & forth-p!\\
$((s',(back,p),t'),(s',t')) $ & whenever $s' \in V^M(p)$ & back-p!\\
$((s'',(forth,\agentb),t'),(s'',t''))$ & whenever $t'' \in R^N_\agentb(t')$& forth-b!\\
$((s',(back,\agentb),t''),(s'',t''))$ & whenever $s'' \in R^M_\agentb(s')$ & back-b!\\
$((s',(back,\agent),t''),(s'',t''))$ & whenever $s'' \in R^M_\agent(s')$ & back-a!
\end{tabular}\\
\hline
\end{tabular}}
\end{center}
\end{definition}
In the game $\calG{a}(M_s,N_t)$, a \emph{play} is a maximal (possibly infinite) sequence
of consecutive moves, or equivalently a maximal sequence of adjacentes
positions in the arena. The play is \emph{winning for \Duplicator}
if it is infinite or if it is finite and ends in position of \Spoiler, otherwise, the play ends in a position of \Duplicator\ and it is winning for \Spoiler.

A \emph{strategy of \Duplicator\ (resp.\ \Spoiler)} is a mapping
$\sigma : V^*V_\Duplicator \to V$ (resp.\ $\sigma : V^*V_\Spoiler \to
V$) which recommends which moves to choose after each prefix of a play.

A play is an \emph{outcome} of a strategy for \Duplicator\ (resp.\ \Spoiler)
if each time \Duplicator\ (resp.\ \Spoiler) had to play, she (resp.\
he) has selected the move recommended by her (resp.\ his) strategy.
A strategy is \emph{winning} if all its outcomes are winning.

\begin{remark}
\label{remark-parity}
One easily sees that the refinement game of
Definition~\ref{def-refinementgame} is a particular \emph{parity game}
\cite{mazala2002infinite}. Henceforth, according to
\cite{kusters2002memoryless}, the refinement game is
\emph{determined}\footnote{In each position, either \Duplicator\ or \Spoiler\ has
  a winning strategy from that position.}, and \emph{memoryless}\footnote{Strategies
  $\sigma$ that only take into account the current position in the
  game, instead of the entire prefix of the game that is currently
  played.} strategies suffice.
\end{remark}

\noindent Notice that there is no forth-a move in the game $\calG{a}(M_s,N_t)$, which captures the refinement relation between the structures:
\begin{lemma}
\label{lem-game}
$M_s \lumis_\agent N_t$ iff \Duplicator\ has a winning strategy in $\calG{a}(M_s,N_t)$.
\end{lemma}

\begin{proof}
  Assume \Duplicator\ has a winning strategy $\sigma$ in
  $\calG{a}(M_s,N_t)$. By Remark~\ref{remark-parity} and without loss
  of generality, this winning $\sigma$ can be taken to be memoryless.
  Namely, $\sigma : V_\Duplicator \to V_\Spoiler$. 
  Now, define the binary relation $\refrel_\sigma \subseteq S^M\times S^N$ as the
  set of pairs $(s',t')\in V_\Spoiler$ that are reachable when
  \Duplicator\ follows her strategy $\sigma$.  Then it is easy to
  check that $\refrel_\sigma$ is an $a$-refinement from $M_s$ to $N_t$.
  Also it is not difficult to see that if some $a$-refinement
  $\refrel_\agent$ exists from $M_s$ to $N_t$, then any strategy of
  \Duplicator\ which maintains \Spoiler's positions in
  $\refrel_\agent$, is winning. Note that by Definition~\ref{def.bisim} 
of a refinement, this is always possible for
  her.
\end{proof}

We now consider a characterization of the refinement in terms of
the logic $\lang_\G$. Namely, given an agent $\agent$, we define the
fragment of the $a$-positive formulas $\lang^{a+} \subseteq \lang$ by
\[ \begin{array}{l}
\lang^{a+} \ \ni \ \phi ::= \atom \ | \ \neg \atom \ | \ (\phi \et \phi) \ | \ (\phi \vel \phi) \ | \ \know_\agentb \phi \ | \ \M_\agentb \phi \ | \
\M_\agent \phi \end{array} \]
where $\agentb\in\Agents\setminus\{\agent\}$ and $\atom\in\Atoms$.

\begin{proposition}
\label{prop-logicalcharacterization-of-ref}
For any finitely branching (every state has only finitely many successors) pointed 
models $M_{s_0}$ and $N_{t_0}$, and for any agent $\agent \in \Agents$,
\begin{center}
$M_{s_0} \lumis_\agent N_{t_0} $ if, and only if, for every $\phi \in \lang^{a+}$, $N_{t_0} \models \phi$ implies $M_{s_0} \models \phi$.
\end{center}
\end{proposition}

\begin{proof}  
  Let us first establish that for every $t \in S^N$ and $s \in S^M$,
  if \Spoiler\ has a winning strategy in $\calG{a}(M_s,N_t)$, then
  there exists a formula $\phi(s,t) \in \lang^{a+}$ called a
  \emph{distinguishing formula for $(M_s,N_t)$}, for which $N_t
  \models \phi(s,t)$ but $M_s \not \models \phi(s,t)$. Note that if
  \Spoiler\ has a winning strategy in $\calG{a}(M_s,N_t)$, all plays
  induced by this strategy have finite length and end in a position
  where \Duplicator\ cannot move. Moreover, by a simple application of
  K\"onig's Lemma (as the game graph $\calG{a}(M_s,N_t)$ is finitely
  branching), the length of those plays is bounded.

  We reason by induction on $k$, the maximal length of these plays;
  note that because \Spoiler\ starts, $k>0$.

  If $k=1$, \Spoiler\ has a winning move from $(s,t)$ to some $v\in
  V_\Duplicator$, where \Duplicator\ is blocked. We reason on the form of
  $v$:
\begin{itemize}
\item if $v=(s,(forth,p),t)$ (resp.\ $v=(s,(back,p),t)$), then there is no move
  back to $(s,t)$ because $t \not \in V^N(p)$ (resp.\ $s \not \in
  V^M(p)$). A distinguishing formula is $\neg p$ (resp.\ $p$).
\item if $v=(s',(forth,\agentb),t)$ (resp.\
  $v=(s,(back,\agentb),t')$), then $tR^N_\agentb=\emptyset$ (resp.\
  $sR^M_\agentb=\emptyset$). A distinguishing formula is
  $\know_\agentb \bot$ (resp.\ $\M_\agentb \top$). The case
  $v=(s,(back,a),t')$ is the same as $(s,(back,b),t')$.  Since forth-a
  moves are not allowed in the game, position
  $v=(s',(forth,\agent),t)$ is not reachable in the game
  $\calG{a}(M_s,N_t)$, so that the formula $\know_\agent \bot \not\in
  \lang^{a+}$ is not needed.
\end{itemize}

Assume now that $k>1$, and pick a winning strategy of \Spoiler\ in
$\calG{a}(M_s,N_t)$. 

We explore the move from the initial position $(s,t)$ that is
given by this strategy; because $k>1$, this move cannot be either
{\bf forth-p?}, or {\bf back-p?}. Three cases remain.
\begin{description}
\item[forth-b?] The reached position becomes $(s',(forth,\agentb),t)$,
  and from there \Duplicator\ loses. That is, for each $t' \in
  tR_\agentb^N$, \Spoiler\ wins the game $\calG{a}(M_{s'},N_{t'})$
  in at most $k-2$ steps. By the induction hypothesis, there exists a
  distinguishing formula $\phi(s',t') \in \lang^{a+}$ for
  $(M_{s'},N_{t'})$. It is easy to see that $\phi(s,t)= \know_\agentb
  (\bigvee_{t' \in tR_\agentb^N}\phi(s',t'))$ is a distinguishing
  formula for $(M_s,N_t)$; notice that since $N$ is finitely
  branching, the conjunction is finitary.
\item[back-b?] This case applies to $\agentb \neq \agent$ and to $\agentb = \agent$. 

The reached position
  becomes $(s,(back,\agentb),t')$, and from there \Duplicator\
  loses. Using a similar reasoning as for forth-b moves, it is easy
  to establish that there exists a formula $\phi(s',t') \in
  \lang^{a+}$, such that $\phi(s,t)= \M_\agentb (\bigwedge_{s'\in
    sR_\agentb^M} \phi(s',t'))$ is a distinguishing formula for
  $(M_s,N_t)$; here, as $M$ is finitely branching, a finitary disjunction is guaranteed.
\end{description}

Now, according to the game characterization of refinement
(Lemma~\ref{lem-game}) and the determinacy of the refinement games
(Remark~\ref{remark-parity}), the existence of a winning strategy for
\Spoiler\ from position $(s_0,t_0)$ is equivalent to
$M_{s_0}\not\lumis_\agent N_{t_0}$; this provides us with the right to left
direction of the proposition.  \\

For the other direction, assume $M_s
\lumis_\agent N_t$, and let $\phi \in \lang^{a+}$ with $N_t \models
\phi$. We prove that $M_s \models \phi$, by induction over the
structure of the formula. Basic cases where $\phi$ is either $p$ or
$\neg p$, but also the cases $\phi \et \psi$ and $\phi \vel \psi$, are
immediate.

Assume $N_t \models \know_\agentb \phi$. Then for every $t' \in
  tR_\agentb^N$, $N_{t'} \models \phi$. If
  $tR_\agentb^N=\emptyset$, then by Property {\bf forth-b} of
  Definition~\ref{def.bisim} this entails $sR_\agentb^M=\emptyset$ and
  consequently $M_s \models \know_\agentb \phi$ (whatever $\phi$
  is). Otherwise, $tR_\agentb^N\neq\emptyset$. Take an arbitrary $s'
  \in sR^M_\agentb$. By Property {\bf forth-b} of
  Definition~\ref{def.bisim}, there is a $t'_{s'} \in tR_\agentb^M$ with
  $M_{s'} \lumis_\agentb N_{t'_{s'}}$ and $N_{t'_{s'}} \models \phi$. By
  induction hypothesis, $M_{s'}\models \phi$, which entails $M_s
  \models \know_\agentb \phi$.

Assume $N_t \models \M_\agentb \phi$, and let $t'\in
  tR_\agentb^N$ be such that $N_{t'}\models\phi$. By Property {\bf back-b} of
  Definition~\ref{def.bisim}, there is some $s' \in sR_\agentb^M$,
  such that $M_{s'} \lumis_\agentb N_{t'}$. By induction hypothesis,
  $M_{s'}\models\phi$ which entails $M_s \models \M_\agentb \phi$.

Note that the argument still holds if we take $\agentb=\agent$.
\end{proof}

\subsection{Refinement as bisimulation plus model restriction} \label{sec.modelrestriction}

A bisimulation is also a refinement, but refinement allows much more semantic variation. How much more? There is a precise relation. Semantically, a refinement is a bisimulation followed by a model restriction.

An $\agent$-refinement needs to satisfy {\bf back} for that agent, but not {\bf forth}. Let an (`initial') model and a refinement of that model be given. For the sake of the exposition we assume that the initial model and the refined model are minimal, i.e., they are bisimulation contractions. Now take an arrow (a pair in the accessibility relation) in that initial model. This arrow may be missing in the refined model namely when {\bf forth} is not satisfied for that arrow. On the other hand, any arrow in the refinement should be traceable to an arrow in the initial model -- the {\bf back} condition. There may be {\em several} arrows in the refinement that are traceable to the {\em same} arrow in the initial model, because the states in which such arrows finish may be non-bisimilar. In other words, we can see the refined model as a blowup of the initial model of which bits and pieces are cut off.

\begin{example}
A simple example is as follows. Consider the structure 

\psset{border=2pt, nodesep=4pt, radius=2pt, tnpos=a}
\pspicture(0,-.5)(9,.5)
$
\rput(4.5,0){\rnode{3}{\underline{\bullet_1}}}
\rput(6,0){\rnode{4}{\bullet_2}}
\rput(7.5,0){\rnode{5}{\bullet_3}}
\rput(9,0){\rnode{6}{\bullet_4}}
\ncline{->}{3}{4}
\ncline{->}{4}{5}
\ncline{->}{5}{6}
$
\endpspicture

\noindent and its refinement 

\psset{border=2pt, nodesep=4pt, radius=2pt, tnpos=a}
\pspicture(0,-.5)(9,.5)
$
\rput(3,0){\rnode{2}{\bullet_{b'}}}
\rput(4.5,0){\rnode{3}{\underline{\bullet_a}}}
\rput(6,0){\rnode{4}{\bullet_b}}
\rput(7.5,0){\rnode{5}{\bullet_c}}
\ncline{<-}{2}{3}
\ncline{->}{3}{4}
\ncline{->}{4}{5}
$
\endpspicture

\noindent by way of refinement relation $\bisrel = \{ (1,a), (2,b), (3,c), (2,b') \}$. The arrow $(3,4)$ has no image in the refined model. On the other hand, the arrow $(1,2)$ has two images, namely $(a,b)$ and $(a,b')$. These two arrows cannot be identified, because $b$ and $b'$ are non-bisimilar, because there is yet another arrow from $b$ but no other arrow from $b'$: arrow $(2,3)$ has only one image in the refined model. 
\end{example}

The cutting off phase can be described such that the relation to restricted bisimulation becomes clear. When expanding the initial model, the blowing up phase, make a certain propositional variable false in all states of the blowup that you want to prune (that are not in the refinement relation) and make it true in all states that you want to keep. Therefore, the blown up model is bisimilar to the initial model except for that variable. (In other words, it is a restricted bisimulation.) Then, remove arrows to states where that atom is false.

\begin{example}
Continuing the previous example, consider the following structure bisimilar to the initial model, except for the value of atom $\atom$---in the visualization $\bullet$ represents that $\atom$ is true and $\circ$ represents that $\atom$ is false. 

\psset{border=2pt, nodesep=4pt, radius=2pt, tnpos=a}
\pspicture(0,-.5)(9,.5)
$
\rput(0,0){\rnode{0}{\circ_{d'}}}
\rput(1.5,0){\rnode{1}{\circ_{c'}}}
\rput(3,0){\rnode{2}{\bullet_{b'}}}
\rput(4.5,0){\rnode{3}{\underline{\bullet_a}}}
\rput(6,0){\rnode{4}{\bullet_b}}
\rput(7.5,0){\rnode{5}{\bullet_c}}
\rput(9,0){\rnode{6}{\circ_d}}
\ncline{<-}{0}{1}
\ncline{<-}{1}{2}
\ncline{<-}{2}{3}
\ncline{->}{3}{4}
\ncline{->}{4}{5}
\ncline{->}{5}{6}
$
\endpspicture

The relation $\bisrel = \{ (1,a), (2,b), (3,c), (4,d), (2,b'), (3,c'), (4,d') \}$ is a bisimulation, except for the value of $\atom$. The refinement from the previous example is a restriction of this structure, namely the result of removing the $\circ$ states and the arrows leading to those states. 
\end{example}

Winding up, performing an $\agent$-refinement clearly corresponds to the following operation: 
\begin{quote} 
Given a pointed model, first choose a bisimilar pointed model, then remove some pairs from the accessibility relation for $\agent$ in that model.
\end{quote} 
Given a propositional variable $\atomb$, this has the same semantic effect as
\begin{quote} 
Given a pointed model, first choose a bisimilar pointed model except for variable $\atomb$, such that $\atomb$ is (only) false in some states that are accessible for $\agent$, then remove all those pairs from the accessibility relation for $\agent$.
\end{quote}
In other words: \begin{quote} Given a pointed model, first choose a bisimilar pointed model except for variable $\atomb$, then remove all pairs from the accessibility relation for $\agent$ pointing to states where $\atomb$ is false.\end{quote}
If we do this for all agents at the same time (or if we strictly regard tree unwindings of models only), we can even see the latter operation as follows:
\begin{quote} Given a pointed model, first choose a bisimilar pointed model except for variable $\atomb$, then restrict the model to the states where $\atomb$ is true.\end{quote} 
Formally, the result is as follows. First, let $M$ be a model with accessibility relation (set of accessibility relations) $R$, and let $R'$ be such that for all $\agent\in\Agents$, $R'_\agent\subseteq R_\agent$, then (analogously to a model restriction) $M | R'$ is the model that is the same as $M$ but with the accessibility {\em restricted} to $R'$.
\begin{lemma} \label{lemma.xx} 
Given $M_s \succeq_\agent N_t$, there is an $N'_t$ (with accessibility function $R'$) and some $R''$ that is the same as $R'$ except that $R''_\agent \subseteq R'_\agent$, such that $M_s \bisim N'_t$ and $N'_t | R'' \bisim N_t$.
\end{lemma}
\begin{proof} 
Let an $\agent$-refinement relation $\bisrel_\agent \subseteq \States^M \times \States^N$ be given, such that $(s,t) \in \bisrel_\agent$. We expand the model $N$ and this relation $\bisrel_\agent$ as follows to a model $N'$ and a bisimulation $\bisrel \subseteq \States^M \times \States^{N'}$. Consider $\States^M_- := \States^M\setminus\bisrel_\agent^{-1}(\States^N)$ ($\States^M_-$ is the set of all states in $M$ that do not have an image in $N$). Now consider $N' = (\States',R',V')$ with domain $\States' = \States^N \union \States^M_-$,  such that for each agent $\agentb$ (including $\agent$), $(u',v') \in R'_\agentb$ iff: \begin{itemize} \item $(u',v') \in R^N_\agentb$, or \item $(u',v') \in R^M_\agentb$, or \item $\agentb=\agent$, $u' \in \States^N$, $v' \in \States^M_-$, there is a $u$ such that $(u,u') \in \bisrel_\agent$, and $(u,v') \in R^M_\agent$;\end{itemize}
and such that $V' = V^N$ on the $\States^N$ part of the domain whereas $V' = V^M$ on the new $\States^M_-$ part of the domain. Now define $\bisrel: \States^M \imp \States'$ as follows: $(u,u') \in \bisrel$ iff $(u,u') \in \bisrel_\agent$ or ($u \in \States^M_-$ and $u=u'$). Then $\bisrel$ is a bisimulation linking $M_s$ and $N'_t$. If we restrict $R'_\agent$ to $R^N_\agent$, we get $N_t$ back (states in the $\States^M_-$ part of $N'$ have become unreachable). We have satisfied the proof requirement that $M_s \bisim N'_t$ and $N'_t | R''_\agent \bisim N_t$ (for $R''_\agent = R^N_\agent$).
\end{proof}

\begin{lemma} \label{lemma.prev}
Given $M_s \succeq_\agent N_t$, there is an $N'_t$ (with accessibility function $R'$) and some $\atom\in\Atoms$ such that $M_s \bisim^\atom N'_t$ and $N'_t | R'' \bisim^\atom N_t$, where $R''$ is the same as $R'$ except that $(\statec,\statec') \in R''_\agent$ iff $N'_{\statec'} \models \atom$.
\end{lemma}
\begin{proof}
To satisfy the requirement for $\atom$, we make $\atom$ false on the $\States^M_-$ part of the domain of $N'$, and true anywhere else on $N'$ (i.e., on the part of $N'$ corresponding to the $\bisrel_\agent^{-1}(\States^N)$ part of $M$). (We do not change the value of other propositional letters on $N'$.)
\end{proof}
Below, $M | \atom$ is the restriction of $M$ to the set of states satisfying $\atom$.
\begin{proposition} \label{lemma.relatag2} 
Given $M_s \succeq_\agent N_t$, there is a $N'_t$ and some $\atom\in\Atoms$ such that $M_s \bisim^\atom N'_t$ and $N'_t | \atom$ is identical to $N_t$ except for maybe the value of $\atom$.
\end{proposition}
\begin{proof} 
Clearly, in Lemma \ref{lemma.prev}, $N'_t | R'' \bisim N'_t | \atom$. The model restriction gets rid of the the $\States^M_-$ part of $N'$, so we now have that $N'_t | \atom$ is identical to (and not merely bisimilar to) $N_t$ except for maybe the value of $\atom$.
\end{proof}

In Section \ref{sec.relativization} we build upon this semantic result by translating the logic with refinement quantifiers into the logic with bisimulation quantifiers plus relativization of formulas.

\subsection{Refinement and action models} \label{subsec.am}

We recall another important result connecting structural refinement to action model execution \cite{baltagetal:1998}. For full details, see \cite{hvdetal.loft:2009}. An {\em action model} $\amodel = (\Actions, \arel, \pre)$ is like a model $M = (\States, R, V)$ but with the valuation replaced by a precondition function $\pre: \Actions \imp \lang$ (for a given language $\lang$). The elements of $\Actions$ are called {\em action points}. A restricted modal product $(M \otimes \amodel)$ consists of pairs $(\state,\action)$ such that $M_\state\models\pre(\action)$, the product of accessibility relations namely such that $((\state,\action),(\stateb,\actionb)) \in R_\agent$ iff $(\state,\stateb) \in R_\agent$ and $(\action,\actionb) \in \arel_\agent$, and keeping the valuation of the state in the pair: $(\state,\action)\in V(\atom)$ iff $\state\in V(\atom)$. A pointed action model $\amodel_\action$ is an {\em epistemic action}. 
\begin{proposition} \label{prop.actionmodel} \cite[Prop.\ 4, 5]{hvdetal.loft:2009} 
The result of executing an epistemic action in a pointed model is a refinement of that model. Dually, for every refinement of a {\em finite} pointed model there is an epistemic action such that the result of its execution in that pointed model is a model bisimilar to the refinement. 
\end{proposition}
It is instructive to outline the proof of these results.

Given pointed model $M_\state$ and epistemic action $\amodel_\action$, the resulting $(M \otimes \amodel)_{(\state,\action)}$ is a refinement of $M_\state$ by way the relation $\bisrel$ consisting of all pairs $(\stateb,(\stateb,\actionb))$ such that $M_\stateb \models \pre(\actionb)$. Some states of the original model may get lost in the modal product, namely if there is no action whose precondition can be executed there. But all `surviving' (state,action)-pairs simply can be traced back to their first argument: clearly a refinement.

For the other direction, construct an epistemic action $\amodel_{\state'}$ that is isomorphic to a given refinement $N_{\state'}$ of a model $M_\state$, but wherein valuations (determining the value of propositional variables) in states $\stateb\in N$ are replaced by preconditions for action execution of the corresponding action points (also called) $\stateb$. Precondition $\pre(\stateb)$ should be satisfied in exactly those states $\state \in M$ such that $(\state,\stateb) \in \bisrel$, where $\bisrel$ is the refinement relation linking $M_\state$ and $N_{\state'}$. Now in a {\em finite} model, we can single out states (up to bisimilarity) by a distinguishing formula \cite{browneetal:1987}. One then shows that $(M  \otimes \amodel, (\state,\state'))$ can be bisimulation-contracted to $N_{\state'}$. It is unknown if the finiteness restriction can be lifted, because the existence of distinguishing formulas plays a crucial part in the proof.

Example \ref{ex-1} presents an action model and its execution in an initial information state, and we will there continue our reflections on the comparison of the frameworks. 

\weg{
\subsection{Refinement and pruning}

Just as refinement is not mere model restriction, it is also immediate to see that refinement is not mere pruning: consider
a model $M$ consisting of a single state $s$ with an
$\agent$-loop. The model $M'$ with three states $s_1,s_2,s_3$ such
that $R_a=\{(s_1,s_2), (s_1,s_3), (s_3,s_3)\}$ satisfies $M
\lumis_\agent M'$ but is not bisimilar to any pruning of $M$.

For refinement and pruning to coincide, one can for example restrict the semantics to the class of deterministic models, that is models such that every accessibility
relation $R_\agent$ is a functional.  This is precisely the classic
setting considered in control theory. We refer to Section~\ref{ex-2} where
an example will be given.

Also, pruning plays an important role in game theory, where strategies
are in one-to-one correspondence with prunings of the unraveled
arena. However, refinement is enough to consider: for example,
concerning turn-based 2-player zero-sum games with
$\omega$-regular winning conditions \cite{booklncs2500}, we have the
following: if $G$ and $G'$ are two bisimilar arenas, then a player
has a winning strategy in $G$ iff she has winning strategy in
$G'$. Therefore, for a given arena $G$, the existence of a refinement
of $G$ such that the winning conditions hold is equivalent to
determining the existence of a winning strategy in $G$ itself. This last
remark strengthens the relevance of our refinement operator.
}

\subsection{Modal specifications refinement}

Modal specifications are classic, convenient, and expressive
mathematical objects that represent interfaces of component-based
systems \cite{LarsenNW07a,TheseJBR07,Raclet2007b,RB-acsd09,AHLNW-08,ryanetal:2002}.
Modal specifications are deterministic automata equipped with
transitions of two types: {\emph{may}} and {\emph{must}}. The
components that implement such interfaces are deterministic automata;
an alternative language-based semantics can therefore be considered,
as presented in \cite{TheseJBR07,Raclet2007b}. Informally, a
must-transition is available in every component that implements the
modal specification, while a may-transition need not be. Modal
specifications are interpreted as logical specifications matching the
conjunctive $\nu$-calculus fragment of the $\mu$-calculus
\cite{jdeds-feuillade-pinchinat07}. In order to abstract from a
particular implementation, an entire theory of modal specifications has
been developed, which relies on a refinement preorder, known as
\emph{modal refinement}. However, although its definition is close to
our definition of refinement, the two notions are incomparable: there
is no way to interpret may and must as different agents (agent $a$ and another agent $b \neq a$ have clearly independent roles in the semantics of $a$-refinement), because `must' is a subtype of `may'.

\section{Refinement modal logic}\label{lang}
\label{sec-synt-sem}

In this section we present the {\em refinement modal logic}, wherein we add a modal operator that we call a refinement quantifier to the language of multi-agent modal logic, or to the language of the modal $\mu$-calculus. From prior publications \cite{hvdetal.loft:2009,hvdetal.felax:2010} refinement modal logic is known as `future event logic'. In that interpretation different $\Box_a$ operators stand for different epistemic operators (each describing what an agent knows), and refinement modal logic is then able express what informative events are consistent with a given information state. However, here we take a more general stance.

We list some relevant validities and semantic properties, and also relate the logic to well-known logical frameworks such as bisimulation quantified modal logic (by way of relativization), and dynamic epistemic logics.

\subsection{Syntax and semantics of refinement modal logic} \label{subsec.syntaxrml}

The syntax and the semantics
of refinement modal logic are as follows.  
\begin{definition}[Languages $\lang_\G$ and $\lang^\mu_\G$] \label{def.langrml}
Given a finite set of agents $\Agents$ and a countable set of propositional atoms $\Atoms$, the language $\lang_\G$ of refinement modal logic is inductively defined as
\[ \begin{array}{l}
\phi ::= \atom \ | \ \neg \phi \ | \ (\phi \et \phi) \ | \ \know_\agent \phi \ | \ \G_\agent \phi \end{array} \]
where $\agent\in\Agents$ and $\atom\in\Atoms$. Similarly, the language $\lang^\mu_\G$ of refinement $\mu$-calculus has an extra inductive clause $\mu x. \phi$, where $X$ is the set of variables and $x \in X$.
\[ \begin{array}{l}
\phi ::= x \ | \ \atom \ | \ \neg \phi \ | \ (\phi \et \phi) \ | \ \know_\agent \phi \ | \ \G_\agent \phi \ | \ \mu x. \phi \end{array} \] 
\end{definition}
We write $\F_\agent \phi$ for $\neg \G_\agent \neg \phi$. For a subset $\{\agent_1,\dots,\agent_n\} = \Group \subseteq \Agents$ of agents we introduce the abbreviation $\F_\Group \phi$ for $\F_{\agent_1} \dots \F_{\agent_n} \phi$ (in any order), where we write $\F \phi$ for $\F_\Agents \phi$, and similarly for $\G_\Group$ and $\G$. (So in the single-agent version we are also entitled to write $\G$ and $\F$.)

Note the two differences between bisimulation quantifiers $\bqall \atom$ and the refinement quantifier $\forall$. The former we write with a `tilde'-symbol over the quantifier. The latter (and also $\forall_\agent$) has no variable. A refinement quantifier can be seen as {\em implicitly} quantifying over a variable, namely over a variable that does not occur in the formula $\phi$ that it binds (nor should it occur in a formula of which $\F \phi$ is a subformula). Section \ref{sec.relativization} will relate bisimulation quantification to the refinement operator.


\begin{definition}[Semantics of refinement] \label{def:satisfaction}
Assume a model $M = (\States, R, V)$.
\[ \begin{array}{l}
M_\state \models \G_\agent \phi \ \mbox{iff} \  \text{for all } M'_{\state'}: M_\state \lumis_\agent M'_{\state'} \text{ implies } M'_{\state'} \models \phi
\end{array} \]
The set of validities of $\lang_\G$ is the logic $\logicRML$ ({\em refinement modal logic}) and the set of validities of $\lang^\mu_\G$ is the logic $\logicRML^\mu$ ({\em refinement $\mu$-calculus}).\footnote{As is usual in the area, we will continue to use the term `logic' in a general sense, beyond that of a set of validities.}
\end{definition}
In other words, $\G_\agent \phi$ is true in a pointed model iff $\phi$ is true in all its $\agent$-{\em refinements}. Typical model operations that produce an $\agent$-refinement are: blowing up the model (to a bisimilar model) such as adding copies that are indistinguishable from the current model and one another, and removing pairs of the accessibility relation for the agent $\agent$ (or, alternatively worded: removing states accessible only by agent $\agent$). In the final part of this section we relate these semantics to the well-known frameworks action model logic and bisimulation quantified logic (and see also \cite{hvdetal.loft:2009}).

\begin{proposition}[Bisimulation invariance] \label{lem:bisim.invar}
Refinement modal logic and refinement $\mu$-calculus are bisimulation invariant.
\end{proposition}
\begin{proof}
Bisimulation invariance is the following property: given $M_s\bisim N_t$ and a formula $\phi$, then $M_s \models \phi$ iff $N_t \models \phi$. If the logic has operators beyond the standard modalities $\Box_a$, this property does not automatically follow from bisimilarity.

For refinement modal logic bisimulation invariance is straightforward, noting that $\Box_\agent$ is bisimulation invariant, and that $\mu x$ is bisimulation invariant. The new operator $\G_\agent$ is bisimulation invariant, because $\agent$-refinement is transitive and bisimulation is just a specific type of $\agent$-refinement. Formally, let $M_s\bisim N_t$, and $M_s \models \G_\agent \phi$, we have to prove that $N_t \models \G_\agent \phi$. Let $O_u$ be arbitrary such that $N_t \lumis_\agent O_u$. From $M_s\bisim N_t$ follows $M_s\lumis_\agent N_t$. From $M_s \lumis_\agent N_t$ and $N_t\lumis_\agent O_u$ follows by Proposition~\ref{lem:lumis} that $M_s \lumis_\agent O_u$. From $M_s \models \G_\agent \phi$ and $M_s \lumis_\agent O_u$ follows $O_u \models \phi$. As $O_u$ was arbitrary, we therefore conclude $N_t\models\G_\agent\phi$. The reverse direction is symmetric.
\end{proof} 

The following result justifies our notation $\F_\Group$ for sets of agents.
\begin{proposition} \label{prop-abba}
For all agents $\agent,\agentb$, $\models \F_\agent \F_\agentb \phi \eq \F_\agentb \F_\agent \phi$.
\end{proposition}
\begin{proof}
Let $M_s$ be given and let $M_t$ and $M_u$ be such that $M_s \lumis_\agent M_t$ and  $M_t \lumis_\agentb M_u$. We have that $M_s(\lumis_\agent \circ \lumis_\agentb) M_u$ iff $M_s \lumis_{\{\agent,\agentb\}} M_u$ iff $M_s(\lumis_\agentb \circ \lumis_\agent) M_u$. (See Proposition \ref{lem:compost}.)
\end{proof}

\begin{proposition} \label{prop.validities}
The following are validities of $\logicRML$.
\begin{itemize}
\item $\G_\agent \phi \imp \phi$ (reflexivity)
\item $\G_\agent \phi \imp \G_\agent\G_\agent \phi$ (transitivity)
\item $\F_\agent (\phi \vel\psi) \eq (\F_\agent\phi \vel \F_\agent \psi)$ and $\G_\agent (\phi \et\psi) \eq (\G_\agent\phi \et \G_\agent \psi)$
\item $\F_\agent \G_\agent \phi \imp \G_\agent \F_\agent \phi$ (Church-Rosser)
\item $\F_\agent \Dia_\agent \phi \eq \Dia_\agent \F_\agent \phi$ 
\end{itemize}
\end{proposition}
\begin{proof} 
The first three items directly follow from Proposition \ref{lem:lumis}. The trivial refinement is an $\agent$-refinement; composition of two refinements is a refinement; and indeed it satisfies the Church-Rosser property. The fourth item directly follows from the semantics; consider the diamond form of the equivalence: the right-to-left direction is trivial, for the left-to-right direction note that if $\phi \vel \psi$ is true in some refinement of a given model, then $\phi$ is true or $\psi$ is true in that refinement, so $\F_\agent\phi$ is true or $\F_\agent \psi$ is true in the given model.

For the fourth, from left to right: let $M_s$ be such that $M_s \models \F_\agent \Dia_\agent \phi$, and let $M'_{s'}$ and $t'\in s'R'_\agent$ be such that $M_s \succeq_\agent M'_{s'}$,  $M'_{s'} \models \Dia_\agent \phi$, and $M'_{t'} \models \phi$. Because of {\bf back}, there is a $t\in sR_\agent$ such that  $M_t \succeq_\agent M'_{t'}$. Therefore $M_t \models \F_\agent \phi$ and thus $M_s \models \Dia_\agent \F_\agent \phi$. 

From right to left: let $M_s$ be such that $M_s \models \Dia_\agent \F_\agent \phi$, and let $t\in sR_\agent$ and $M'_{t'}$ be such that $M_t \succeq_\agent M'_{t'}$, $M_{t} \models \F_\agent \phi$, and $M'_{t'} \models \phi$. Consider the model $N$ with point $s$ that is the disjoint union of $M$ and $M'$ except that: all outgoing $a$-arrows from $s$ in $M$ are removed (all pairs $(s,t) \in R_a$), a new $a$-arrow links $s$ to $t'$ in $M'$ (add $(s,t')$ to the new $R_a$). Then $N_s$ is an $a$-refinement of $M_s$ that, obviously, satisfies $\Dia_\agent \phi$, so $M_s$ satisfies $\F_\agent \Dia_\agent \phi$. (This construction is typical for refinement modal logic semantics. It will reappear in various more complex forms later, e.g., in the soundness proof of the axiomatization $\axiomRML$.)
\end{proof}

The semantics of refinement modal logic is with respect to the class $\mathcal K$ of all models (for a given set of agents and atoms). If we restrict the semantics to a specific model class only, we get a very different logic. For example $\F \Box \bot$ is a validity in $\logicRML$: just remove all access. But in refinement epistemic logic, interpreted on ${\mathcal S5}$ models, this is not a validity: seriality of models must be preserved in every refinement. See \cite{hvdetal.felax:2010,halesetal:2011}. 


\subsection{Examples} \label{sec.examples}

\paragraph{Change of knowledge} \label{ex-1}

Given are two agents that are uncertain about the
  value of a fact $\atom$, and where this is common knowledge, and
  where $\atom$ is true. Both accessibility relations
  are equivalence relations, so the epistemic operators model
  the agents' knowledge. An informative event is possible after which
  $\agent$ knows that $\atom$ but $\agentb$ does not know that; this is expressed by
  $$\F_\agent (\know_\agent \atom \et \neg \know_\agentb \know_\agent   \atom)$$

  In Figure~\ref{fig-knowledge}, the initial state of information is on the left, and its
  refinement validating the postcondition is on the right. 
  In the visualization the actual states are underlined. If states are accessible for both $\agent$ and $\agentb$ we have labelled the (single) arrow with $\agent\agentb$. 

\begin{figure}[h]
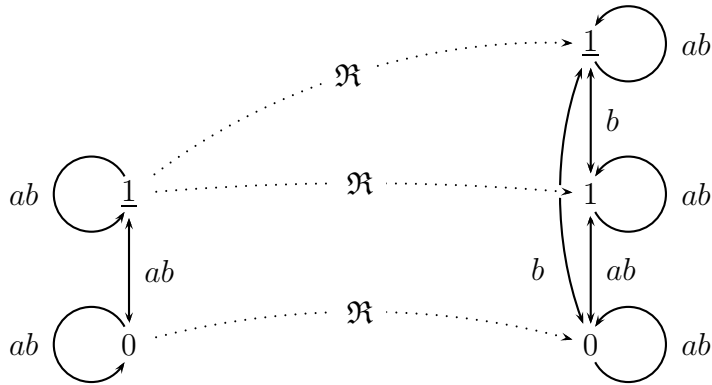

\[
\psset{border=2pt, nodesep=4pt, radius=2pt, tnpos=a}
\pspicture(0,0)(2,4)
$
\rput(0,0){\rnode{00}{0}}
\rput(0,2){\rnode{10}{\underline{1}}} 
\ncline{<->}{00}{10} \nbput{ab}
\nccircle[angle=90]{->}{00}{.5} \nbput{\agent\agentb}
\nccircle[angle=90]{->}{10}{.5} \nbput{\agent\agentb}
$
\endpspicture
\hspace{2cm} \ \hspace{2cm}
\psset{border=2pt, nodesep=4pt, radius=2pt, tnpos=a}
\pspicture(0,0)(2,2)
$
\rput(0,0){\rnode{01a}{0}}
\rput(0,2){\rnode{11a}{1}}
\rput(0,4){\rnode{10a}{\underline{1}}} 
\ncline{<->}{01a}{11a} \nbput{ab}
\ncline{<->}{10a}{11a} \naput{b}
\ncarc[arcangle=20,npos=0.15]{<->}{01a}{10a} 
\rput(-.7,1){\rnode{trick}{b}} 
\nccircle[angle=270]{->}{01a}{.5} \nbput{\agent\agentb}
\nccircle[angle=270]{->}{10a}{.5} \nbput{\agent\agentb}
\nccircle[angle=270]{->}{11a}{.5} \nbput{\agent\agentb}
\ncarc[arcangle=15,linestyle=dotted]{->}{00}{01a} \ncput*{\bisrel}
\ncarc[arcangle=20,linestyle=dotted]{->}{10}{10a} \ncput*{\bisrel}
\ncarc[arcangle=5,linestyle=dotted]{->}{10}{11a} \ncput*{\bisrel}
$
\endpspicture
\]
\caption{An example of refinement as change of knowledge}
\label{fig-knowledge}
\end{figure}

On the left, the formula $\F (\know_\agent \atom \et \neg \know_\agentb \know_\agent \atom)$ is true, because $\know_\agent \atom \et \neg \know_\agentb \know_\agent \atom$ is true on the right. On the right, in the actual state there is no alternative for agent $\agent$ (only the actual state itself is considered possible by $\agent$), so $\know_\agent \atom$ is true, whereas agent $\agentb$ also considers another state possible, wherein agent $\agent$ considers it possible that $\atom$ is false. Therefore, $\neg \know_\agentb \know_\agent \atom$ is also true in the actual state on the right.

The model on the right in the figure is neither an $\agent$-refinement of the model on the left, nor a $\agentb$-refinement of it, but an ${\{\agent,\agentb\}}$-refinement.

Recalling Section \ref{subsec.am} on action models, a refinement of a pointed model can also be obtained by executing an epistemic action (Proposition \ref{prop.actionmodel}). Therefore, we should be able to see the refinement in this example as produced by an epistemic action. This is indeed the case. The epistemic action consists of two action points $\mathsf{t}$ and $\mathsf{p}$, they can be distinguished by agent $\agent$ but not by agent $\agentb$. What really happens is $\mathsf{p}$; it has precondition $\atom$. Agent $\agentb$ cannot distinguish this from $\mathsf{t}$ with precondition $\top$. 

The execution of this action is depicted in Figure \ref{fig.action}. The point of the structure is the one with precondition $\atom$: in fact, $\agent$ is learning that $\atom$, but $\agentb$ is uncertain between that action and the `trivial' action wherein nothing is learnt. The trivial action has precondition $\T$. It can be executed in both states of the initial model. The actual action can only be executed in the state where $\atom$ is true. Therefore, the resulting structure is the refinement with three states.

\begin{figure}[h]
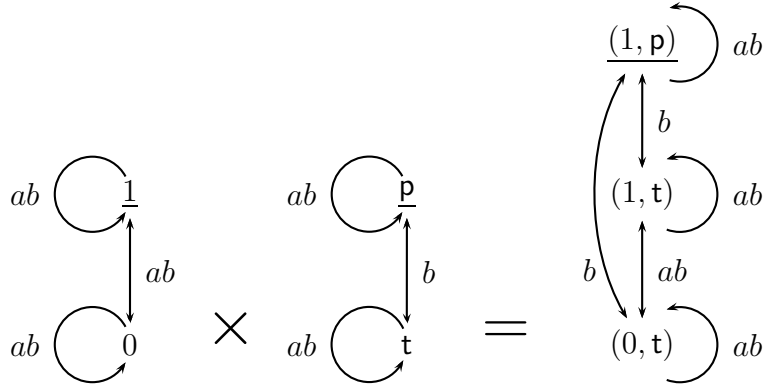

\[
\psset{border=2pt, nodesep=4pt, radius=2pt, tnpos=a}
\pspicture(-1,0)(0,4)
$
\rput(0,0){\rnode{00}{0}}
\rput(0,2){\rnode{10}{\underline{1}}} 
\ncline{<->}{00}{10} \nbput{ab}
\nccircle[angle=90]{->}{00}{.5} \nbput{\agent\agentb}
\nccircle[angle=90]{->}{10}{.5} \nbput{\agent\agentb}
$
\endpspicture
\hspace{1cm} \text{\huge $\times$} \hspace{1cm}
\pspicture(-1,0)(0,4)
$
\rput(0,0){\rnode{00}{\mathsf{t}}}
\rput(0,2){\rnode{10}{\underline{\mathsf{p}}}} 
\ncline{<->}{00}{10} \nbput{b}
\nccircle[angle=90]{->}{00}{.5} \nbput{\agent\agentb}
\nccircle[angle=90]{->}{10}{.5} \nbput{\agent\agentb}
$
\endpspicture
\hspace{1cm} \text{\huge $=$} \hspace{1cm}
\psset{border=2pt, nodesep=4pt, radius=2pt, tnpos=a}
\pspicture(-.5,0)(1,4)
$
\rput(0,0){\rnode{01a}{(0,\mathsf{t})}}
\rput(0,2){\rnode{11a}{(1,\mathsf{t})}}
\rput(0,4){\rnode{10a}{\underline{(1,\mathsf{p})}}} 
\ncline{<->}{01a}{11a} \nbput{ab}
\ncline{<->}{10a}{11a} \naput{b}
\ncarc[arcangle=30,npos=0.15]{<->}{01a}{10a} 
\rput(-.7,1){\rnode{trick}{b}} 
\nccircle[angle=270]{->}{01a}{.5} \nbput{\agent\agentb}
\nccircle[angle=270]{->}{10a}{.5} \nbput{\agent\agentb}
\nccircle[angle=270]{->}{11a}{.5} \nbput{\agent\agentb}
$
\endpspicture
\]
\caption{The refinement in Example~\ref{ex-1}.}
\label{fig.action}
\end{figure}

Action models can also be added as primitives to the multi-agent modal logical language and are then interpreted with a dynamic modal operator --- similar to automata-PDL. To get a well-defined logical language, the set of action model frames needs to be enumerable, and therefore such action models must be finite. Thus we get {\em action model logic}. We now recall the result in Proposition \ref{prop.actionmodel} that on finite models every refinement corresponds to the execution of an action model and vice versa (where the action model constructed from a given refinement may be infinite), but that it is unknown if that finiteness restriction can be lifted. If that result can be generalized, that would be of interest, as that would suggest that refinement modal logic is equally expressive as action model logic with quantification over action models. If these logics were equally expressive, action model logic with quantification would be decidable---a surprising fact, given that public announcement logic with quantification over public announcements (singleton action models) is undecidable \cite{frenchetal:2008}.

\paragraph{Software verification and design} \label{ex-2}
Consider a class of discrete-event systems, whose elements represent
devices that interact with an environment. Each device is described
by means of actions $c$ and $u$, respectively called `controllable'
and `uncontrollable' actions.  Given an expected property described
by some formula $\phi$, say in $\lang^\mu$, we use refinement
quantifiers to express several classic verification/synthesis
problems. We let $\Box\phi$ stand for $\Box_c\phi\et\Box_u\phi$.

The \emph{the control problem} \cite{DBLP:conf/mfcs/RiedwegP03}, known
as the question ``is there a way to control actions $c$ of the system
$S$ so that property $\phi$ is guaranteed?'', can be expressed in
$\lang_\G$ by wondering whether \[ S \models \F_c \phi \ . \]

The \emph{module checking problem} \cite{kupferman01} is the problem
of determining whether an \emph{open} system satisfies a given
property. In other words, whether the property holds when the system
is composed with an arbitrary environment. Let us say that action $c$
is an abstract action that denotes internal ones, while action $u$
abstracts all external actions, i.e.\ actions performed by the environment. Also, assume there
is an atomic proposition $e$ that distinguishes states where it is
the turn of the system to act (thus only action $c$ is available) from states
where it is the turn of the environment (thus only action $u$ is
available). In this setting, we answer positively to the module
checking problem iff $S \models \G_u \phi$. As arbitrary environments
are too permissive, we may force hypotheses such as restricting to
\emph{non-blocking} environments: the property can be captured by the
$\lang^\mu$-formula $\text{\tt NonBlockingEnv} := \nu x. (e
\Rightarrow \M_u \top) \wedge \know x$, which formally says that it is
always the case ($ \nu x. (....) \wedge \know x$) that
whenever in an environment state, there is an outgoing transition from
that state ($e \Rightarrow \M_u \top$).  Now, by `guarding' the
universal quantification over all $u$-refinements (i.e.\ all
environments) with the $\text{\tt NonBlockingEnv}$ assumption, the
statement becomes \[ S \models \G_u (\text{\tt NonBlockingEnv}
\Rightarrow \phi) \]

The \emph{generalized control problem} is the combination of the two
previous problems, by questioning the existence of a control such
that the controlled system satisfies the property in all possible
environments. This is expressed by wondering whether \[ S \models \F_c
\G_u (\text{\tt NonBlockingEnv} \Rightarrow \phi) \ . \]

A last example is borrowed from \emph{protocol synthesis
  problems}. Consider a specification, \text{\tt MUTEX}, of a mutual exclusion
protocol involving processes $1,2,\ldots k$, and some property
$\phi$ specified in $\lang^\mu$. Now we may ask if we can find a refinement
of \text{\tt MUTEX} that satisfies $\phi$ but also such that if
process $i$ is in the critical section ($cs_i$) at time $n+1$, then
this is known at time $n$. This is expressed as
\begin{equation*}
\text{\tt MUTEX} \models \F[\AG(\M cs_i \Rightarrow \know cs_i) \wedge \phi]
\end{equation*}
where $\AG$ is the CTL-modality, which is defined in $\lang^\mu$ as
$\AG(\psi)\equiv\nu x.\psi \wedge \know x$ and meaning that this is
true at any time. The refinement consists in moving the
nondeterministic choices forward, so that a fork at time $n$ becomes a
fork at time $n-1$ with each branch having a single successor at time
$n$, as depicted in Figure~\ref{mutex}.
\begin{figure}[h]
\begin{center}\scalebox{0.5}{
\begin{picture}(0,0)%
 \epsfig{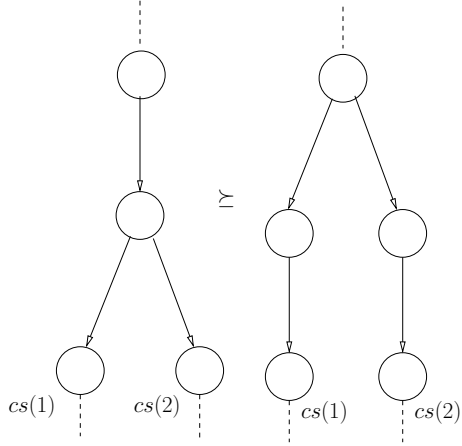}%
 \end{picture}%
 \setlength{\unitlength}{3947sp}%
 \begingroup\makeatletter\ifx\SetFigFont\undefined%
 \gdef\SetFigFont#1#2#3#4#5{
   \fontsize{#1}{#2pt}%
   \fontfamily{#3}\fontseries{#4}\fontshape{#5}%
    \selectfont}%
    \fi\endgroup%
    \begin{picture}(5930,5574)(1201,-5173)
    \put(6301,-4861){\makebox(0,0)[lb]{\smash{{\SetFigFont{17}{20.4}{\rmdefault}{\mddefault}{\updefault}{\color[rgb]{0,0,0}$cs(2)$}%
    }}}}
    \put(3826,-2161){\makebox(0,0)[lb]{\smash{{\SetFigFont{17}{20.4}{\rmdefault}{\mddefault}{\updefault}{\color[rgb]{0,0,0}$\lumis$}%
    }}}}
    \put(2776,-4786){\makebox(0,0)[lb]{\smash{{\SetFigFont{17}{20.4}{\rmdefault}{\mddefault}{\updefault}{\color[rgb]{0,0,0}$cs(2)$}%
    }}}}
    \put(1201,-4786){\makebox(0,0)[lb]{\smash{{\SetFigFont{17}{20.4}{\rmdefault}{\mddefault}{\updefault}{\color[rgb]{0,0,0}$cs(1)$}%
    }}}}
    \put(4876,-4861){\makebox(0,0)[lb]{\smash{{\SetFigFont{17}{20.4}{\rmdefault}{\mddefault}{\updefault}{\color[rgb]{0,0,0}$cs(1)$}%
    }}}}
    \end{picture}%
}
\end{center}
\caption{The refinement of \text{\tt MUTEX}.\label{mutex}}
\end{figure}

\subsection{Refinement quantification is bisimulation quantification plus relativization} \label{sec.relativization}

In Section \ref{sec.modelrestriction} we presented a semantic perspective of refinement as bisimulation followed by model restriction, or, alternatively and equivalently, as a restricted bisimulation, namely except for some propositional variable, followed by a model restriction to that variable. We now lift this result to a corresponding syntactic, logical, perspective of the refinement quantifier as a bisimulation quantifier followed by relativization.

More precisely, in this section we will show that a refinement formula $\F_\agent \phi$ is equivalent to a bisimulation quantification over a variable not occurring in $\phi$, followed by a (non-standard) relativization for that agent to that variable, for which we write $\bqis \atomb \phi^{(\agent,\atomb)}$ (to be defined shortly). For refinement $\succeq$ for the set of all agents (recall that we write $\succeq$ for $\succeq_\Agents$, and $\F$ for $\F_\Agents$) we can expand this perspective to even more familiar ground: a refinement formula $\F \phi$ is equivalent to a bisimulation quantification over a variable not in $\phi$ followed by (standard) relativization to that variable: $\bqis \atomb \phi^\atomb$. These results immediately clarify in what sense the refinement modality constitutes `implicit' quantification, namely over a variable not occurring in the formula bound by it.

For the syntactic correspondence we first introduce the notion of relativization (for settings in modal logic, see \cite{jfak.lonely:2006,milleretal:2005}). We propose a  definition of relativization that may be considered non-standard for several reasons. Firstly, it is relativization not merely to a propositional variable but also for a given agent only. The standard definition is then the special case of relativization to that variable for all agents (we will prove that consecutive relativization to the same variable for two different agents is commutative, in other words, order independent). Secondly, the relativization that we propose corresponds in the semantics to arrow elimination and not to state elimination (in other words, it does not correspond to submodel restriction). 
From the modal logical literature, the approach in \cite{milleretal:2005} is arrow-eliminating but that in \cite{jfak.lonely:2006} is state-eliminating. 

The arrow-eliminating relativization need only be done in accessible states but not in the actual state (e.g., the relativization of a variable $\atomb$ to a variable $\atom$ is that same variable $\atomb$ and not $\atom \et \atomb$). 

The difference between state-eliminating relativization and arrow-eliminating relativization is similar to the difference between state-eliminating public announcement semantics \cite{plaza:1989,baltagetal:1998} and arrow-eliminating public announcement semantics \cite{kooi.jancl:2007,gerbrandyetal:1997}, in the area of dynamic epistemic logic. As our relativization is with respect to a given agent, we have no option but to use arrow-eliminating relativization. 

Given our purpose to translate refinement modal logic into bisimulation quantified modal logic, we also expand the definition of relativization to include quantifiers. This definition will then be used in Section \ref{sec-FEL-mu}.
\begin{definition}[Relativization] \label{def-relativization} 
Relativization $\bullet^{(\agent,\atom)}: \lang_\bqall \imp \lang_\bqall$ to propositional variable $\atom$ for agent $\agent\in\Agents$ is defined as follows. 
\[ \begin{array}{lcl}
\atomb^{(\agent,\atom)} & = & \atomb \\
(\neg\phi)^{(\agent,\atom)} & = & \neg\phi^{(\agent,\atom)} \\
(\phi\et\psi)^{(\agent,\atom)} & = & \phi^{(\agent,\atom)} \et \psi^{(\agent,\atom)} \\
(\Box_\agent\phi)^{(\agent,\atom)} & = & \Box_\agent(\atom \imp \phi^{(\agent,\atom)}) \\
(\Box_\agentb\phi)^{(\agent,\atom)} & = & \Box_\agentb \phi^{(\agent,\atom)} \hspace{3cm} \hfill \text{for } \agentb\neq\agent \\
(\bqall \atomb \phi)^{(\agent,\atom)} & = & \bqall \atomb \phi^{(\agent,\atom)} \hspace{3cm} \hfill \text{for } \atomb\neq\atom \\
(\bqall \atom \phi)^{(\agent,\atom)} & = & \bqall \atomb \phi[\atomb\backslash\atom]^{(\agent,\atom)} \hspace{3cm} \hfill \text{choose } \atomb\text{ that does not occur in } \phi
\end{array} \]
\end{definition}
\begin{lemma} \label{lemma.relatother}
Let $M_s$ be a model with accessibility function $R$ and $R'_\agent \subseteq R_\agent$ such that: $(t,t') \in R'_\agent$ iff $M_{t'} \models \atom$. Then $M_s \models \phi^{(\agent,\atom)}$ if and only if $M_s | R'_\agent \models \phi$.
\end{lemma}
\begin{proof}
The proof is by induction on the structure of $\phi$.
\begin{itemize}
\item $
M_s \models \atomb^{(\agent,\atom)} \Eq \\
M_s \models \atomb \Eq \hspace{2cm} \text{propositional variables do not change value} \\
M_s | R'_\agent \models \atomb
$
\item The clauses for negation and conjunction are elementary.
\item $
M_s \models (\Box_\agent \phi)^{(\agent,\atom)} \Eq \\
M_s \models \Box_\agent (\atom \imp \phi^{(\agent,\atom)}) \Eq \\
\text{for all } t \in sR_\agent: M_t \models \atom \imp \phi^{(\agent,\atom)} \Eq \\
\text{for all } t \in sR_\agent: M_t \models \atom \text{ implies } M_t \models \phi^{(\agent,\atom)} \Eq \hspace{2cm} \text{I.H.} \\
\text{for all } t \in sR_\agent: M_t \models \atom \text{ implies } M_t  | R'_\agent \models \phi \Eq \hspace{1cm} t \in sR_\agent \text{ and } t \models \atom \text{ iff } t \in sR'_\agent \\
\text{for all } t \in sR'_\agent: M_t | R'_\agent \models \phi \Eq \\
M_s | R'_\agent \models \Box_\agent \phi
$
\item $
M_s \models (\Box_\agentb \phi)^{(\agent,\atom)} \Eq \\
M_s \models \Box_\agentb \phi^{(\agent,\atom)} \Eq \\
\text{for all } t \in sR_\agentb: M_t \models \phi^{(\agent,\atom)} \Eq  \hspace{2cm} \text{I.H.} \\
\text{for all } t \in sR_\agentb \text{ (in } M_t): M_t | R'_\agent \models \phi \Eq \hspace{2cm} sR_\agentb \text{ in } M \text{ equals } sR_\agentb \text{ in } M | R'_\agent \\
\text{for all } t \in sR_\agentb \text{ (in } M_t | R'_\agent): M_t | R'_\agent \models \phi \Eq \\
M_s | R'_\agent \models \Box_\agentb \phi
$
\item For a more natural argument we take the existential quantifier instead of the universal quantifier. First, observe that: 

\medskip

\noindent $
M_s \models (\bqis\atomb\phi)^{(\agent,\atom)} \Eq \\
M_s \models \bqis\atomb \phi^{(\agent,\atom)} \Eq \\
\text{there is an } N_t \bisim^\atomb M_s: N_t \models \phi^{(\agent,\atom)} \Eq \hspace{.5cm} \text{I.H.} \\
\text{there is an } N_t \bisim^\atomb M_s: N_t| R''_\agent \models \phi \hspace{.5cm} \text{where } R''_\agent \subseteq R^N_\agent \text{ s.t. } (u,u') \in R''_\agent \text{ iff } M_{u'} \models \atom
$

\medskip

We also have that, by definition:

\medskip

\noindent $
M_s | R'_\agent \models \bqis\atomb\phi \Eq \\ 
\text{there is an } N'_{t'} \bisim^\atomb M_s | R'_\agent: N'_{t'} \models \phi
$

\medskip

It remains to show that the two final statements in these chains of equivalences are also equivalent. 

From left to right is easy. If $\bisrel: N_t \bisim^\atomb M_s$, then also $\bisrel: N_t | R''_\agent \bisim^\atomb M_s| R'_\agent$. In $N$ and $M$ we remove all $\agent$-arrows to $\neg\atom$ states; and if it is already a bisimulation, then the forth and back requirements still hold for fewer pairs in the accessibility relation for $\agent$. So we can take $N'_{t'} = N_t | R''_\agent$. 

From right to left is not easy. Let us first explain this informally. Given that $N'_{t'} \bisim^\atomb M_s | R'_\agent$, the part of $M$ that is inaccessible from $M_s | R'_\agent$ (i.e., not in the $s$-generated submodel) may not be bisimilar to anything in $N'$. This is problematic, because we need to transform $N'_{t'}$ to some $N_t$ in a way that establishes a $\atomb$ restricted bisimulation between $N_t$ and all of $M_s$. Fortunately, the transformation can be an extension of $N'_{t'}$, wherein we uniformly treat states in that inaccessible part of $M$ and other states of $M$: we do not need to be economic in our construction. These are the details.

Let $\bisrel: N'_{t'} \bisim^\atomb M_s | R'_\agent$ be the restricted bisimulation. To be explicit, let $M = (\States,R,V)$ and let $N' = (\States',R',V')$. Consider $\States^{\agent\overline{p}} = \{ u \in \States^M \mid \exists v \in S, (v,u) \in R_\agent\setminus R'_\agent \}$. For each $u \in \States^{\agent\overline{p}}$ we need an exact copy $M^u$ of $M$ (let $M^u = (\States^u,R^u,V^u)$) in our construction. We now define $N = (\States^N,R^N,V^N)$ as follows: \begin{itemize} \item $\States^N = \States' \union \Union \{ \States^u \mid u \in \States^{\agent\overline{p}} \}$; \item for all $\agentb\neq\agent$, $R^N_\agentb = R'_\agentb \union \Union \{ R^u_\agentb \mid u \in \States^{\agent\overline{p}} \}$; \item $R^N_\agent = R'_\agent \union \Union \{ R^u_\agent \mid u \in \States^{\agent\overline{p}} \} \union \{(v',u) \mid (v',v) \in \bisrel$ and $(v,u) \in R_\agent\setminus R'_\agent \}$; \item for all $\atom\in\Atoms$, $V^N(\atom) = V'(\atom) \union \Union \{ V^u(\atom) \mid u \in \States^{\agent\overline{p}} \}$. 
\end{itemize}
Now take $t = t'$, and let $R''_\agent$ as before be restriction of $R^N_\agent$ to pairs $(u,u') \in R^N_\agent$ such that $u'$ satisfies $p$. It is now immediate that $N_t \bisim^\atomb M_s$ and therefore also $N_t| R''_\agent \models \phi$.

\item The other clause for the universal quantifier starts with a renaming operation (that equally applies to the existential quantifier), and then proceeds as in the previous clause.
\end{itemize}
\end{proof}
Agent relativization relates as expected to the standard notion of relativization (for the set of all agents simultaneously). This is because relativization to different variables for different agents is commutative.
\begin{lemma} \label{lemma.relswap}
Let $\phi \in \lang_\bqall$. Then $(\phi^{(\agent,\atom)})^{(\agentb,\atomb)} = (\phi^{(\agentb,\atomb)})^{(\agent,\atom)}$.
\end{lemma}
\begin{proof}
By induction on the structure of $\phi$. The non-trivial cases are $\Box_\agent \phi$, $\Box_\agentb \phi$ (follows dually), $\bqall \atom \phi$, and $\bqall \atomb \phi$ (also follows dually). Note that $(\agent,\atom)$-relativization distributes over implication.
\begin{itemize}
\item $((\Box_\agent \phi)^{(\agent,\atom)})^{(\agentb,\atomb)} \Eq \\
(\Box_\agent (\atom \imp \phi^{(\agent,\atom)}))^{(\agentb,\atomb)} \Eq \\
\Box_\agent (\atom \imp \phi^{(\agent,\atom)})^{(\agentb,\atomb)} \Eq \\
\Box_\agent (\atom^{(\agentb,\atomb)} \imp (\phi^{(\agent,\atom)})^{(\agentb,\atomb)}) \Eq \hspace{2cm} \text{I.H., and clause for variables} \\
\Box_\agent (\atom \imp (\phi^{(\agentb,\atomb)})^{(\agent,\atom)}) \Eq \\
(\Box_\agent \phi^{(\agentb,\atomb)})^{(\agent,\atom)} \Eq \\
((\Box_\agent \phi)^{(\agentb,\atomb)})^{(\agent,\atom)}
$

\item $((\bqall \atom \phi)^{(\agent,\atom)})^{(\agentb,\atomb)} \Eq \hspace{2cm} \text{choose $\atomc\neq\atomb$ (or else, yet another step)} \\
(\bqall \atomc \phi[\atomc \backslash \atom]^{(\agent,\atom)})^{(\agentb,\atomb)} \Eq \\
\bqall \atomc (\phi[\atomc \backslash \atom]^{(\agent,\atom)})^{(\agentb,\atomb)} \Eq \hspace{2cm} \text{I.H.} \\
\bqall \atomc (\phi[\atomc \backslash \atom]^{(\agentb,\atomb)})^{(\agent,\atom)} \Eq  \hspace{2cm} \text{substitution of other variables than } \atomb \\
\bqall \atomc (\phi^{(\agentb,\atomb)}[\atomc \backslash \atom])^{(\agent,\atom)} \Eq \\
(\bqall \atom \phi^{(\agentb,\atomb)})^{(\agent,\atom)} \Eq \\
((\bqall \atom \phi)^{(\agentb,\atomb)})^{(\agent,\atom)}
$
\end{itemize}
\end{proof}
Given Lemma \ref{lemma.relswap}, we may view a nesting of relativizations $( \dots (\phi^{(\agent_1,\atom)})\dots^{(\agent_n,\atom)})$ as a relativization $\phi^{(\{\agent_1,\dots,\agent_n\},\atom)}$ for the set of agents $\{\agent_1,\dots,\agent_n\}$. Furthermore, for $\phi^{(\Agents,\atom)}$ we can write $\phi^\atom$: the usual relativization for all agents simultaneously. 

To make the syntactic correspondence we now introduce a translation.
\begin{definition} \label{def.tsjakka}
The translation $t: \lang_\G \imp \lang_\bqall$ is defined by induction on $\phi\in\lang_\G$. All clauses except $\G_\agent \phi$ are trivial.
\[ \begin{array}{lcl}
t(\atom) & = & \atom \\
t(\neg \phi) & = & \neg t(\phi) \\
t(\phi\et\psi) & = & t(\phi) \et t(\psi) \\
t(\Box_\agent\phi) & = & \Box_\agent t(\phi) \\
t(\G_\agent \phi) & = & \bqall \atom \ t(\phi)^{(\agent,\atom)} \hspace{2cm} \hfill \text{where } \atom\text{ does not occur in } \phi
\end{array} \]
\end{definition}
\begin{example}
\[ \begin{array}{llllll}
t(\F_\agent \F_\agentb \atomc) &=& 
\bqis \atom \ t(\F_\agentb \atomc)^{(\agent,\atom)} &= \\ 
\bqis \atom (\bqis \atom \ t(\atomc)^{(\agentb,\atom)})^{(\agent,\atom)} &=& 
\bqis \atom (\bqis \atom \ \atomc^{(\agentb,\atom)})^{(\agent,\atom)} &= \\ 
\bqis \atom (\bqis \atom \ \atomc)^{(\agent,\atom)} &=& 
\bqis \atom \bqis \atomb \ \atomc^{(\agent,\atomb)} &= & 
\bqis \atom \bqis \atomb \ \atomc
\end{array} \]
\end{example}
\begin{proposition} \label{prop-ref-and-bisimquant} Let $\phi\in\lang_\G$. Then $\phi$ is equivalent to $t(\phi)$.
\end{proposition}
\begin{proof} 
In the proposition we allowed ourselves a slight abuse of language: it means that, given any $M_s$, the value of $\phi$ in the semantics for refinement modal logic is equivalent to the value of $t(\phi)$ in the semantics for bisimulation quantified modal logic. The proposition follows from Lemma \ref{lemma.prev}, Lemma \ref{lemma.relatother} and Def.\  \ref{def.tsjakka}. We show the case $\G_\agent\phi$ of the inductive proof---and to suit the intuition we take the existential quantifier $\F_\agent$.

\medskip

$M_\state \models \F_\agent \phi$ 

iff 

there is an $M'_{\state'}$ such that $M_\state \lumis_\agent M'_{\state'}$ and $M'_{\state'} \models \phi$

iff (I.H.) 

there is an $M'_{\state'}$ such that $M_\state \lumis_\agent M'_{\state'}$ and $M'_{\state'} \models t(\phi)$

iff (Lemma \ref{lemma.prev})

there is an $N'_{\stateb'}$ with $R''_\agent \subseteq R'_\agent$ (restr.\ to $\atom$ true) s.t.\ $M_\state \bisim^\atom N'_{\stateb'}$ and $N'_{\stateb'} | R''_\agent \models t(\phi)$



iff (Lemma \ref{lemma.relatother})

there is an $N'_{\stateb'}$ such that $M_\state \bisim^\atom N'_{\stateb'}$ and $N'_{\stateb'} \models t(\phi)^{(\agent,\atom)}$

iff

$M_\state \models \bqis\atom \ t(\phi)^{(\agent,\atom)}$ 

iff 

$M_\state \models t(\F_\agent \phi)$.
\end{proof}
This corollary makes the characteristic cases of Proposition \ref{prop-ref-and-bisimquant} stand out.
\begin{corollary} \label{cor-ref-and-bisimquant}
Consider $\F\phi$ with $\phi\in\lang$ (i.e., $\F$-free). Then
\begin{itemize}
\item $\agent$-refinement is bisimulation quantification plus $\agent$-relativization: \\ $\F_\agent\phi$ is equivalent to $\bqis\atom\phi^{(\agent,\atom)}$;
\item refinement is bisimulation quantification plus relativization: \\ $\F\phi$ is equivalent to $\bqis\atom\phi^\atom$.
\end{itemize} 

\end{corollary}
In the logic of public announcements, the latter is written as: $\F\phi$ is equivalent to $\bqis\atom\dia{\atom!}\phi$.

\subsection{Alternating refinement relations}

Alternating transition systems (ATS) were introduced
\cite{alur98alternating} to model multi-agent systems, where in each move of
the game between the agents of an ATS, the choice of an agent at a
state is a set of states and the successor state is determined by
considering the intersection of the choices made by all agents. A
notion of \emph{$a$-alternating refinement} was introduced to reflect
a refined behavior of agent $a$ while keeping intact the behavior of
the others. When restricting to \emph{turn-based} ATS where only one
agent plays at a time (concurrent moves are also allowed in the full
setting), $a$-alternating refinement amounts to requiring `forth' for
all $b \in A \setminus \{a\}$ as we do, but `back' just for agent
$a$. As a consequence, an $a$-refinement is a particular
$a$-alternating refinement. A logical characterization of
$a$-alternating refinement has been proposed (it essentially relies on
the modality $\F_\agent$ combined with the linear time temporal logic
LTL) in the sense that if an ATS $\mathcal{S'}$ $a$-refines an ATS
$\mathcal{S}$, every formula true in $\mathcal{S'}$ is also true in
$\mathcal{S}$. Notice however that the operator $\F_\agent$ has a more
restricted semantics than the one we propose, since the quantification
does not range over all possible refinements of the structure but only
over refinements obtained by pruning the unraveling of the ATS. Soon
after, the more general setting of \emph{alternating-time temporal
  logics} \cite{alur98} considered universal and
existential quantifications over $a$-refinements, for arbitrary $a$,
combined with LTL formulas. It is worthwhile noticing that the
quantifiers still range over particular refinements, and always in the
original structure. As a consequence, the language cannot express the
ability to nest refinements for different agents. This is easily done
in our language $\lang_\G$, as the formula $\F_a (\know_b p \wedge
\M_a (\F_b \know_a p))$ exemplifies. This formula tells us that one of the choices that $a$ can make, results in $b$ knowing $p$ and $a$ contemplating a subsequent choice by $b$ that makes her to get to know $p$ as well. 

\section{Axiomatization $\axiomRML$}\label{axiom1}
\label{sec-FEL}

Here we present the axiomatization $\axiomRML$ for the logic $\logicRML$. We show the axioms and rules to be sound, we give example derivations, and this is followed by the completeness proof.

The axiomatization presented is a substitution schema, since the
substitution rule is not valid. The {\em substitution rule} says that: if $\phi$ is a theorem, and $\atom$ occurs in $\phi$, and $\psi$ is any formula, then $\phi[\psi\backslash\atom]$ is a theorem. Note that for all atomic propositions $p$, $p\imp \G p$ is valid, but the same is not true for an arbitrary formula, e.g.\ $\susp_\agent \T \imp \G \susp_\agent \T$ is not valid, because after the maximal refinement there is no accessible state, so that $\susp_\agent \T$ is then false even if it was true before. The logic $\logicRML$ is therefore not a normal modal logic.

\begin{definition}[Axiomatization $\axiomRML$] \label{def:FEL}
The axiomatization $\axiomRML$ consists of all substitution instances of the axioms
\[ \begin{array}{rl}
{\bf Prop} & \text{All tautologies of propositional logic}\\
{\bf K} & \know_\agent(\phi\imp\psi)\imp \know_\agent\phi\imp \know_\agent\psi\\
{\bf R} & \G_\agent(\phi\imp\psi)\imp \G_\agent\phi \imp \G_\agent\psi\\
{\bf RProp} & \G_\agent \atom \eq \atom \text{ and } \G_\agent \neg\atom \eq\neg\atom \\
{\bf RK} & \F_\agent\covers_\agent\Phi\eq\bigwedge\susp_\agent\F_\agent\Phi  \\
{\bf RKmulti} & \F_\agent\covers_\agentb\Phi\eq \covers_\agentb\F_\agent\Phi  \hspace{3cm} \hfill \text{ where } \agent\neq\agentb \\
{\bf RKconj} & \F_\agent \Et_{\agentb\in\Group}\covers_\agentb \Phi^\agentb \eq\Et_{\agentb\in\Group} \F_\agent \covers_\agentb \Phi^\agentb
\end{array}\]
and the rules
\[ \begin{array}{rl} 
{\bf MP} & \text{From } \phi\imp\psi \text{ and } \phi \text{ infer } \psi\\
{\bf NecK} & \text{From } \phi \text{ infer } \know_\agent\phi\\
{\bf NecR} & \text{From } \phi \text{ infer } \G_\agent\phi
\end{array} \]
where $\agent,\agentb\in\Agents$, $\atom\in\Atoms$, and $\Group\subseteq\Agents$. If $\phi$ is derivable, we write $\proves \phi$, and $\phi$ is called a {\em theorem}, as usual. The well-known axiomatization {\bf K} for the logic $\logicK$ consists of the axioms {\bf Prop}, {\bf K}, and the rules {\bf MP} and {\bf NecK}.
\end{definition}
In the definition, given $\Phi = \{ \phi_1,\dots,\phi_n \}$, note that $\F_\agent\covers_\agent\Phi \eq \bigwedge\susp_\agent\F_\agent\Phi$ stands for $\F_\agent\covers_\agent\Phi \eq \bigwedge_{\phi\in\Phi}\susp_\agent\F_\agent\phi$ (see the technical preliminaries) and so for $\F_\agent\covers_\agent \{\phi_1,\dots,\phi_n\} \eq \susp_\agent\F_\agent\phi_1 \et \dots \et \susp_\agent\F_\agent\phi_n$. The axiomatization $\axiomRML$ is surprisingly simple given the complexity of the semantic definition of the refinement operator $\G$; and given the well-known complexity of axiomatizations for logics involving bisimulation quantifiers instead of this single refinement quantifier. We note that while refinement is reflexive, transitive and
satisfies the Church-Rosser property (Proposition~\ref{lem:lumis}, and Proposition \ref{prop.validities}), the
corresponding modal axioms are not required. These properties are  schematically derivable. First, we demonstrate soundness of $\axiomRML$.

\begin{quote} {\em Given the definitions of $\Box$ and $\Dia$ in terms of cover, it may be instructive to see how the {\bf RK} axiom works as a reduction principle for $\F\Box\phi$ and $\F\Diamond\phi$---note that we need both, as there is no principle for $\F\neg\phi$. For simplicity we do not label the operators with agents. We get:
\[\begin{array}{rcl} \F\Box\phi & \eq & \F (\covers \{\phi\} \vel \covers \emptyset) \\ & \eq & \F \covers \{\phi\} \vel \F \covers \emptyset \\ \text{(use {\bf RK})} & \eq & \F \covers \{\phi\} \vel \Et \Dia \F \emptyset \\ \text{(empty conj.\ is true)} & \eq & \F \covers \{\phi\} \vel \top \\ & \eq &  \top 
\end{array}\]
and
\[\begin{array}{lcl} \F\Diamond \phi & \eq & \F \covers \{\phi,\top\} \\ \text{(use {\bf RK})} & \eq & \Diamond \F \phi \et \Diamond \F \top \\ & \eq & \Diamond \F \phi \end{array}\]
One may wonder why we did not choose $\F\Box\phi \eq \top$ and $\F\Diamond \phi \eq \Diamond \F \phi$ (we recall Proposition \ref{prop.validities}) as primitives in the axiomatization, as, after all, these are very simple axioms. They are of course valid, but the axiomatization would not be complete. The axiom ${\bf RK}$ is much more powerful, as this not merely allows $\Phi = \{\phi\}$, $\Phi = \emptyset$, and $\Phi = \{\phi,\top\}$, but any finite set of formulas.
}
\end{quote}

\subsection{Soundness}\label{sect:soundness}

\begin{theorem}
\label{theo:soundness}
The axiomatization $\axiomRML$ is sound for $\logicRML$.
\end{theorem}
\begin{proof}
As all models of $\lang_\G$ are models of 
$\lang$, the schemas {\bf Prop}, {\bf K} and the rule {\bf MP} and {\bf NecK} are all sound. We deal with the remaining schemas and rules below.

\bigskip

\noindent {\bf R} 

Suppose that $M_s$ is a model such that $M_s\models \G_\agent(\phi\imp\psi)$, and $M_s\models\G_\agent\phi$. Then for every $N_t$, where $N_t\simul_\agent M_s$, we have $N_t\models \phi\imp\psi$, and also $N_t\models \phi$. From $N_t\models \phi\imp\psi$ and $N_t\models \phi$ follows $N_t\models\psi$. As $N_t$ was arbitrary model such that $N_t\simul_\agent M_s$, from that and $N_t\models\psi$ follows $M_s \models \G_\agent\psi$.

\bigskip

\noindent {\bf RProp} 

Let $M_s$ and $N_t$ be given such that $N_t\simul_\agent M_s$. By Definition~\ref{def.bisim} for the semantics of refinement, we have that $s\in V^M(p)$ if and only if $t\in V^N(p)$. Therefore $M_s\models\atom$ iff $N_t\models\atom$, for every $M_s$ and $N_t$ with $N_t\simul_\agent M_s$. Therefore $M_s\models\atom$ iff $M_s\models\G_\agent\atom$ for every $M_s$, i.e.\ $\models \atom \eq \G_\agent \atom$. Similarly, for $\models \neg\atom \eq \G_\agent \neg\atom$, using that $s\not\in V^M(p)$ if and only if $t\not\in V^N(p)$.

\bigskip

\noindent {\bf RK} 

Suppose $M_s$ is a model, where $M = (S,R,V)$, such that for some set $\Phi$, $M_s\models\F_\agent\covers_\agent\Phi$. Therefore, there is a model $N_t\simul_\agent M_s$ such that $N_t\models\covers_\agent\Phi$---where $N = (S^\Phi,R^\Phi,V^\Phi)$. Expanding the definition, we have that for every $\phi\in\Phi$ there is some $u\in tR_\agent^\Phi$ such that $N_u\models\phi$. Also, because of {\bf back}, for every such $u\in tR_\agent^\Phi$ there is some $v\in sR_\agent$ such that $N_u\simul_\agent M_v$. Combining these statements we have that for every $\phi\in\Phi$ there is some $v\in sR_\agent$ such that $M_v\models\F_\agent\phi$, and thus $M_s\models\bigwedge\susp_\agent\F_\agent\Phi$.

Conversely, suppose that 
  $M_s\models\bigwedge\susp_\agent\F_\agent\Phi$. Therefore, for
  every $\phi\in \Phi$ there is some $t^\phi\in sR_\agent$ such that
  $M_{t^\phi}\models\F_\agent\phi$. Thus, for each $\phi\in\Phi$,
  there is some model $N^\phi_{u^\phi}\simul_\agent M_{t^\phi}$, where $N^\phi = (S^\phi,R^\phi,V^\phi)$, such that $N^\phi_{u^\phi}\models\phi$. Without loss of generality, we may assume that for all $\phi,\phi'\in\Phi$ the models $N^\phi$ and $N^{\phi'}$ are disjoint. 

We construct the model $M^\Phi = (S^\Phi, R^\Phi, V^\Phi)$ such that: 
\[ \begin{array}{lcl} 
S^\Phi & = & \{s'\}\cup S \cup \bigcup_{\phi\in\Phi} S^\phi \\
R_\agent^\Phi & = & \{(s',u^\phi)\ |\ \phi\in\Phi\}\cup R_\agent \cup \bigcup_{\phi\in\Phi} R_\agent^\phi  \\
R_\agentb^\Phi & = & \{(s',t)\ |\ (s,t)\in R_\agentb \}\cup R_\agentb \cup \bigcup_{\phi\in\Phi} R_\agentb^\phi \hfill \text{ for } \agentb \neq \agent \\
V^\Phi(p) & = & \overline{\{s'\}} \cup V(p)\cup\bigcup_{\phi\in\Phi} V^\phi(p) \hspace{4cm} \hfill \text{ for } p\in\Atoms
\end{array} \]
where $\overline{\{s'\}} = \{s'\}$ if $s \in V(p)$ and else $\overline{\{s'\}} = \emptyset$.
 
We can see that $M_s \lumis_\agent M^\Phi_{s'}$, via the relation ${\cal R}^\Phi = \{(s,s')\} \cup {\mathcal I} \cup \bigcup_{\phi\in\Phi}{\cal R}^\phi$ where ${\mathcal I}$ is the identity on $S$ and each ${\cal R}^\phi$ is the refinement relation corresponding to $M_{t^{\phi}} \lumis_\agent N^\phi_{u^\phi}$ (see also \cite{hales:2011}). Furthermore, for each $t\in s'R_\agent^\Phi$ it is clear that $M^\Phi_t\bisim N^\phi_{u^\phi}$ for some $\phi$, and thus $M^\Phi_t\models\phi$, and so $M^\Phi_t\models\Vel\Phi$. Therefore $M^\Phi_{s'}\models\knows_\agent\bigvee\Phi$. Finally, for each $\phi\in\Phi$ there is some $u^\phi\in s'R_\agent^\Phi$ where $M^\Phi_{u^\phi}\models\phi$, so for each $\phi\in\Phi$ we have $M^\Phi_s\models\Diamond_\agent\phi$, so we have $M^\Phi_{s'}\models\bigwedge\susp_\agent\Phi$. Combined, $M^\Phi_{s'}\models\knows_\agent\bigvee\Phi$ and $M^\Phi_{s'}\models\bigwedge\susp_\agent\Phi$ state that $M^\Phi_{s'}\models\covers_\agent\Phi$, and therefore $M_s\models\F_\agent\covers_\agent\Phi$.

\bigskip

\noindent {\bf RKmulti}

Suppose that $M_s\models\F_\agent\covers_\agentb\Phi$. Therefore, there is a model $M'_t\simul_\agent M_s$ such that $M'_t\models\covers_\agentb\Phi$---let the accessibility relation for agent $\agentb$ in $M'$ be $R'_\agentb$. Expanding the definition, we have that for every $\phi\in\Phi$ there is some $u\in tR'_\agentb$ such that $M'_u\models\phi$. Also, because of {\bf back}, for every such $u\in tR'_\agentb$ there is some $v\in sR_\agentb$ such that $M'_u\simul_\agent M_v$. Combining these statements we have that for every $\phi\in\Phi$ there is some $v\in sR_\agentb$ such that $M_v\models\F_\agent\phi$, and thus $M_s\models\bigwedge\susp_\agentb\F_\agent\Phi$. However, as {\bf forth} also holds for agent $\agentb$, the $v\in sR_\agentb$ we could construct above are also {\em all} the states $v$ accessible from $s$. Therefore we also have $M_s\models\Box_\agentb\Vel\F_\agent\Phi$, so together we get $M_s \models \covers_\agentb \F_\agent \Phi$.

For the converse direction, suppose that $M_s\models\covers_\agentb\F_\agent\Phi$. From the definition of $\covers_\agentb$ it follows that 
  $M_s\models\bigwedge\susp_\agentb\F_\agent\Phi$. We now proceed in a similar way as in the case {\bf RK}. From $M_s\models\bigwedge\susp_\agentb\F_\agent\Phi$ it follows that for
  every $\phi\in \Phi$ there is some $t^\phi\in sR_\agentb$ such that
  $M_{t^\phi}\models\F_\agent\phi$. Thus, for each $\phi\in\Phi$,
  there is some model $N^\phi_{u^\phi}\simul_\agent M_{t^\phi}$, where $N^\phi = (S^\phi,R^\phi,V^\phi)$, such that $N^\phi_{u^\phi}\models\phi$. Define the model $M^\Phi = (S^\Phi, R^\Phi, V^\Phi)$ similar to the case {\bf RK}, except that: the roles of $\agent$ and $\agentb$ have been swapped, and the accessibility relation for all agents $\agentc$ different from $\agent$ and $\agentb$ is defined as that for $\agent$.
\[ \begin{array}{lcl} 
S^\Phi & = & \{s'\}\cup S \cup \bigcup_{\phi\in\Phi} S^\phi \\
R_\agentb^\Phi & = & \{(s',u^\phi)\ |\ \phi\in\Phi\}\cup R_\agentb \cup \bigcup_{\phi\in\Phi} R_\agentb^\phi  \\
R_\agentc^\Phi & = & \{(s',t)\ |\ (s,t)\in R_\agentc \}\cup R_\agentc \cup \bigcup_{\phi\in\Phi} R_\agentc^\phi \hfill \text{ for } \agentc \neq \agentb \\
V^\Phi(p) & = & \overline{\{s'\}} \cup V(p)\cup\bigcup_{\phi\in\Phi} V^\phi(p) \hspace{4cm} \hfill \text{ for } p\in\Atoms
\end{array} \]
where $\overline{\{s'\}} = \{s'\}$ if $s \in V(p)$ and else $\overline{\{s'\}} = \emptyset$ ($R_\agentc^\Phi$ also defines $R_\agent^\Phi$, namely for $\agentc=\agent$).
 
We can see that $M_s \lumis_\agent M^\Phi_{s'}$, via the relation ${\cal R}^\Phi = \{(s,s')\} \cup {\mathcal I} \cup \bigcup_{\phi\in\Phi}{\cal R}^\phi$ where ${\mathcal I}$ is the identity on $S$ and each ${\cal R}^\phi$ is the refinement relation corresponding to $M_{t^{\phi}} \lumis_\agent N^\phi_{u^\phi}$ (see also \cite{hales:2011}). Furthermore, for each $t\in s'R_\agentb^\Phi$ it is clear that $M^\Phi_t\bisim N^\phi_{u^\phi}$ for some $\phi$, and thus $M^\Phi_t\models\phi$, and so $M^\Phi_t\models\Vel\Phi$. Therefore $M^\Phi_{s'}\models\knows_\agentb\bigvee\Phi$. Finally, for each $\phi\in\Phi$ there is some $u^\phi\in s'R_\agentb^\Phi$ where $M^\Phi_{u^\phi}\models\phi$, so for each $\phi\in\Phi$ we have $M^\Phi_s\models\Diamond_\agentb\phi$, so we have $M^\Phi_{s'}\models\bigwedge\susp_\agentb\Phi$. Combined, $M^\Phi_{s'}\models\knows_\agentb\bigvee\Phi$ and $M^\Phi_{s'}\models\bigwedge\susp_\agentb\Phi$ state that $M^\Phi_{s'}\models\covers_\agentb\Phi$, and therefore $M_s\models\F_\agent\covers_\agentb\Phi$.

\bigskip

\noindent {\bf RKconj}

The direction $\F_\agent \Et_{\agentb\in\Group}\covers_\agentb \Phi^\agentb \imp\Et_{\agentb\in\Group} \F_\agent \covers_\agentb \Phi^\agentb$ is merely a more complex form of pattern $\F_\agent (\phi \et \psi) \imp (\F_\agent \phi \et \F_\agent \psi)$ which is derivable similar to $\Dia_a (\phi \land \psi) \imp (\Dia_a \phi \land \Dia_a \psi)$ in the modal logic $\logicK$, using the axiom {\bf R} in place of {\bf K}.

For the other direction, suppose that $M_s$ is such that $M_s \models
\bigwedge_{b \in B} \F_a \covers_b \Phi^b$, where $B \subseteq A$. We need to show that $M_s \models \F_a \bigwedge_{b \in B} \covers_b
\Phi^b$. To do this we follow the same strategy as for proving {\bf RK}: we
construct an $a$-refinement $N_t$ of $M_s$, and show that $N_t \models
\bigwedge_{b \in B} \covers_b \Phi^b$.

We begin by constructing the model $N_t$. Suppose that $a \in B$. Then we have
$M_s \models \F_a \covers_a \Phi^a$, and by {\bf RK} this implies that
$M_s \models \bigwedge \Dia_a \F_a \Phi^a$. We also have that for every $b \in B - \{a\}$, $M_s \models \F_a \covers_b \Phi^b$, and by {\bf RKmulti} this implies that $M_s \models \covers_b \F_a \Phi^b$, and by the definition of the cover operator, this implies that $M_s \models \bigwedge \Dia_b \F_a \Phi^b$. Hence for every $b \in B$ and $\phi \in \Phi^b$, we have that $\Dia_b \F_a \phi$. (In other words, for some big set of formulas $\Psi$ we have that $M_s \models\Et\Dia_b \F_a \Psi$.) At this stage it suffices to refer to the very similar construction in the soundness proof for axiom {\bf RK}, from which, similarly to there, it follows that $N_t \models \bigwedge_{b \in B} \covers_b \Phi^b$.

\bigskip

\noindent {\bf NecR} 

If $\phi$ is a validity, then it is satisfied by every model, so for any model $M_s$, $\phi$ is satisfied by every model $N_t\simul_\agent M_s$, and hence every model $M_s$ satisfies $\G_\agent\phi$.
\end{proof}

\begin{figure}
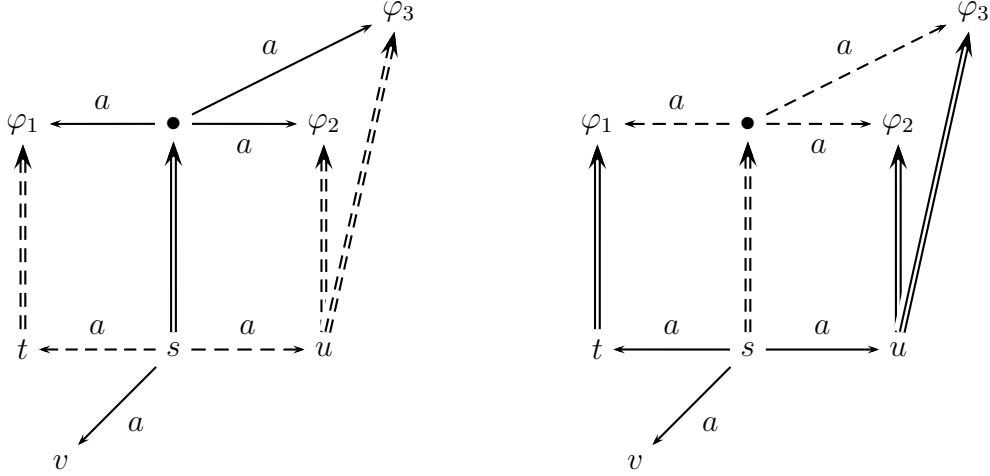

\psset{border=2pt, nodesep=4pt, radius=2pt, tnpos=a}
\pspicture(0,-2)(5.5,5)
$
\rput(0,0){\rnode{00}{t}}
\rput(2,0){\rnode{10}{s}}
\rput(4,0){\rnode{20}{u}}
\rput(0.5,-1.5){\rnode{0m}{v}}
\rput(0,3){\rnode{01}{\phi_1}}
\rput(2,3){\rnode{11}{\bullet}}
\rput(4,3){\rnode{21}{\phi_2}}
\rput(5,4.5){\rnode{2m}{\phi_3}}
\ncline[linestyle=dashed]{<-}{00}{10}  \naput{\agent}
\ncline[linestyle=dashed]{->}{10}{20}  \naput{\agent}
\ncline{->}{10}{0m}  \naput{\agent}
\ncline[doubleline=true]{->}{10}{11} 
\ncline{<-}{01}{11}  \naput{\agent}
\ncline{->}{11}{21}  \nbput{\agent}
\ncline{->}{11}{2m}  \naput{\agent}
\ncline[doubleline=true,linestyle=dashed]{->}{00}{01} 
\ncline[doubleline=true,linestyle=dashed]{->}{20}{21} 
\ncline[doubleline=true,linestyle=dashed]{->}{20}{2m} 
$
\endpspicture
\hspace{2cm}
\psset{border=2pt, nodesep=4pt, radius=2pt, tnpos=a}
\pspicture(0,-2)(5.5,5)
$
\rput(0,0){\rnode{00}{t}}
\rput(2,0){\rnode{10}{s}}
\rput(4,0){\rnode{20}{u}}
\rput(0.5,-1.5){\rnode{0m}{v}}
\rput(0,3){\rnode{01}{\phi_1}}
\rput(2,3){\rnode{11}{\bullet}}
\rput(4,3){\rnode{21}{\phi_2}}
\rput(5,4.5){\rnode{2m}{\phi_3}}
\ncline{<-}{00}{10}  \naput{\agent}
\ncline{->}{10}{20}  \naput{\agent}
\ncline{->}{10}{0m}  \naput{\agent}
\ncline[doubleline=true,linestyle=dashed]{->}{10}{11} 
\ncline[linestyle=dashed]{<-}{01}{11}  \naput{\agent}
\ncline[linestyle=dashed]{->}{11}{21}  \nbput{\agent}
\ncline[linestyle=dashed]{->}{11}{2m}  \naput{\agent}
\ncline[doubleline=true]{->}{00}{01} 
\ncline[doubleline=true]{->}{20}{21} 
\ncline[doubleline=true]{->}{20}{2m} 
$
\endpspicture
\caption{The interaction between refinement and modality involved in axiom {\bf RK}.}
\label{fig.rk}
\end{figure}

The soundness of axiom {\bf RK} is visualized in Figure \ref{fig.rk}. It depicts the interaction between refinement and modality involved in this axiom $\F_\agent\covers_\agent\Phi \eq \bigwedge\susp_\agent\F_\agent\Phi$, for the case that $\Phi = \{ \phi_1, \phi_2, \phi_3 \}$. The single lines are modal accessibility, and the double lines the refinement relations. The solid lines are given, and the dashed lines are required. Accessibility relations for other agents than $\agent$ are omitted. The picture on the left depicts the implication from left to right in the axiom, and the picture on the right depicts the implication from right to left. Note that the states satisfying $\phi_2$ and $\phi_3$ have the same origin $u$ in $M$---the typical sort of duplication (resulting in non-bisimilar states) allowed when having {\bf back} but not {\bf forth}. Apart from $u$ and $t$, state $s$ in $M$ has yet another accessible state $v$, that does not occur in the refinement relation: the other typical sort of thing when having {\bf back} but not {\bf forth}. Therefore, on the right side of the equivalence in axiom {\bf RK} we only have $\bigwedge\susp_\agent\F_\agent\Phi$ and we cannot guarantee that $\Box_\agent \Vel \F_\agent \Phi$ also follows from the left-hand side. 

The axiom {\bf RKmulti}, defined as $\F_\agent\covers_\agentb\Phi\eq\covers_\agentb\F_\agent\Phi$ for $\agent\neq\agentb$, says that refinement with respect to one agent does not interact with the modalities (the uncertainty, say) for another agent: the operators $\covers_\agentb$ and $\F_\agent$ simply commute. This in contrast to the axiom {\bf RK} where on the right-hand side a construct $\Box_\agent\Vel\F_\agent\Phi$ is `missing', so to speak. If it had been $\Box_\agent\Vel\F_\agent\Phi \et \bigwedge\susp_\agent\F_\agent\Phi$, then we would have had $\covers_\agent \F_\agent \Phi$, as in {\bf RKmulti} but with $\agent=\agentb$. 

The axioms {\bf RK} and {\bf RKmulti} are different, because in an $\agent$-refinement the condition {\bf forth} is not required, whereas for other agents $\agentb$ {\bf forth} is required. Given some refinement wherein we have a cover of $\Phi$, so that at least one of $\Phi$ is necessary (the $\F_\agent\covers_\agent\Phi$ bit), for each of the covered states we can trace an origin before the refinement, because of {\bf back}. But there may be more originally accessible states, so whatever holds in those origins, although it is all possible, is not necessary. So we have $\bigwedge\susp_\agent\F_\agent\Phi$, but we do not have $\Box_\agent\Vel\F_\agent\Phi$. In contrast, when the agents are different, {\bf back} and {\bf forth} must hold for agent $\agentb$ in a refinement $\succeq_\agent$ witnessing the operator $\F_\agent$: for an $\agent$-refinement, {\bf back} and {\bf forth} must hold for all agents $\agentb \neq \agent$. Figure \ref{fig.rkmulti} should further clarify the issue---compare this to Figure \ref{fig.rk}. The main difference between the figures is that there cannot now be yet another state $v$ accessible from $s$ but not `covered' as the origin of one of the refined states. In Figure \ref{fig.rk} what holds in $t$ and $u$ is not necessary for $\agent$, but in Figure \ref{fig.rkmulti} what holds in $t$ and $u$ is necessary for $\agentb$.

\begin{figure}
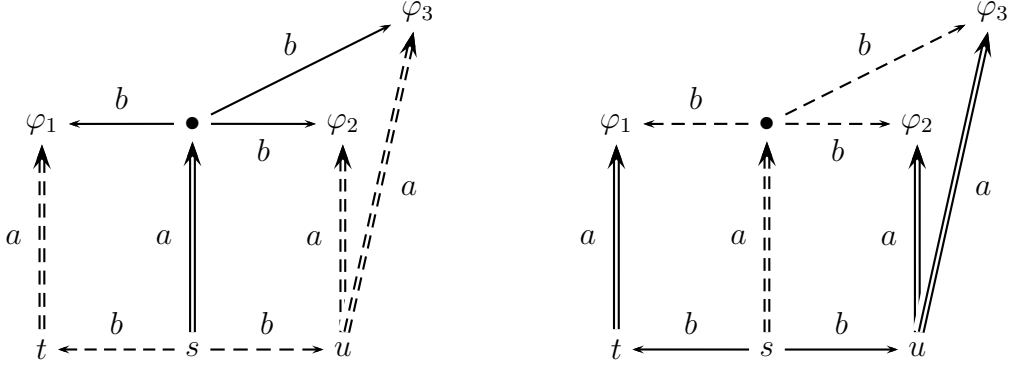

\psset{border=2pt, nodesep=4pt, radius=2pt, tnpos=a}
\pspicture(0,-.5)(5.5,5)
$
\rput(0,0){\rnode{00}{t}}
\rput(2,0){\rnode{10}{s}}
\rput(4,0){\rnode{20}{u}}
\rput(0,3){\rnode{01}{\phi_1}}
\rput(2,3){\rnode{11}{\bullet}}
\rput(4,3){\rnode{21}{\phi_2}}
\rput(5,4.5){\rnode{2m}{\phi_3}}
\ncline[linestyle=dashed]{<-}{00}{10}  \naput{\agentb}
\ncline[linestyle=dashed]{->}{10}{20}  \naput{\agentb}
\ncline[doubleline=true]{->}{10}{11} \naput{\agent}
\ncline{<-}{01}{11} \naput{\agentb}
\ncline{->}{11}{21} \nbput{\agentb}
\ncline{->}{11}{2m} \naput{\agentb}
\ncline[doubleline=true,linestyle=dashed]{->}{00}{01} \naput{\agent}
\ncline[doubleline=true,linestyle=dashed]{->}{20}{21} \naput{\agent}
\ncline[doubleline=true,linestyle=dashed]{->}{20}{2m} \nbput{\agent}
$
\endpspicture
\hspace{2cm}
\psset{border=2pt, nodesep=4pt, radius=2pt, tnpos=a}
\pspicture(0,-.5)(5.5,5)
$
\rput(0,0){\rnode{00}{t}}
\rput(2,0){\rnode{10}{s}}
\rput(4,0){\rnode{20}{u}}
\rput(0,3){\rnode{01}{\phi_1}}
\rput(2,3){\rnode{11}{\bullet}}
\rput(4,3){\rnode{21}{\phi_2}}
\rput(5,4.5){\rnode{2m}{\phi_3}}
\ncline{<-}{00}{10}  \naput{\agentb}
\ncline{->}{10}{20}  \naput{\agentb}
\ncline[doubleline=true,linestyle=dashed]{->}{10}{11} \naput{\agent}
\ncline[linestyle=dashed]{<-}{01}{11}  \naput{\agentb}
\ncline[linestyle=dashed]{->}{11}{21}  \nbput{\agentb}
\ncline[linestyle=dashed]{->}{11}{2m}  \naput{\agentb}
\ncline[doubleline=true]{->}{00}{01} \naput{\agent}
\ncline[doubleline=true]{->}{20}{21} \naput{\agent}
\ncline[doubleline=true]{->}{20}{2m} \nbput{\agent}
$
\endpspicture
\caption{The interaction between refinement and modality involved in axiom {\bf RKmulti}.}
\label{fig.rkmulti}
\end{figure}

\subsection{Example derivations} \label{sec.axiomex}

In these examples we also use `substitution of equivalents', see Proposition~\ref{lem:substeq}, ahead.

\begin{example} \label{ex.gettingtok}
$\proves \susp_\agent \T \imp \F_\agent (\susp_\agent \T \et (\Box_\agent \atom \vel\Box_\agent \neg \atom))$
\end{example}
In an epistemic setting, where $\Box_\agent \atom$ means that the agent knows $\atom$, and where (in $S5$ models) the condition $\susp_\agent \T$ is always satisfied, this validity expresses that the agent can always find out the truth about $\atom$: if true, announce $\atom$ to the agent (and announcement is a model restriction, and therefore a refinement), after which $\atom$ is known by the agent to be true, and if false, announce that $\atom$ is false, after which $\atom$ is known to be false. This validity is indeed also a theorem of $\logicRML$. For that, it suffices to derive the equivalent $\susp_\agent \T \imp \F_\agent(\covers_\agent \{p\} \vel\covers_\agent \{\neg p\})$. In some cases several deductions have been combined into single statements, but this is restricted to cases of well-known modal theorems.

\[\begin{array}{ll}
\proves \susp_\agent \T \eq \susp_\agent(p \vel \neg p) & {\bf Prop},{\bf NecK},{\bf K} \\
\proves \susp_\agent(p \vel \neg p) \eq (\susp_\agent p \vel \susp_\agent \neg p) & {\bf Prop},{\bf NecK},{\bf K} \\
\proves \susp_\agent p \imp \F_\agent \covers_\agent\{ p\} & \text{See below} \\
\proves \susp_\agent \neg p \imp \F_\agent\covers_\agent\{ \neg p\} & \text{See below} \\
\proves \susp_\agent p \imp \F_\agent (\covers_\agent\{ p\} \vel \covers_\agent\{ \neg p\}) & {\bf Prop}, {\bf NecR},{\bf R} \\
\proves \susp_\agent \neg p \imp \F_\agent (\covers_\agent\{ p\} \vel \covers_\agent\{\neg p\}) \hspace{2cm}  & {\bf Prop}, {\bf NecR}, {\bf R} \\
\proves \susp_\agent \T \imp \F_\agent (\covers_\agent\{ p \}\vel \covers_\agent\{ \neg p \}) & {\bf Prop}, {\bf MP}
\end{array}\]
Lines 3 and 4 of the derivation require the following derivation,
where $\phi$ is a propositional formula (i.e., $\phi\in\lang_0$).
\[\begin{array}{ll}
\proves \phi \eq \F_\agent \phi & \text{Proposition~\ref{prop.prop}, ahead} \\
\proves \susp_\agent \phi \eq \susp_\agent\F_\agent\phi & {\bf Prop},{\bf NecK},{\bf K} \\
\proves \susp_\agent \phi \eq \F_\agent\covers_\agent\{\phi\} \hspace{2cm} & {\bf RK} [\Phi = \{\phi\}]
\end{array}\]

\begin{example}
$\proves (\susp_\agent \atom \et \susp_\agentb \atom \et \susp_\agent \neg \atom \et \susp_\agentb \neg \atom) \imp \F_\agent (\know_\agent \atom \et \neg \know_\agentb \atom)$
\end{example}

Consider the informative development described in Example~\ref{ex-1}: given an initial information state wherein agents $\agent$ and $\agentb$ consider either value of $\atom$ possible, $\agent$ can be informed such that afterwards $\agent$ believes that $\atom$ but not $\agentb$. This theorem formalizes that. In the following, let $\phi$ be $(\susp_\agent \atom \et \susp_\agentb \atom \et \susp_\agent \neg \atom \et \susp_\agentb \neg \atom)$.
$$
\begin{array}{ll}
\proves \phi \imp
\Dia_\agent p \land \Dia_\agentb \neg p & \text{\bf Prop}\\
\proves \phi \imp
\Dia_\agent p \land \covers_\agentb \{\neg p, \top\} & \text{Definition of $\covers$}\\
\proves \phi \imp
\Dia_\agent \neg \neg p \land \covers_\agentb \{\neg \neg \neg p, \neg \neg \top\} &
\text{\bf Prop}\\
\proves \phi \imp
\Dia_\agent \neg \G_\agent \neg p \land \covers_\agentb \{\neg \G_\agent \neg \neg
p, \neg \G_\agent \neg \top\} \hspace{2cm}  & \text{\bf RProp}\\
\proves \phi \imp
\Dia_\agent \F_\agent p \land \covers_\agentb \{\F_\agent \neg
p, \F_\agent \top\} & \text{Definition of $\F$}\\
\proves \phi \imp
\F_\agent \covers_\agent \{p\} \land \covers_\agentb \{\F_\agent \neg
p, \F_\agent \top\} & \text{\bf RK}\\
\proves \phi \imp
\F_\agent \covers_\agent \{p\} \land \F_\agent \covers_\agentb \{\neg
p, \top\} & \text{\bf RKmulti}\\
\proves \phi \imp
\F_\agent (\covers_\agent \{p\} \land \covers_\agentb \{\neg
p, \top\}) & \text{\bf RKconj} \\
\proves \phi \imp
\F_\agent (\knows_\agent p \land \Dia_\agent p \land \Dia_\agentb \neg
p) & \text{Definition of $\covers$}\\
\proves \phi \imp
\F_\agent (\knows_\agent p \land \Dia_\agentb \neg
p) & \text{\bf Prop}\\
\proves \phi \imp
\F_\agent (\knows_\agent p \land \neg\knows_\agentb p) & \text{Definition of $\Dia$}\\
\end{array}
$$

\subsection{Completeness}\label{completeness}

Completeness is shown by a fairly but not altogether straightforward reduction argument: every formula in refinement modal logic is equivalent to a formula in modal logic. So it is a theorem, if its modal logical equivalent is a theorem. In the axiomatization $\axiomRML$ we can observe that all axioms involving refinement operators $\F$ are equivalences, except for {\bf R}; however, $\F_\agent (\phi\vel\psi) \eq \F_\agent\phi\vel\F_\agent\psi$ is a derivable theorem. This means that by so-called `rewriting' we can push the $\F$ operators further inward into a formula, until we reach some expression $\F \phi$ where $\phi$ contains no refinement operators. Now we come to the less straightforward part. Because there is a hitch: there is no general way to push a $\F$ beyond a negation (or, for that matter, beyond a conjunction). For that, we use another trick, namely that all modal logical formulas are equivalent to formulas in the cover logic syntax, and that all those are equivalent to formulas in disjunctive form (see the introduction) in cover logic. Using that, once we reached some innermost $\F \phi$ where $\phi$ contains no refinement operators, we can continue pushing that refinement operator downward until it binds a propositional formula only, and disappears in smoke because of the {\bf RProp} axiom. Then, iterate this. All $\F$ operators have disappeared in smoke. We have a formula in modal logic.


For a smooth argument we first give some general results, after which we apply the reduction argument and demonstrate completeness. 

\begin{definition}[Substitution of equivalents]
An axiomatization satisfies {\em substitution of equivalents} if the following holds. Let  $\phi_1,\phi_2,\phi_3\in\lang$ and $\atom\in\Atoms$. If $\proves \phi_1 \eq \phi_2$ then $\proves \phi_3[\phi_2\backslash\atom] \eq\phi_3[\phi_1\backslash\atom]$. 
\end{definition} 
\begin{proposition}\label{lem:substeq}
The axiomatization $\axiomRML$ satisfies substitution of equivalents.
\end{proposition}
\begin{proof}
This can be shown by induction on $\phi_3$. All cases are standard. The case $\Box_\agent\phi$ is shown by using an inductive hypothesis $\proves \phi[\phi_2\backslash\atom] \eq\phi[\phi_1\backslash\atom]$ and then successively applying {\bf NecK}, {\bf K}, and some elementary tautologies and applications of {\bf MP}. (The required pattern is: from $\proves x \imp y$, to $\proves \Box (x \imp y)$, to $\proves \Box x \imp \Box y$. Then, similarly, for the other direction of the equivalence $x \eq y$. Then, some more propositional steps to wind it up.) Whereas the case $\G_\agent\phi$ is shown with the same inductive hypothesis but applying {\bf NecR} and {\bf R} instead of {\bf NecK} and {\bf K}. \end{proof}

\begin{proposition} \label{prop.ggg} \
\begin{enumerate}
\item $\proves \G_\agent (\phi\et\psi) \eq \G_\agent\phi\et\G_\agent\psi$ \label{zxcvzxcv}
\item $\proves \F_\agent (\phi\vel\psi) \eq \F_\agent\phi\vel\F_\agent\psi$ \label{fiets}
\item $\proves \F_\agent (\phi\et\psi) \imp \F_\agent\phi\et\F_\agent\psi$
\end{enumerate}
\end{proposition}
\begin{proof}
Item 1.\ can be easily derived from {\bf R}, {\bf NecR} and {\bf MP},  similarly to the way that in modal logic we derive $\proves \Box(\phi\et\psi) \eq \Box\phi\et\Box\psi$. Item 2.\ is the dual of item 1.\ and requires mere propositional reasoning. Item 3.\ can be derived using the tautologies $\phi\et\psi \imp \phi$ and $\phi\et\psi \imp \psi$, respectively, propositional reasoning, and {\bf R}. (Alternatively, for Item 3., we can think of deriving its dual, with the crucial steps in the derivation that $\phi \imp \phi \vel \psi$ is a tautology, from which with {\bf R} and {\bf MP} we get $\G\phi \imp \G(\phi \vel \psi)$.)
\end{proof}

\begin{proposition} \label{prop.prop} \ 
\begin{enumerate}
\item $\proves \G_\agent \phi \eq \phi$ for all propositional $\phi$. \label{eeneen}
\item $\proves \F_\agent \phi \eq \phi$ for all propositional $\phi$. \label{tweetwee}
\end{enumerate}
\end{proposition}
\begin{proof}
We show $\proves \G_\agent \phi \eq \phi$ for all propositional $\phi$.\footnote{Of course we do not have for all $\phi\in\lang_\G$ that $\proves \G_\agent \phi \eq \phi$. But we then still have $\proves \G\phi \imp \phi$, or, dually, $\proves \phi \imp \F\phi$. This can be easily shown by induction on the disjunctive form structure of a formula.}  The proof of $\proves \F_\agent \phi \eq \phi$ for all propositional $\phi$ is similar. For convenience in the proof we omit the agent label and write $\G$.

We first show $\proves \phi \imp \G\phi$. Assume that $\phi$ is in disjunctive normal form (i.e., for propositional logic, different from the {\em disjunctive form}, {\em df}, often used in this work). Formula $\phi$ therefore has the form $\Vel_{\gamma \in \Gamma}$, where each formula $\gamma$ is a conjunction of atoms or their negation, for which we write, slightly abusing the language, $\gamma = \Et_{\atom\in\gamma}\overline{\atom}$ --- where $\overline{\atom} = \atom$ if $\atom$ is a conjunct of $\gamma$ and $\overline{\atom} = \neg\atom$ if $\neg\atom$ is a conjunct of $\gamma$. We now get the following. We omit trivial steps of chaining implications and applying {\bf MP}. For readability we assume the `$\phi\imp$' part in some derived formulas.

\[ \begin{array}{llll}
\proves \phi & \imp & \Vel_{\gamma \in \Gamma} \Et_{\atom\in\gamma}\overline{\atom} & \text{DNF of } \phi, {\bf Prop} \\
\proves \dots && \Vel_{\gamma \in \Gamma} \Et_{\atom\in\gamma}\G\overline{\atom} \hspace{2cm} & {\bf RProp} \\
\proves \dots && \Vel_{\gamma \in \Gamma} \G\Et_{\atom\in\gamma}\overline{\atom} & \text{{\bf R}, {\bf NecR}, and Prop.\ \ref{prop.ggg}.\ref{zxcvzxcv} } (\G (\phi\et\psi) \eq \G\phi\et\G\psi) \\
\proves \dots && \G\Vel_{\gamma \in \Gamma} \Et_{\atom\in\gamma}\overline{\atom} & \text{{\bf R}, {\bf NecR}, and tautology } \phi \imp \phi\vel\psi \\
\proves \phi & \imp & \G\phi & \text{DNF of } \phi
\end{array} \]

For the converse direction we convert $\phi$ to the conjunctive normal form for propositional formulas, i.e., $\phi$ is equivalent to $\Et_{\gamma \in \Gamma} \Vel_{\atom\in\gamma}\overline{\atom}$ (where we now write $\overline{\atom} = \atom$ if $\atom$ is a {\em disjunct}---not conjunct---of $\gamma$ and $\overline{\atom} = \neg\atom$ if $\neg\atom$ is a disjunct of $\gamma$).
\[ \begin{array}{llll}
\proves \G\phi & \imp & \G\Et_{\gamma \in \Gamma} \Vel_{\atom\in\gamma}\overline{\atom} \hspace{2cm} & \text{CNF of } \phi, {\bf Prop}, {\bf NecR}, {\bf R} \\
\proves \dots && \Et_{\gamma \in \Gamma} \G\Vel_{\atom\in\gamma}\overline{\atom} & \text{Prop.\ \ref{prop.ggg}.\ref{zxcvzxcv}} \\
\proves \dots && \Et_{\gamma \in \Gamma} \Vel_{\atom\in\gamma}\overline{\atom} & * \\
\proves \G\phi & \imp & \phi & \text{CNF of } \phi
\end{array} \]
We show why * holds by outlining the method and giving an example: write the conjunct $\Vel_{\atom\in\gamma}\overline{\atom}$ in implicative fashion, e.g., instead of $p \vel q \vel \neg r \vel s$ we write $\neg p \imp \neg q \imp r \imp s$. Then, applying {\bf NecR} and {\bf R} and {\bf MP} repeatedly, we get first $\G (\neg p \imp \neg q \imp r \imp s)$ and then $\G \neg p \imp \G \neg q \imp \G r \imp \G s$. Then, applying {\bf RProp}, we get $\neg p \imp \neg q \imp r \imp s$, in other words, we have $p \vel q \vel \neg r \vel s$ back.
\end{proof}

\begin{proposition} \label{prop.prop2} \ 
$\proves (\phi \et \F_\agent\psi) \eq \F_\agent(\phi\et\psi)$ for all propositional $\phi$ (and any $\psi\in\lang_\G$).
\end{proposition}
\begin{proof}
Proposition \ref{prop.ggg} demonstrated that $\F_\agent(\phi\et\psi) \imp \F_\agent \phi \et \F_\agent \psi$ from which, using Proposition \ref{prop.prop}.\ref{tweetwee}, also follows $\phi \et \F_\agent \psi$. For the other direction we first derive $(\G_\agent\phi \et \F_\agent\psi) \imp \F_\agent(\phi\et\psi)$ by propositional means and applications of {\bf Nec} and {\bf R}. This goes as follows. For convenience of applying the available axioms, instead of $(\G_\agent\phi \et \F_\agent\psi) \imp \F_\agent(\phi\et\psi)$ use the equivalent form $\G_\agent\neg(\phi\et\psi) \imp \G_\agent\phi \imp \G_\agent\neg\psi$. Now we observe that $\neg(\phi\et\psi) \imp \phi \imp \neg\psi$ is a tautology and therefore derivable, applying {\bf NecR} gets us $\G_\agent( \neg(\phi\et\psi) \imp \phi \imp \neg\psi)$ and successively applying {\bf R} gets us $\G_\agent\neg(\phi\et\psi) \imp \G_\agent\phi \imp \G_\agent\neg\psi$. Then, finally, we use that $\G_\agent \phi \eq \phi$ (Proposition \ref{prop.prop}.\ref{eeneen}) and thus get $(\phi \et \F_\agent\psi) \imp \F_\agent(\phi\et\psi)$.
\end{proof}

We now first show that every $\lang_\G$ formula is logically equivalent to a $\lang$ formula. We then show that if the latter is a theorem in {\bf K}, the former is a theorem in $\axiomRML$.

\begin{proposition}\label{lem:express}
Every formula of $\lang_\G$ is logically equivalent to a formula of $\lang$.
\end{proposition}

\begin{proof} 
Given a formula $\psi\in\lang_\G$, we prove by induction on the number of the occurrences of $\F_\agent$ in $\psi$ (for any $\agent\in\Agents$) that it is equivalent to an $\F_\agent$-free formula, and therefore to a formula $\phi\in\lang$, the standard modal logic. The base is trivial. Now assume $\psi$ contains $n+1$ occurrences of $\F_\agent$-operators for some $\agent\in\Agents$ (so these may be refinement operators for different agents). Choose a subformula of type $\F_\agent \phi$ of our given formula $\psi$, where $\phi$ is $\F_\agentb$-free for any $\agentb\in\Agents$ (i.e.\, choose an innermost $\F_\agent$). Let $\phi'$ be a disjunctive formula that is equivalent to $\phi$. We prove by induction on the structure of $\phi'$ that $\F_\agent \phi'$ is logically equivalent to a formula $\chi$ without $\F_\agent$. There are two cases:
\begin{itemize}
\item $\F_\agent (\phi \vee \psi)$;
\item $\F_\agent (\phi_0 \et \Et_{\agentb\in\Group}\covers_\agentb \Phi^\agentb)$ \hfill where $\phi_0$ is propositional, $\Group\subseteq\Agents$, and each $\Phi^\agentb$ a set of {\em df}s.
\end{itemize}
In the first case, apply Proposition \ref{prop.ggg}.\ref{fiets}, we get $\F_\agent \phi \vee \F_\agent \psi$, and then apply induction. In the second case, if $\Group = \emptyset$ we use that $\F_\agent\phi_0 \eq \phi_0$ (Proposition \ref{prop.prop}.\ref{tweetwee}). If $\Group \neq \emptyset$, then from Proposition \ref{prop.prop2} follows that this is equivalent to $\phi_0 \et \F_\agent \Et_{\agentb\in\Group}\covers_\agentb \Phi^\agentb$, and we further reduce the right conjunct with one of the axioms {\bf RK} (if $\Group = \{\agent\}$), {\bf RKmulti} (if $\Group = \{\agentb\}$ with $\agentb\neq\agent$), or {\bf RKconj} (if $|\Group| > 1$), and apply induction again.

Thus we are able to push the refinement operators deeper into the formula until they eventually reach a propositional formula, at which point they disappear and we are left with the required $\F$-free formula $\chi$ that is equivalent to $\F\phi$. Replacing $\F \phi'$ by $\chi$ in $\psi$ gives a result with one less $\F$-operator, to which the (original) induction hypothesis applies. 
\end{proof}

\begin{proposition}\label{lem:expressb}
Let $\phi\in\lang_\G$ be given and $\psi\in\lang$ be equivalent to $\phi$. If $\psi$ is a theorem in {\bf K}, then $\phi$ is a theorem in $\axiomRML$.
\end{proposition}

\begin{proof} Given a $\phi\in\lang_\G$, Proposition~\ref{lem:express} gives us an equivalent $\psi\in\lang$. Assume that $\psi$ is a theorem in {\bf K}. We can extend the derivation of $\psi$ to a derivation of
$\phi$ by observing that all steps used in Proposition~\ref{lem:express} are not merely logical but also provable equivalences --- where we also apply Proposition \ref{lem:substeq} of substitution of equivalents.
\end{proof}

\begin{theorem}\label{theo:SandC}
The axiom schema $\axiomRML$ is sound and complete for the logic $\logicRML$.
\end{theorem}

\begin{proof} The soundness proof is given in Theorem~\ref{theo:soundness}, so
we are left to show completeness. Suppose that $\phi\in\lang_\G$ is valid:
$\models\phi$. Applying Lemma~\ref{lem:express} we know that there is
some equivalent formula $\psi\in\lang$, i.e., not containing any refinement
operator. As $\phi$ is valid, from that and the validity
$\phi\eq\psi$ it follows that $\psi$ is also valid in refinement modal logic, and therefore also valid in the logic {\bf K} (note that the model class is the same). From the completeness of {\bf K} it follows that $\psi$ is derivable, i.e.\ it is
a theorem. From Proposition~\ref{lem:expressb} it follows that $\phi$ is a
theorem.  \end{proof}

\subsection{The single-agent case} \label{sec.singleax}

The axiomatization for the single-agent case is the unlabelled version of $\axiomRML$, minus the axioms {\bf RKmulti} and {\bf RKconj}.\footnote{It is clear that axiom {\bf RKmulti} is not needed in the single-agent case, as this is for different agents. But axiom {\bf RKconj} is also not necessary in the single-agent case. We recall that $\covers_\agent \Phi \et \covers_\agent \Psi$ is equivalent to $\covers_\agent ((\Phi\et\Vel\Psi)\union (\Psi\et\Vel\Phi))$, see page \pageref{sec.coverlogic}. So, we can assume that there are no conjunctions of cover formulas in the single-agent case.} The single-agent axiomatization was presented in \cite{hvdetal.felax:2010}. The completeness proof there is (slightly) different from the multi-agent case of the proof here. In \cite{hvdetal.felax:2010} it is used that every refinement modal logical formula is equivalent to a formula in cover logic with the special syntax $\phi \ ::= \ \bot \mid \top \mid \phi \vee \phi \mid \atom \wedge \phi \mid \neg \atom \wedge \phi \mid \covers \{ \phi, \dots, \phi \}$ \cite{bilkovaetal:2008,kupkeetal:2008}, plus induction on that form. (This syntax is of course very `disjunctive formula like'.) That proof was suggested by Yde Venema, as a shorter alternative to the proof with disjunctive forms.

\subsection{Refinement epistemic logic}

Refinement modal logic $\logicRML$ is presented with respect to the class of all models. As mentioned in Section \ref{subsec.syntaxrml}, by restricting the class of models that the logic is interpreted over, we may associate different meanings with the modalities. 
For example, the epistemic logic {\bf S5}, a.k.a.\ the logic of knowledge, is interpreted over  the model class ${\mathcal S}5$, and the  logic of belief {\bf KD45} is interpreted over the class ${\mathcal KD}45$. Given any class of models $\mathcal{C}$, the semantic interpretation of $\G$ is given by:
$$
M_\state \models \G_\agent \phi  \text{ iff for all } 
M'_{\state'}\in \mathcal{C}:\ M_\state \lumis_\agent M'_{\state'} \text{ implies } M'_{\state'} \models \phi.
$$
Thus we can consider various {\em refinement epistemic logics}. Although $\F \Box \bot$ is a validity in $\logicRML$ (just remove all access) it is not so in the refinement logic of knowledge, interpreted on ${\mathcal S5}$ models, because seriality of models must be preserved in every refinement. And therefore it is also not valid in the refinement logic of belief. 

Our axiomatization $\axiomRML$ may not be sound for more restricted model classes. Let us consider the single-agent case, and the axiom \[ {\bf RK} \qquad \F\covers\Phi\eq\bigwedge\susp\F\Phi. \] For example, in $\mathcal{S}5$ we have that $\F\covers\{\knows p,\neg\knows p\}$ is inconsistent, but that $\susp\F\knows p\land\susp\F\neg\knows p$ is consistent: you do not consider an informative development possible after which you both know and don't know $p$ at the same time. Therefore, axiom {\bf RK} is invalid for that class.

The axioms replacing {\bf RK} in refinement logic of knowledge and refinement logic of belief are, respectively:
\[ {\bf RS5}\qquad \F\covers\Phi\eq(\bigvee\Phi\land\bigwedge\susp\Phi) , \]
and, for $\Phi\neq\emptyset$,
\[ {\bf RKD45}\qquad \F\covers\Phi\eq\bigwedge\susp\Phi , \]
where $\Phi$ is a set of purely propositional formulas. Now if apart from {\bf RS5} we also add the usual ${\bf S5}$ axioms {\bf T}, {\bf 4}, and {\bf 5}, we have a complete axiomatization for the refinement logic of knowledge. In the case of the refinement logic of belief, we add axioms {\bf D} (for seriality), {\bf 4}, and {\bf 5} and {\bf RKD45} to get a complete axiomatization. For details, see \cite{halesetal:2011}.

A study of how various classes of models affect the properties of
bisimulation quantified logics is given in \cite{french:2006}. Refinement epistemic logics are investigated in \cite{halesetal:2011,hales:2011}. In \cite{hales:2011} a multi-agent $KD45$ axiomatization is also reported. (For multi-agent $S5$, see `Recent results' in Section~\ref{sec.conclusions}.)

\section{Axiomatization $\axiomRML^\mu$}\label{muAxio}
\label{sec-FEL-mu}

In this section we give the axiomatization for refinement modal $\mu$-calculus. We restrict ourselves to {\em single-agent} refinement modal $\mu$-calculus. The axiomatization is an extension of the (single-agent) axiomatization $\axiomRML$ for refinement modal logic. 

We recall the definition of modal $\mu$-calculus in the technical introductory Section \ref{sec.tech}. In \cite[Lemma 2.43]{french:2006} a bisimulation quantifier characterization of fixed points is given. The characterization employs the universal modality $\blacksquare$ which quantifies over all states in the model. Let $\lang_{\bqall\blacksquare}$ be the language of bisimulation quantified modal logic with $\blacksquare$ as well. First, observe that this impacts the semantics of bisimulation quantification. For two models to be bisimilar, it must now also be the case that every state in one model is bisimilar to a state in the other.

We can inductively define a truth-preserving translation $t: \lang^\mu \mapsto \lang_{\bqall\blacksquare}$. The crucial clauses are those for the fixed-point operators. The atoms $p$ introduced in the translation are required not to occur in $\phi$.
\[\begin{array}{l}
t(\nu x.\phi) \text{ is equivalent to } \bqis p (p \land \blacksquare(p \imp t(\phi[p\backslash x]))) \\
t(\mu x.\phi) \text{ is equivalent to } \bqall p(\blacksquare(t(\phi[p\backslash x])\imp p)\imp p) 
\end{array} \]
The first equation captures the intuition of a greatest fixed point as a least upper bound of the set of states that are postfixed points of $\phi$, whereas the second equation captures a least fixed point as the greatest lower bound of the set of states that are prefixed points of $\phi$. From \cite{dagostinoetal:2000} we know that bisimulation
quantifiers are also expressible in the modal $\mu$-calculus, and thus these equivalences also hold in the modal $\mu$-calculus.

\medskip

Having these tools for modal $\mu$-calculus at our disposition, let us now apply them in  {\em refinement} modal $\mu$-calculus. In order to demonstrate the soundness of the axiomatization defined below, we need to expand the relativization $\bullet^\atom: \lang_\bqall \imp \lang_\bqall$ (Definition \ref{def-relativization}), single-agent version, to a version $\bullet^\atom: \lang_{\bqall\blacksquare} \imp \lang_{\bqall\blacksquare}$ by including a clause for the universal modality: \[ (\blacksquare\phi)^\atom = \blacksquare\phi^\atom \] 
Employing that expanded relativization we can expand the translation $t: \lang_\G \imp \lang_\bqall$ (Definition \ref{def.tsjakka}) to a translation \[ t: \lang_\G^\mu \imp \lang_{\bqall\blacksquare} \] by adding the two clauses above for fixed points (this explains why we also wrote $t(\bullet)$ there). This translation $t$ remains truth-preserving (due to Proposition~\ref{prop-ref-and-bisimquant} and \cite[Lemma 2.43]{french:2006}). We recall the crucial interaction of the translation and the relativization, namely that $t(\F\phi)$ is equivalent to $\bqis \atom \ t(\phi)^\atom$. The translation plays an important role in the soundness proof: axioms are shown to be sound by showing that their translations are valid.

\begin{definition}[axiomatization $\axiomRML^\mu$] \label{def.refmu}  
The axiomatization $\axiomRML^\mu$ is a substitution schema of the
(single-agent) axioms and rules of $\axiomRML$ along with the axiom and rule for the modal $\mu$-calculus:
$$
\begin{array}{rl} 
{\bf F1} & \phi[\mu x.\phi\backslash x]\imp \mu x.\phi\\
{\bf F2} & \text{From } \phi[\psi\backslash x]\imp\psi \text{ infer } \mu x.\phi\imp \psi
\end{array}
$$   
and two new interaction axioms:
$$
\begin{array}{rl}
{\bf R^\mu} & \G\mu x.\phi \eq \mu x. \G\phi\ {\rm where }\ \phi\ {\rm is\ a\ {\it df}} \\
{\bf R^\nu} & \G\nu x.\phi \eq \nu x. \G\phi\ {\rm where }\ \phi\ {\rm is\ a\ {\it df}}
\end{array}
$$
\end{definition}
For single-agent $\axiomRML$, see Definition~\ref{def:FEL} and Section~\ref{sec.singleax}. We recall that single-agent $\axiomRML$ does not contain the axioms {\bf RKmulti} and {\bf RKconj}.

We emphasize that the interaction axioms have the important associated condition that the refinement quantification will only commute with a fixed-point operator if the fixed-point formula is a disjunctive formula.

\subsection{Soundness}\label{sect:soundness-mu}


The soundness proofs of Section~\ref{sect:soundness} still apply and
the soundness of {\bf F1} and {\bf F2} are well known \cite{arnoldetal:2001}, so we are left to show that {\bf R$^\mu$} and {\bf R$^\nu$} are sound. In the proof we use the characterization of refinement quantification in terms of bisimulation quantification and relativization that was established in Proposition~\ref{prop-ref-and-bisimquant}. We will also use the characterization of both fixed points in terms of bisimulation quantification as in the previous subsection. 
\begin{theorem}\label{theo:soundness-mu}
The axioms {\bf R$^\mu$} and {\bf R$^\nu$} are sound.
\end{theorem}
\begin{proof}
The proof consists of two cases, {\bf R$^\mu$} and {\bf R$^\nu$}.

\bigskip

\noindent {\bf Case {\bf R$^\mu$}}

It is more convenient in this proof to reason about the axiom in its contrapositive form: $\F\nu x.\phi\leftrightarrow\nu x.\F\phi$. The proof demonstrates that $t(\F\nu x.\phi)$ is equivalent to $t(\nu x.\F\phi)$ in bisimulation quantified logic (with the universal modality). Using the translation and relativization equivalences above we have that, for any $\phi\in\lang_\G$:
\[
\begin{array}{lcl}
t(\F\nu x.\phi) &\Eq& \bqis p \ t(\nu x.\phi)^p  \\
&\Eq& \bqis p (\bqis q (q \et \blacksquare (q \imp t(\phi[q\backslash x]))))^p  \\
&\Eq& \bqis p \bqis q (q \et (\blacksquare (q \imp t(\phi[q\backslash x])))^p)  \\
&\Eq& \bqis p \bqis q (q \et \blacksquare (q \imp t(\phi[q\backslash x])^p))  \\
&\Eq& \bqis q \bqis p (q \et \blacksquare (q \imp t(\phi[q\backslash x])^p))  \\
&\Eq& \bqis q (q \et \bqis p \blacksquare (q \imp t(\phi[q\backslash x])^p))  \\
&\Imp& \bqis q (q \et \blacksquare \bqis p (q \imp t(\phi[q\backslash x])^p)) \hspace{2cm} (*) \\
&\Eq& \bqis q (q \et \blacksquare (q \imp \bqis p \ t(\phi[q\backslash x])^p))  \\
&\Eq& \bqis q (q \et \blacksquare (q \imp t(\F \phi[q\backslash x]))) \\
&\Eq& t(\nu x. \F \phi) \\
\end{array}
\]
This proof simply applies known validities of bisimulation
quantifiers. Note that line $(*)$ is not an equivalence. The other direction holds if $\phi$ is a $df$. This we now prove.

We may assume w.l.o.g.\ that disjunctive formula $\nu x.\phi$ contains no free variables, i.e., $\phi$ is (also) a disjunctive formula with only the free variable $x$. We recall that in a disjunctive formula, a conjunction can only be between a purely propositional part and a cover modality part, and that fixed-point variables are not allowed in the propositional part (see Section \ref{sec.tech}). Importantly this means that propositional variable $q$ (witnessing fixed-point variable $x$), that occurs in the formula $(\phi[q\backslash x])^p$, can only appear in a conjunction, if it appears in the scope of a cover operator within that conjunction. This has the following significant consequence:

\begin{quote} {\em 
If $M_s\models\phi[q\backslash x]$, where $\phi$ is a disjunctive formula, then there is a model $N_u\bisim^q M_s$ such that $N_u^*\models\phi[q\backslash x]$ where $N_u^*$ is the restriction of $N_u$ to states that are not successors of $q$ states.
}
\end{quote}
That is, whether or not $N_u$ satisfies $\phi[q\backslash x]$ is invariant to any successors of states in $V^{N^*}(q)$.\footnote{Throughout this proof we will assume that all models are trees or forests (i.e.\ every state has at most one predecessor). As every model is bisimilar to a tree, and $\lang_\G^\mu$ and $\lang_{\bqall\blacksquare}$ are bisimulation invariant, this will not affect the validity of the presented argument.} 
To see this, we note that a disjunctive formula $\phi[q\backslash x]$ is true at $M_s$, if and only if there is some pointed model $N_u$ that is bisimilar to $M_s$, and some minimal relation $\rho$ between the states of $S^{N_u}$ and subformulas of $\phi[q\backslash x]$ such that:
\begin{enumerate}
\item $u \ \rho \ \phi[q\backslash x]$;
\item if $v \ \rho \ (\psi_1\lor \psi_2)$, then either $v \ \rho \ \psi_1$ or $v \ \rho\ \psi_2$ but not both;
\item if $v \ \rho \ (\chi\land\covers\Phi)$, then $N_v\models\chi$ and for every successor $v'$ of $v$ there is a unique $\psi\in\Phi$ such that $v' \ \rho \ \psi$, and for every $\psi\in\Phi$, there is at least one successor $v'$ of $v$ where $v' \ \rho \ \psi$;
\item if $v \ \rho \ \nu y.\psi$, then $v \ \rho \ \psi[\nu y.\psi\backslash y]$.
\end{enumerate}
It is clear that if such a relation exists then $N_u\models\phi$. As $q$ is replacing the fixed-point variable $x$ (which can only appear in the scope of a cover operator), the minimality of $\rho$ guarantees that if $v \ \rho \ q$, then there is no formula $\psi\neq q$ such that $v \ \rho \ \psi$, and hence, for all successors $v'$ of $v$ there is no formula $\psi$ such that $v' \ \rho \ \psi$. Consequently these successors do not impact the existence of the relation $\rho$, and thus do not affect whether or not $N_u\models\phi[q\backslash x]$.  

An explicit construction for $N_u$ can be given via the tableaux of Janin and Walukiewicz \cite{janinetal:1995}. Using their tableaux \cite[Def.\ 3.1]{janinetal:1995}, the concept of a marking \cite[Def.\ 3.6]{janinetal:1995} can be adapted to give the required model, $N_u$. This construction is important for the proof now to follow.

Suppose $M_s$ is any countable model such that
  $M_s\models\bqis q(q \land \blacksquare\bqis p(q\imp t(\phi[q\backslash x])^p))$,
  where $\phi$ is a $df$. 
  We would like to build some model $M^\omega_{u}$ such that
  \begin{itemize}
  \item $M^\omega_{u}\bisim^{p,q} M_s$,
  \item $M^\omega_{u}\models q\land \blacksquare(q\imp t(\phi[q\backslash x])^p)$
  \end{itemize}

  We inductively build a sequence of (pointed) models $M^i_u = (S^i, R^i, V^i,u)$ 
  such that $M^i_u\bisim^{p,q} M_s$, and furthermore, the models $M^i$ are fixed 
  up to a given set of states. 


\begin{definition}
Suppose that $M_s = (S, R, V,s)$ is a pointed tree like model (so for each $t\in S$, there is at most one $t'\in S$ such that $(t',t)\in R$). Let $T\subseteq S-\{s\}$. The model $M_s$ {\em up to }$T$ (written $M_s\uparrow T$) is the model $(S', R', V, s)$ where $S'$ is the set of states that are not proper descendants of $T$ and $R' = R\cap (S'\times S')$.
\end{definition}
Effectively, the model $M_s\uparrow T$ is the model $M_s$ with all the successors of any state in $T$ removed.
For each $i$ there will be a set of states $T^i \subset S^i$ such that for all $j>i$, $M^i\uparrow{T^i} = M^j\uparrow{T^i}$.

  This means we are able to give a well-defined limit for this sequence.  
  At each point of the induction, $T^i$ will represent a frontier of states in the model where we require $q\et t(\phi[q\backslash x])^p$ to be true.
  Because we are working with disjunctive formulas, we can change the submodels rooted at states in $T^i$, 
  without affecting the interpretation of $t(\phi[q\backslash x])^p$ in other parts of the model.
  This way we are able to find a single model $M^\omega_s$ with the required properties.


  We now define the sequence of models $M^i$. For each $i$ we define a model 
  and a set of states $T^i\subseteq S^i$ on which we will extend the construction.
  The proposition to be shown by inductive proof is
  $$\begin{array}{l}
  M^i_u\bisim^{p,q} M_s\\
  \forall u'\in T^i,\ M^i_{u'}\models\bqis q(q \et \blacksquare \bqis p (q\imp \ t(\phi[q\backslash x])^p)),\\
  \forall j<i, \forall u'\in T^j,\ M^i_{u'}\models q\et t(\phi[q\backslash x])^p,\ \text{and} \\
  \forall j<i, N^i\uparrow{T^j} = N^j\uparrow{T^j}.
  \end{array}
  $$
  To define the base case, it is sufficient to let $M^0 = M$ and $T^0 = \{s\}$.
  It is clear that the induction hypothesis holds here. 
  Now, for the inductive step, assume that the proposition holds for $i$. For each $u\in T^i$, we have 
  $$M^i_{u}\models \bqis q(q\land\blacksquare\bqis p(q\imp \ t(\phi[q\backslash x])^p)).$$
  Hence, for each $u\in T^i$, there is some (tree-like) $N^u_{v^u}\bisim^{p,q} M^i_{u}$ such that 
  $$ N^u_{v^u}\models q\et t(\phi[q\backslash x])^p\et \blacksquare\bqis p(q\imp t(\phi[q\backslash x])^p).$$
We will assume w.l.o.g.\ that all models $N^u$ for $u\in T^i$ and $M^i$ have disjoint sets of states.  As $\phi$ is a disjunctive formula, we may further assume that $N^u_{v^u}\models t(\phi[q\backslash x])^p$ is invariant to any successors of $V^{N^u}(q)\backslash\{v^u\}$. 
  This allows us (as the induction proceeds) to replace the submodels rooted at $v^u$ without affecting whether $t(\phi[q\backslash x])^p$ is satisfied in other parts of the model. 

  We now append the models $N^u_{v^u}$ (for $u\in T^i$) to the model $M^i$. Formally, let $M' = M^i\uparrow{T^i} = (S', R', V')$, then
  $$
  \begin{array}{l}
  S^{i+1} = S'\cup \bigcup_{u\in T^i} S^{N^u};\\
  R^{i+1} = R'\cup\bigcup_{u\in T^i} R^{N^u} \cup \{(u,v)\ |\ u\in T^i,\ v^u R^{N^u} v\};\\
  \text{for all } r: V^{i+1}(r) = V'(r) \bigcup_{u\in T^i} V^{N^u}(r).
  \end{array}
  $$
  Finally, we let $T^{i+1} = \bigcup_{u\in T^i} V^{N^u}(q)$. 


  We can see that the proposition to be shown holds for $i+1$ as follows:
  \begin{itemize}
  \item $M^{i+1}_s\bisim^{p,q} M_s$ since, for all $u\in T^i$,  $M^{i}_u\bisim^{p,q} N^u_{v^u}$, and $M^{i}_s\bisim^{p,q} M_s$ from the induction hypothesis. 
  	A $\{p,q\}$-bisimulation between $M^{i+1}_s$ and $M_s$ can be constructed by composing these bisimulations.\footnote{Specifically, let $\bisrel^u$ be the $\{p,q\}$-bsimulation between $M^i_u$ and $N^u_{v^u}$, and $\bisrel^i$ be the $\{p,q\}$-bisimulation between $M^i_s$ and $M_s$. We define the $\{p,q\}$-bisimulation $\bisrel^{i+1}$ from $M^{i+1}_s$ to $M_s$ by: for all $t\in S^{i+1}$, for all $t'$ in $S$, $(t,t')\in \bisrel^{i+1}$ if and only if either ($t\in S^i$ and $(t,t')\in \bisrel^i$), or ($t\in S^{N^u}$, and for some $v\in S^i$, $(t,v)\in \bisrel^u$ and $(v,t')\in \bisrel^i$). It is straightforward to check that $\bisrel^{i+1}$ is a bisimulation.}
  \item $\forall v\in T^{i+1},\ M^{i+1}_{v}\models \bqis q(q\land \blacksquare \bqis p (q\imp \ t(\phi[q\backslash x])^p))$, since
  	for all $u\in T^{i+1}$, $u\in V^{i+1}(q)$, and $N^u_{v^u}\models q\et\blacksquare\bqis p(q\imp t(\phi[q\backslash x])^p)$.
  \item $\forall j<i+1, \forall u\in T^j,\ M^j_{u}\models q\et t(\phi[q\backslash x])^q$; by the reasoning presented above,
	$M^j_u\models t(\phi[q\backslash x])^q$ is invariant to the successors of the states in $T^{j+1}$. 
	Therefore, if $M^j_u\models q\et t(\phi[q\backslash x])^q$, then $M^{j+1}_u\models t(\phi[q\backslash x])^q$.
  \item $ \forall j<i, M^i\uparrow{T^j} = M^j\uparrow{T^j}$ follows immediately from the construction.
  \end{itemize}
  We now let $M^{\omega} = (S^{\omega}, R^{\omega}, V^{\omega})$ where 
  \begin{itemize}
  \item $s'\in S^{\omega}$ iff for some $i$, $s'\in S^j$ for all $j>i$,
  \item $u R^\omega v$ iff for some $i$, $u R^j v$ for all $j>i$,
  \item $u\in V^\omega(r)$ iff for some $i$, $u\in V^j(r)$ for all $j>i$,
  \end{itemize}
  and let $T^\omega=\emptyset$. It is clear that the limit step will also preserve the induction hypothesis, 
  so we have $M^\omega_u\bisim^{p,q}M_s$ and $N^\omega\models q\land\blacksquare(q\imp \ t(\phi[q\backslash x])^p)$, since
  by construction $V^{\omega}(q) = \bigcup_{i<\omega} T^i$.
  Thus, $M_s\models\bqis p\bqis q(q\land\blacksquare(q\imp \ t(\phi[q\backslash x])^p))$ (i.e., $M_s\models\F\nu x.\phi$) as required.

The construction is represented in Figure~\ref{mu-sound-1}.

\begin{figure}[h]
\begin{center}
\scalebox{0.4}{
\includegraphics{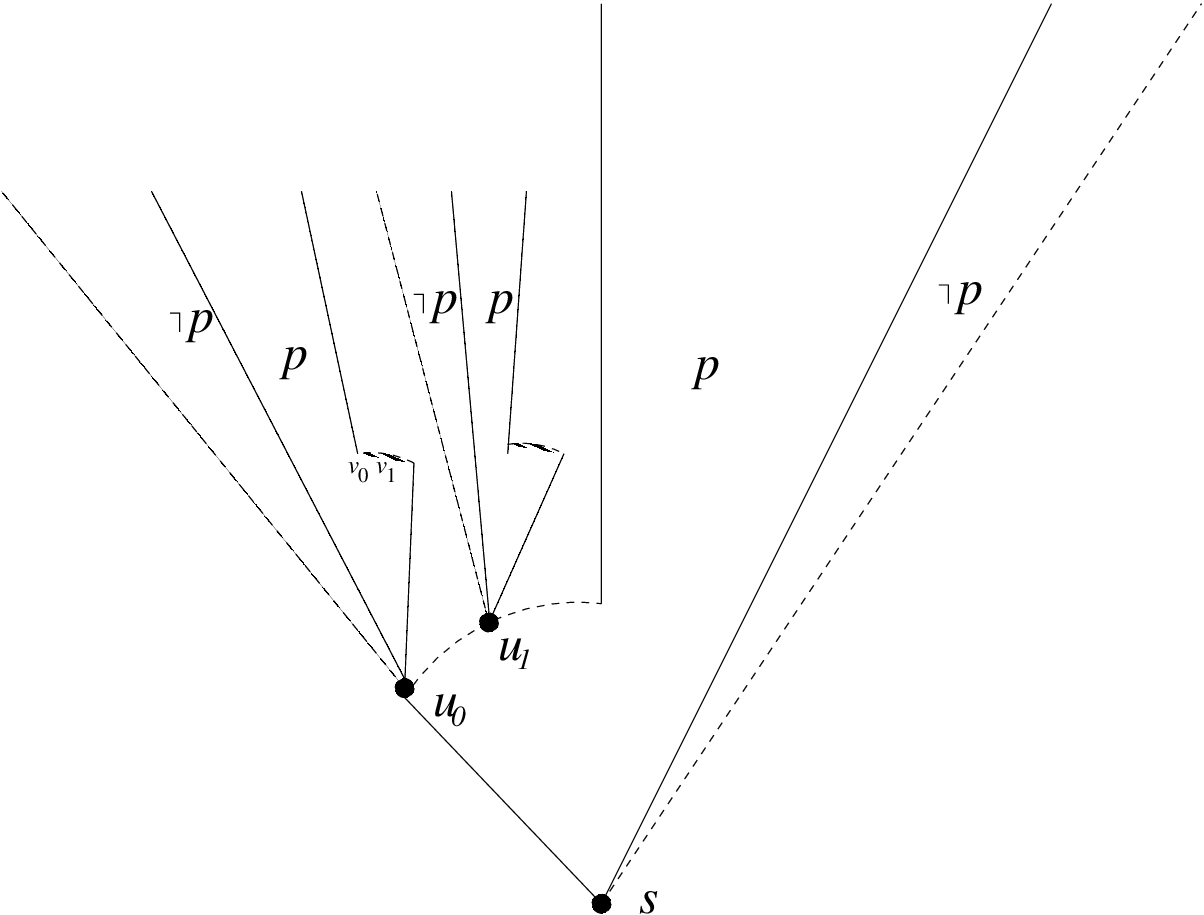}
}
\end{center}
\caption{The inductive step for the construction of $M^\omega$. The
  formula $t(\phi[q\backslash x])^p$ is independent of any state where $p$ is not
  true, or any state beyond the frontier defined by
  $u_0,u_1,...$.\label{mu-sound-1}}
\end{figure}

\bigskip

\noindent {\bf Case R$^\nu$}

We also use the contrapositive form of the axiom:
  $\F\mu x.\phi\leftrightarrow\mu x.\F\phi$. For any $\phi\in\lang_\G$ we have that:
\[
\begin{array}{lcl}
t(\F\mu x.\phi) &\Eq& \bqis p \ t(\mu x.\phi)^p \\
&\Eq& \bqis p (\bqall q (\blacksquare(t(\phi[q \backslash x])\imp q)\imp q))^p \\
&\Eq& \bqis p \bqall q (\blacksquare(t(\phi[q \backslash x])^p\imp q)\imp q) \\
&\Imp& \bqall q \bqis p (\blacksquare(t(\phi[q \backslash x])^p\imp q)\imp q) \hspace{2cm} (**) \\
&\Eq& \bqall q\bqis p(\blacklozenge(t(\phi[q \backslash x])^p\land \lnot q)\lor q)\\
&\Eq& \bqall q(\bqis p\blacklozenge(t(\phi[q \backslash x])^p\land\lnot q)\lor q)\\
&\Eq& \bqall q(\blacklozenge\bqis p(t(\phi[q \backslash x])^p\land\lnot q)\lor q) \hspace{2cm} (***) \\
&\Eq& \bqall q(\blacklozenge(\bqis p \ t(\phi[q \backslash x])^p\land\lnot q)\lor q)\\
&\Eq& \bqall q(\blacksquare(\bqis p \ t(\phi[q \backslash x])^p\imp q)\imp q)\\
&\Eq& \bqall q(\blacksquare(\F \phi[q \backslash x] \imp q)\imp q)\\
&\Eq& t(\mu x.\F\phi)
\end{array} \]
The equivalence in (***) is true because $\blacklozenge$ is the existential modality which quantifies over all states in the model. Obviously, the implication in line (**) is only true in one direction (the usual quantifier swap $\exists\forall \imp \forall\exists$). 

To prove the other direction in the equivalence $\F\mu x.\phi\leftrightarrow\mu x.\F\phi$, we now show directly that $\models \mu x.\F\phi \imp \F\mu x.\phi$ in refinement $\mu$-calculus, for $\phi$ a $df$ (observe that $\mu x.\phi$ is then a $df$ as well). We use the inductive characterization of $\mu x.\F\phi$ of \cite{arnoldetal:2001} which tells that $M_s\models\mu x.\F\phi$ if and only if $s\in\|\F\phi\|_\tau$ for some ordinal $\tau$, where we recall the definition of the semantic operation $\|\bullet\|$: $\|\F\phi\|_0 = \emptyset$, and $s \in \|\F\phi\|_\tau$ whenever $M^\tau_s\models\F\phi$, where $M^\tau = M^{[\sigma]}$ with $\sigma = x\mapsto\bigcup_{\tau'<\tau}\|\F\phi\|_{\tau'}$.

Suppose $M_s\models\mu x.\F\phi$. Since
$\lang_\G^\mu$ is bisimulation invariant, without loss of generality
we may suppose that $M$ is a countable tree-like model. As $M_s$
satisfies $\mu x.\F\phi$, there must be some least ordinal $\tau$
whereby $s\in\|\F\phi\|_\tau$. We give a proof by induction over
$\tau$ that $s\in\|\F\phi\|_\tau$ implies $M_s\models \F \mu
x.\phi$. The base case where $\tau = 0$ is trivial. Now consider $M^\tau = M^{[\sigma]}$ with $\sigma = x\mapsto\bigcup_{\tau'<\tau}\|\F\phi\|_{\tau'}$. Then
$M^\tau_s\models\F\phi$. As $\mu x.\phi$ is a $df$, there 
is a refinement of $M^\tau$ with a frontier such
that $x$ may only be true at $s$ or on this frontier, and no point
beyond the frontier affects the interpretation of $\phi$. Formally,
there is a set of states $\{u_0,u_1,...\}\in V^\tau(x)$ such that
$M'_s\models \F \phi$ (i.e., $M'_s\models \bqis \atom \ t(\phi)^\atom$), where $M' = (S', R', V')$ with
\begin{itemize}
\item  $S'\subseteq S^\tau$ is the set of states reachable from $s$, but not from any $u_i$;
\item $V'(x) = \{t,u_0,u_1,...\}$, $V'(y) = V^{M^\tau}(y)$ for $y\neq x$; and
\item $R' = R^\tau\backslash\{(u_i,t) \mid t\in S^\tau,i = 0,1,...\}$.
\end{itemize}
We note that $M_s'$ is a refinement of $M_s^\tau$. Now as for each $i$,
$u_i\in\|\F\phi\|_j$ for some $j<\tau$, by the inductive hypothesis
we may assume there is some model $N^i = (S^i,R^i,V^i)$ where
$N^i_{v_i}\simul M^\tau_{u_i}$ and $N^i_{u_i}\models\mu x.\phi$. We
may append these models to $M'$, to define $M^* = (S^*, R^*, V^*)$
where $S^*=S'\cup\bigcup_i S^i$, $R^* = R'\cup \bigcup_i R^i\cup \{(t,
v_i)\ |\ (t,u_i)\in R'\}$, and $V^*(y) = V'(y)\cup\bigcup_i V^i(y)$
for all $y\in\Atoms$. (Notice the similar construction in the soundness proof of axiom {\bf RK}.)  It is clear that $M_s^*$ is a refinement of $M_s$,
and by the axiom {\bf F1} we can see $M_s^*\models\mu x.\phi$ as
required.
\end{proof}

The general form of {\bf R$^\mu$} is not sound. For
example, take $\phi = \mu z.\susp(p\imp q)\imp\susp(\neg p\imp
x)$. Then $\G\mu x.\phi$ is true if $p$ is true at every immediate
successor of the current state, whereas $\mu x.\G\phi$ is only true at
states with no successor. Likewise {\bf R$^\nu$} is not true in the
general case, as can be seen by taking $\phi =
p\land\know(\susp\top\imp x)$. Then $\nu x. \G\phi$ is true if and
only if $p$ is true at every reachable state, and $\G\nu x.\phi$ is
true only if $p$ is true at every state within one step.

\subsection{Completeness}

The completeness proof of $\axiomRML^\mu$ proceeds
exactly as for Theorem~\ref{theo:SandC}, replacing the formulas in
cover logic with disjunctive formulas, to get a statement similar to that of Proposition~\ref{lem:express}.

\begin{proposition}\label{lem:express-mu}
Every formula of $\lang^\mu_\G$ is equivalent to a formula of $\lang^\mu$.
\end{proposition}
\begin{proof}
Given a formula $\psi$, we prove by induction on the number of
the occurrences of $\F$ in $\psi$ that it is equivalent to an
$\F$-free formula, and therefore to a formula in the modal
$\mu$-calculus $\lang^\mu$. The base is trivial. Now assume $\psi$
contains $n+1$ $\F$-operators. Choose a subformula of type $\F \phi$
of our given formula $\psi$, where $\phi$ is $\F$-free (i.e.\, choose
an innermost $\F$). As $\phi$ is $\F$-free, it is semantically equivalent to a formula in disjunctive normal form, and by the completeness of Kozen's
axiom system \cite{Walukiewicz00} this equivalence is provable in $\axiomRML^\mu$. By {\bf NecR} and {\bf R} it follows that $\F\phi$ is
provably equivalent to some formula $\F\psi$ where $\psi$ is a
disjunctive formula (analogously to Proposition \ref{lem:substeq} one can easily show that $\axiomRML^\mu$ satisfies substitution of equivalents). Thus without loss of generality, we may
assume in the following that $\phi$ is in disjunctive normal form. We
may now proceed by induction over the complexity of $\phi$, and conclude 
that $\F
\phi$ is logically equivalent to a formula $\chi$ without $\F$. All cases 
of this induction are as before, we only show the final two, different 
cases:
\begin{itemize}
\item $\F\mu x.\phi$ iff $\mu x.\F\phi$ (by {\bf R$^\nu$} noting that all 
subformulas of a disjunctive formula are themselves disjunctive); IH.
\item $\F\nu x.\phi$ iff $\nu x.\F\phi$ (by {\bf R$^\mu$}); IH.
\end{itemize}
Replacing $\F \phi$ by $\chi$ in $\psi$ gives a result with one less 
$\F$-operator, to which the (original) induction hypothesis applies.
\end{proof}

\begin{theorem}\label{theo:SandC-mu}
The axiom schema $\axiomRML^\mu$ is sound and complete for the logic $\logicRML^\mu$
\end{theorem}

\begin{proof}
Soundness follows from Theorem~\ref{theo:soundness-mu} and Theorem~\ref{theo:soundness}. To see $\axiomRML^\mu$ is complete, suppose $\phi\in\lang_\G^\mu$ is a valid formula. Then by Lemma~\ref{lem:express-mu}, $\phi$ is provably equivalent to some valid formula $\psi\in\lang^\mu$. As $\psi$ is valid, it must be provable since {\bf Prop}, {\bf K}, {\bf F1}, {\bf F2}, {\bf NecK}, and {\bf MP} give a sound and complete proof system for the modal $\mu$-calculus \cite{Walukiewicz00}. A proof of $\phi$ follows by {\bf MP}.
\end{proof}

\section{Complexity}
\label{sec-complexity}

Decidability for both $\lang_\G$ and $\lang_\G^{\mu}$ follows from the
fact that a computable translation is given in the completeness proofs
of Sections~\ref{sec-FEL} and~\ref{sec-FEL-mu}: note that the given
translations to $\lang$ and $\lang^{\mu}$ respectively, are recursive
and involve transforming formulas into their disjunctive normal
forms. Hence they are non-elementary in the size of the original
formula. This non-elementary procedure for $\lang_\G^\mu$ is
optimal as shown in Section~\ref{sec-RMLmuisNonElem} below. 

Unfortunately we were not able to corroborate in this paper the upper complexity
claims for $\logicRML$ reported in \cite{hvdetal.felax:2010}. But
towards some indication of a result in that direction, we further
establish a doubly exponential succinctness proof for $\lang_\G$ in
Section~\ref{sec-succinctness}. (On complexity, see also `Recent results' in Section~\ref{sec.conclusions}.)

\subsection{$\logicRML^\mu$ is non-elementary}
\label{sec-RMLmuisNonElem}

This section is dedicated to the proof of the following result.

\begin{theorem}\label{theo:RLmuNonElem}
The satisfiability problem for $\logicRML^\mu$ is non-elementary, even for the single-agent setting.
\end{theorem}

In the rest of this section, we only consider a single-agent setting.

First, we recall a fragment, written $\CTLfrag$, of the standard
branching-time logic \emph{Computation Tree Logic} ($\CTL$)
\cite{CE81}, which in turn is a fragment of $\lang^\mu$ (see also the example Section~\ref{sec.examples}). 
$$ \CTLfrag \ni \phi::=\top\ |\ \bot\\ |\ p\ |\ \lnot \phi\ |\ \phi\land\phi\ |  \know\phi\ |\ \susp\phi\ |\ \EFR\phi\ |\ \AFR\phi$$

Let $M$ be a model and $\state$ be an $M$-state. A \emph{path from
  $\state$} is a finite or infinite sequence of states
$\pi=\state_0,\state_1,\ldots$ s.t. $\state_0=\state$ and each
$\state_{i+1}$ is an successor of $\state_i$. Only the semantics of $\AFR$
and $\EFR$ is recalled (as for other formulas it is clear):

\[ \begin{array}{l}
  M_\state \models \EFR\varphi \ \mbox{iff} \  \text{ there are a maximal path }\pi=\state_0,\state_1,\ldots \text{ from  }\state \text{ and  }i\geq 0\\
  \text{\hspace{2.8cm}}\text{ such that  } M_{\state_i}\models\varphi \\
  M_\state \models \AFR\varphi \ \mbox{iff} \  \text{ for each maximal path }\pi=\state_0,\state_1,\ldots \text{ from  }\state, \\
  \text{\hspace{2.8cm}}\text{ there is  }i\geq 0 \text{ such that  } M_{\state_i}\models\varphi 
\end{array} \]

Directly translating $\CTLfrag$ in $\lang^\mu$ is routine via the
following mapping $\CTLtoLmu : \CTLfrag \to \lang^\mu$, defined by induction
over the formulas\label{page-CTLtoLmu}:  $\CTLtoLmu(\top)=\top$, $\CTLtoLmu(p)=p $,
$\CTLtoLmu(\lnot \phi)=\lnot \CTLtoLmu(\phi)$,
$\CTLtoLmu(\phi\land\phi')= \CTLtoLmu(\phi) \land \CTLtoLmu(\phi')$,
$\CTLtoLmu(\know\phi)=\know\CTLtoLmu(\phi)$, $\CTLtoLmu(\M\phi) = \M\CTLtoLmu(\phi)$,
$\CTLtoLmu(\EFR\phi)=\mu x. \CTLtoLmu(\phi) \lor \M x$,
$\CTLtoLmu(\AFR\phi) =\mu x. \CTLtoLmu(\phi) \lor \know x$.

We also use standard abbreviations for the duals $\AGR\varphi$ iff
$\neg\EFR\neg\varphi$ (`universal always'), and $\EGR\varphi$ iff
$\neg\AFR\neg\varphi$ (`existential always').  A $\CTLfrag$ formula is
in \emph{positive form} if negation is applied only to propositional
variables. A $\CTLfrag$ formula $\varphi$ is \emph{existential} if it
is in positive form and there are no occurrences of the universal
modalities $\AFR$ and $\AGR$ and the modality $\know$. The following can
be proved by using
Proposition~\ref{prop-logicalcharacterization-of-ref}, enriched for
the case of $\EFR$ formulas (with a transfinite induction argument for
this fixed-point formula).

\begin{proposition}\label{prop:ExistCTL} Let $M_\state$ and $N_{\stateb}$ be two models with $M_{\state} \lumis N_{\stateb}$.
  Then for each existential $\CTLfrag$ formula $\varphi$,
  $N_{\stateb}\models \varphi$ implies $M_{\state}\models \varphi$.
\end{proposition}

\begin{definition}[Refinement $\CTLfrag$] \emph{Refinement} $\CTLfrag$
  ($\RCTL$, for short) is the extension of $\CTLfrag$ with the
  refinement quantifiers $\F$ and $\G$.
\end{definition}

\begin{definition}[Refinement Quantifier Alternation Depth] 
  We first define the \emph{alternation length} $\ell(\seqQ)$ of
  finite sequence $\seqQ \in \{\F,\G\}^*$ of quantifiers, as the
  number of alternations of existential and universal refinement
  quantifiers in $\seqQ$. Formally, $\ell(\epsilon)=0$, $\ell(Q)=0$
  for every $Q \in \{\F,\G\}^*$, and $\ell(QQ'\seqQ)=\ell(Q'\seqQ)$ if
  $Q=Q'$, $\ell(Q'\seqQ)+1$ otherwise.

  Given a $\lang_\G$ (resp., $\lang^\mu_\G$, resp., $\RCTL$) formula
  $\varphi$, the \emph{refinement quantifier alternation depth}
  $\ad(\phi)$ of $\varphi$ is defined via the standard tree-encoding
  $T(\varphi)$ of $\varphi$, where each node is labeled by either a
  modality, a boolean connective, or a propositional
  variable. Then, $\ad(\phi)$ is the maximum of the alternation
  lengths $\ell(\seqQ)$ where $\seqQ$ is the sequence of refinement
  quantifiers along a maximal path of $T(\varphi)$ from the
  root.
\end{definition}


\begin{theorem}\label{theo:lowerBoundRefCTL} Let the class ${\mathcal C}_k=\{\phi \in \RCTL
  \,|\, \ad(\phi) \leq k\}$. The satisfiability problem for ${\mathcal
    C}_k$ is $k$-\EXPSPACE-hard, for any $k$.
\end{theorem}

Theorem~\ref{theo:lowerBoundRefCTL} is proved by a polynomial-time
reduction from satisfiability of \emph{Quantified Propositional
  Temporal Logic} ($\QLTL$) \cite{SVW87}.  First, we recall the syntax
and the semantics of $\QLTL$.  The syntax of $\QLTL$ formulas
$\varphi$ over a countable set $P$ of propositional variables is
defined as follows:
$$\varphi ::= p \ | \ \neg \varphi \ | \ \varphi \wedge \varphi \ | \ \varphi \vee \varphi\ | \ \Next \varphi \ 
| \ \eventually \varphi \ | \ \qis \atom.\varphi$$ where $p\in P$,
$\Next$ is the `next' modality, $\eventually$ is the `eventually'
modality, and $\qis$ is the existential quantifier.\footnote{We distinguish (domain) quantifiers $\qis$ and $\qall$ in use here, from the refinement quantifiers $\F$ and $\G$, and from the bisimulation quantifiers $\bqis$ and $\bqall$.} We also use
standard abbreviation $\always\varphi$ for
$\neg\eventually\neg\varphi$ (`always').  

The semantics is given w.r.t.\ elements of $(2^P)^\omega$, namely
infinite words $w$ over $2^{P}$.  Beforehand, we need some technical
notions. Let $w \in (2^P)^\omega$. For each $i\geq 0$, $w(i)$ denotes
the $i$th symbol of $w$. Moreover, for each $P'\subseteq P$, we define
the equivalence relation $\weq{P'}$ over $(2^P)^\omega$: two infinite
words $w_1$ and $w_2$ are $\weq{P'}$-equivalent whenever their
\emph{projections onto $P'$} are equal. The projection of an infinite
word $w$ onto $P'$, written $\proj(w,P')$, is obtained by removing
from each symbol of $w$ all the propositions in $P\setminus
P'$. Hence, $w_1\weq{P'}w_2$ iff $\proj(w_1,P')=\proj(w_2,P')$.

Given a $\QLTL$ formula $\varphi$, an infinite word $w$ over
$2^{P}$, and a position $h\geq 0$ along $w$, the satisfaction relation
$(w,h)\models \varphi$ is inductively defined as follows (we omit the
clauses for the boolean connectives):
\[ \begin{array}{l}
(w,h) \models p \ \mbox{iff} \  p \in w(h) \\
(w,h) \models \Next \varphi \ \mbox{iff} \  (w,h+1) \models \varphi \\
(w,h) \models \eventually \varphi \ \mbox{iff} \  \text{ there is  } h'\geq h \text{ such that }(w,h) \models \varphi \\
(w,h) \models \qis p. \varphi \ \mbox{iff} \  \text{ there is  } w', w'\weq{P\setminus\{p\}}w \text{ and } (w',h) \models \varphi
\end{array} \]

We say that the word $w$ satisfies $\varphi$, written $w\models
\varphi$, if $(w,0)\models \varphi$. A $\QLTL$ formula $\varphi$ is in
\emph{positive normal form} if it is of the form $Q_1 p_1.Q_2
p_2.\ldots Q_n p_n.\varphi_{n+1}$, where $Q_j\in\{\qis,\qall\}$ for
each $1\leq j\leq n$, and $\varphi_{n+1}$ is a quantification-free
$\QLTL$-formula in which negation is applied only to propositional
variables\footnote{Every $\QLTL$ formula is constructively
  equivalent to a formula in positive normal form, with linear size.}. The
\emph{quantifier alternation depth} of $Q_1 p_1.Q_2
p_2.\ldots Q_n p_n.\varphi_{n+1}$ is the number of alternations of
(existential and universal) quantifiers in the string
$Q_1Q_2\ldots Q_n$. The following is a well-known result.

\begin{theorem}\cite{SVW87}\label{theo:QLTLisNonElem} Let $k\geq
  0$. Then, the satisfiability problem for the class of $\QLTL$
  formulas in positive normal form whose quantifier alternation depth
  is $k$ is $k$-\EXPSPACE-hard.
\end{theorem}

Note that Theorem~\ref{theo:QLTLisNonElem} holds even if we assume
that formulas in positive normal form like $Q_1 p_1.Q_2 p_2.\ldots Q_n
p_n.\varphi_{n+1}$ (with $\varphi_{n+1}$ quantification-free) are
such that $p_1,\ldots,p_n$ are pairwise distinct, each proposition
occurring in $\varphi_{n+1}$ is in $\{p_1,\ldots,p_n\}$, and
$Q_n=\qall$.\\

Theorem~\ref{theo:lowerBoundRefCTL} directly follows from Theorem~\ref{theo:QLTLisNonElem} and the
following theorem, whose proof is given in the rest of this section.

\begin{theorem}\label{theo:FromQLTLtoRCTL} For every $\varphi \in \QLTL$, one can construct in time 
  polynomial in the size of $\varphi$ a formula
  $\translate{\varphi} \in \RCTL$, such that  $\varphi$ is
  satisfiable if, and only if, $\translate{\varphi}$ is satisfiable. Moreover, the refinement
  quantifier alternation depth of $\translate{\varphi}$, $\ad(\translate{\varphi})$, is equal to
  the quantifier alternation depth of $\varphi$.
\end{theorem}

Before proving Theorem~\ref{theo:FromQLTLtoRCTL}, we need
additional definitions. Let $P=\{p_1,\ldots,p_n\}$ and
$\translate{P}=P\cup \{p_0,\overline{p}_1,\ldots,\overline{p}_n\}$,
where $p_0$, $\overline{p}_1,\ldots,\overline{p}_n$ are fresh
propositional variables (intuitively, $\overline{p}_i$ is used to
encode the negation of $p_i$ for each $1\leq i\leq n$, and $p_0$ is a
new variable that will be used to mark a path). For a model $M$ and
two states $\state$ and $\state'$ in $M$, $\state'$ is \emph{reachable
  from $\state$} if there is a finite path from $\state$ leading to
$\state'$.  Let $0\leq j\leq n$. A pointed model $M_\state$ $($over
$\translate{P}$$)$ is \emph{well-formed w.r.t.\ $j$} if the following
holds:
\begin{enumerate}
\item for each state $\state'$ of $M$ which is reachable from
  $\state$, there is exactly one proposition $p\in \translate{P}$ such
  that $s'\in V^{M}(p)$ (we say that $\state'$ is a $p$-state);
  moreover, $\state$ is a $p_0$-state;
\item each state $\state'$ reachable from $\state$ which is not a
  $p_0$-state has no successor;
\item each $p_0$-state $\state'$ which is reachable from $\state$
  satisfies: (i) $\state'$ has some $p_0$-successor, (ii) for all
  $1\leq i\leq j$, either $\state'$ has a $p_i$-state successor or (exclusive)
  a $\overline{p}_i$-state successor, and (iii) for all $j+1\leq i\leq n$, $\state'$ has both a
  $p_i$-state successor and a $\overline{p}_i$-state successor.
  \end{enumerate}
  For each $0\leq j\leq n$, the following $\CTLfrag$ formula $\psi_j$ over
$\translate{P}$ characterizes the set of pointed models which are
well-formed w.r.t.\ $j$:
\begin{center}
\begin{tabular}{lcl}
$\psi_j$ & $:=$ & $p_0\wedge
  \AGR\Bigl\{\bigl[\bigvee_{p\in \translate{P}} (p\wedge
  \bigwedge_{p'\in \translate{P}\setminus\{p\}}\neg p')\bigr] \wedge
  \bigl[\neg p_0\rightarrow \know \bot \bigr] \wedge$\\
& & $p_0 \rightarrow\bigl[\M p_0\wedge
  \bigwedge_{j+1\leq i\leq n}(\M p_i\wedge \M \overline{p}_i)\wedge
  \bigwedge_{1\leq i\leq j}(\M(p_i\vee \overline{p}_i)\wedge (\know
  \neg p_i\vee \know\neg \overline{p}_i))\bigr]\Bigr\}$
\end{tabular}
\end{center}

\noindent Intuitively, $\psi_j$ enforces the existence of infinite paths  
$\pi=s_0 s_1\ldots $ which visit only $p_0$-states $s_i$ such that the  
following holds: the set of successors of $s_i$ `encodes' a specific  
truth valuation of the variables $p_1,\ldots,p_j$ and all the possible  
truth valuations of the variables
$p_{j+1},\ldots,p_n$.


A pointed model $M_\state$ is \emph{well-formed} if it is well-formed
w.r.t.\ $j$ for some $0\leq j\leq n$. In this case, we say that
$M_\state$ is \emph{minimal} if, additionally, each $p_0$-state which
is reachable from $\state$ has exactly one $p_0$-successor. 

A well-formed pointed model $M_\state$ encodes a set of infinite words
over $2^{P}$, written $\words(M_\state)$, given by: $w\in
\words(M_\state)$ iff there is an infinite path
$\pi=\state_0,\state_1,\ldots$ of $M$ from $\state$ (note that $\pi$
consists of $p_0$-states) such that for all $h\geq 0$ and $1\leq j\leq
n$, either $p_j\in w(h)$ and $\state_h$ has some $p_j$-successor, or
$p_j\notin w(h)$ and $\state_h$ has some $\overline{p}_j$-successor.

Note that if $M_\state$ is well-formed w.r.t.\ $0$, then
$\words(M_\state)=(2^{P})^\omega$. If
instead $M_\state$ is well-formed w.r.t.\ $j$ for some $0< j\leq n$
and $M_\state$ is  also minimal, then there is an infinite word
$u_j \in (2^{\{p_1,\ldots,p_j\}})^\omega$ such that 
$\words(M_\state) = \{w\in (2^{P})^\omega\,|\, \proj(w,\{p_{1},\ldots,p_j\})=u_j\}$. In
particular, when $j=n$, $\words(M_\state)$ is a singleton.

Also, one can easily see that if $M_\state \lumis N_\stateb$ then $\words(M_\state) \supseteq \words(N_\stateb)$.

\paragraph{Construction of the $\RCTL$ formula $\translate{\varphi}$
  (in Theorem~\ref{theo:FromQLTLtoRCTL}).} Pick an $\QLTL$ formula
$\varphi= Q_1 p_1.Q_2 p_2.\ldots Q_n p_n.\varphi_{n+1}$.  For each
$1\leq j\leq n$, we let $\varphi_j=Q_j p_j.Q_{j+1} p_{j+1}.\ldots Q_n
p_n.\varphi_{n+1}$ (note that $\varphi_1$ corresponds to $\varphi$).

First, we construct a $\RCTL$ formula $\translate{\varphi}_j$ over
$\translate{P}$ by using the $\CTLfrag$ formulas $\psi_{j-1}$,  for each
$1\leq j\leq n+1$. The construction is based on an induction on
$n+1-j=0,\ldots,n$ as follows:
\begin{description}
\item[Base case] ($j=n+1$). Recall that $\varphi_{n+1}$ is a
  quantification-free $\QLTL$ formula in positive normal form over
  $P$. Let $\Upsilon$ be the following mapping from the set of
  quantification-free $\QLTL$ formulas $\xi$ over $P$ in positive
  normal form to the set of \emph{existential} $\CTLfrag$ formulas
  over $\translate{P}$ (it is defined by induction).
\begin{itemize}
\item $\Upsilon(p)=\M p$\, and \, $\Upsilon(\neg p)=\M\overline{p}$ for each $p \in P$;
\item $\Upsilon(\xi_1\vee \xi_2)=\Upsilon(\xi_1)\vee\Upsilon(\xi_2)$ and $\Upsilon(\xi_1\wedge \xi_2)=\Upsilon(\xi_1)\wedge \Upsilon(\xi_2)$;
\item $\Upsilon(\Next\xi)=\M(p_0\wedge\,\Upsilon(\xi))$,\,  $\Upsilon(\eventually\xi)=\EFR (p_0\wedge\,\Upsilon(\xi))$, \,and\,  $\Upsilon(\always\xi)=\EGR(p_0\wedge\,\Upsilon(\xi))$.
\end{itemize}

Then, $\translate{\varphi_{n+1}}:= \Upsilon(\varphi_{n+1})$.
\item[Induction case] ($1\leq j\leq n$).   Recall $\varphi_j=Q_j p_j.\varphi_{j+1}$.

Then, $\translate{\varphi_{j}} := 
\left\{
  \begin{array}{ll}
    \F (\psi_j\,\wedge\, \translate{\varphi}_{j+1}) & \textrm{ if }  Q_j =\qis\\
 \G(\psi_j\,\rightarrow\, \translate{\varphi}_{j+1}) & \textrm{ if } Q_j =\qall
  \end{array}
\right.$
\end{description}

\noindent Finally, the $\RCTL$ formula $\translate{\varphi}$ over
$\translate{P}$ is given by
$\translate{\varphi}:=\psi_0\,\wedge\,\translate{\varphi}_1$.

\paragraph{Correctness of the construction.} Note that the size
of $\translate{\varphi}$ is polynomial in the size of
$\varphi$. Moreover, the refinement quantifier alternation depth of
$\translate{\varphi}$ is equal to the quantifier alternation depth of
$\varphi$. Thus, in order to prove
Theorem~\ref{theo:FromQLTLtoRCTL}, it remains to show that
$\varphi$ is satisfiable iff $\translate{\varphi}$ is satisfiable. For
this, we need three preliminary lemmata.

\begin{lemma}\label{lemma:basicNonElem} Let $M_\state$ be a pointed
  model which is well-formed w.r.t.\ $n$ and \emph{minimal}, with
  $\words(M_\state)=\{w\}$. Then, for each quantification-free $\QLTL$
  formula $\xi$ in positive normal form, $w\models \xi$ if and only if
  $M_s\models \Upsilon(\xi)$.
\end{lemma}
\begin{proof} Let $\pi=\state_0,\state_1,\ldots$ be the unique
  infinite path of $M$ from state $\state$ (note that $\pi$ consists
  of $p_0$-states). Then, by a straightforward structural induction,
  one can show that for each quantification-free $\QLTL$ formula in
  positive normal form $\xi$, the following holds: for all $h\geq 0$,
  $M_{\state_h}\models \Upsilon(\xi)$ iff $(w,h)\models \xi$. Hence, the result
  follows.
\end{proof}

Let $0\leq j\leq n$ and let $M_\state$ be a pointed model which is  well-formed
w.r.t.\ $j$. For each $j\leq i\leq n$, an \emph{$h$-segment of
  $M_\state$} is a refinement $N_\stateb$ of $M_\state$ which is
well-formed w.r.t.\ $h$ and \emph{minimal}. Note that for each $w\in
\words(M_\state)$ and for each $j\leq h\leq n$, by construction, there exists
an $h$-segment $N_\stateb$ of $M_\state$ such that $w\in
\words(N_\stateb)$.

\begin{lemma}\label{lemma:Aux2NonElem} Let $1\leq j\leq n$ and
  $M_\state$ be a pointed model which is well-formed w.r.t.\ $j-1$ such
  that for each $w\in \words(M_\state)$, $w\models \varphi_j$. Then,
  $M_\state\models \translate{\varphi}_j$.
\end{lemma}
\begin{proof} The proof is by induction on $n-j=0,\ldots,n-1$.
\begin{description}
\item[Base case] ($j=n$). Recall $\varphi_n = \qall
  p_n.\varphi_{n+1}$, where $\varphi_{n+1}$ is a quantification-free
  $\QLTL$ formula in positive normal form.  By construction,
  $\translate{\varphi}_n= \G (\psi_n \,\rightarrow
  \Upsilon(\varphi_{n+1})\,)$. Let $N_{\stateb}$ be a refinement of
  $M_\state$ which satisfies formula $\psi_n$ (if any). We need to show that
  $N_{\stateb}\models \Upsilon(\varphi_{n+1})$.  By definition of
  $\psi_n$, $N_{\stateb}$ is well-formed w.r.t.\ $n$. Let
  $N'_{\statec}$ be any $n$-segment of $N_{\stateb}$, and let $\words(N'_{\statec})=\{w\}$.  By transitivity, $N'_{\statec}$ is a
  refinement of $M_\state$, so that $w\in
  \words(M_\state)$. Thus, by hypothesis, $w\models \varphi_n=\qall
  \,p_n.\varphi_{n+1}$, which implies $w\models \varphi_{n+1}$.  By
  Lemma~\ref{lemma:basicNonElem}, it follows that $N'_{\statec}\models
  \Upsilon(\varphi_{n+1})$. Since $N'_{\statec}$ is a refinement of
  $N_{\stateb}$ and $\Upsilon(\varphi_{n+1})$ is an \emph{existential}
  $\CTLfrag$ formula, by Proposition~\ref{prop:ExistCTL} we deduce that
  $N_{\stateb}\models \Upsilon(\varphi_{n+1})$ as well. Hence, the
  result holds.  
\item[Induction step] ($1\leq j\leq n-1$). By construction, there are two cases:

(1) $\varphi_j=\qis \,p_j.\varphi_{j+1}$ and $\translate{\varphi}_j=\F
  (\psi_j\,\wedge\,\translate{\varphi}_{j+1})$: let $w_0\in
  \words(M_{\state})$. By hypothesis, $w_0\models \varphi_j$. Hence,
  there is infinite word $w'_0$ over $2^{P}$ such that
  $w'_0\weq{P\setminus\{p_j\}}w_0$ and $w'_0\models
  \varphi_{j+1}$. Since $M_\state$ is well-formed w.r.t.\ $j-1$ and
  $w_0\in \words(M_{\state})$, it follows that $w'_0\in
  \words(M_\state)$ as well. Let $N_\state$ be any $j$-segment of
  $M_\state$ such that $w'_0\in \words(N_\state)$. By definition of
  $\psi_j$, $N_\state\models \psi_j$. Thus, it suffices to show that
  $N_\state\models \translate{\varphi}_{j+1}$. Since $N_\state$ is
  well-formed w.r.t.\ $j$ and minimal, and $w'_0\in
  \words(N_{\state})$, it holds that for each $w'\in
  \words(N_{\state})$, $w'\weq{\{p_{1},\ldots,p_j\}}w'_0$. Since every
  proposition in $\{p_{j+1},\ldots,p_n\}$ does not occur free in
  $\varphi_{j+1}$ and $w'_0\models \varphi_{j+1}$, it follows that for
  each $w'\in \words(N_{\state})$, $w'\models \varphi_{j+1}$. Thus, by
  the induction hypothesis, we obtain that $N_\state\models
  \translate{\varphi}_{j+1}$, and the result holds.

 (2)  $\varphi_j=\qall \,p_j.\varphi_{j+1}$ and $\translate{\varphi}_j=\G
  (\psi_j\,\rightarrow\,\translate{\varphi}_{j+1})$: let $N_{\stateb}$
  be a refinement of $M_{\state}$ which satisfies formula $\psi_j$ (if
  any). We need to show that $N_{\stateb}\models
  \translate{\varphi}_{j+1}$. By definition of $\psi_j$, $N_{\stateb}$
  is well-formed w.r.t.\ $j$. Thus, by the induction hypothesis it
  suffices to show that for each $w\in \words(N_{\stateb})$, $w\models
  \varphi_{j+1}$. Let $w\in \words(N_{\stateb})$. Since $N_{\stateb}$
  is a refinement of $M_\state$, it holds that $w\in
  \words(M_\state)$. Thus, by hypothesis, $w\models \varphi_j=\qall
  \,p_j.\varphi_{j+1}$. Hence, $w\models \varphi_{j+1}$, and the
  result follows.
\end{description}
\end{proof}

\begin{lemma}\label{lemma:Aux1NonElem} Let $1\leq j\leq n$ and let $M_\state$ be a 
  pointed model which is well-formed w.r.t.\ $(j-1)$ and such that $M_\state \models \translate{\varphi}_j$. 
  Then, there is a $(j-1)$-segment $N_\stateb$ of $M_\state$ such that
  $N_\stateb\models \translate{\varphi}_j$ and for each $w\in
  \words(N_\stateb)$, $w\models \varphi_j$.
\end{lemma}
\begin{proof} The proof is by induction on $n-j=0,\ldots,n-1$, for
  which there are two cases. Recall that $\varphi_n=\qall
  p_n.\varphi_{n+1}$.

\noindent (1) $\varphi_j=\qall \,p_j.\varphi_{j+1}$ and $\translate{\varphi}_j=\G
  (\psi_j\,\rightarrow\,\translate{\varphi}_{j+1})$: let $N_\stateb$ be any
  $(j-1)$-segment of $M_\state$. By hypothesis $M_\state\models
  \translate{\varphi}_j$. Since every refinement of $N_\stateb$ is
  also a refinement of $M_\state$, it follows that $N_\stateb\models
  \translate{\varphi}_j$. Thus, it suffices to show that for each
  $w\in \words(N_\stateb)$, $w\models \varphi_j$. Fix $w\in
  \words(N_\stateb)$ and let $w'$ be an infinite word over $2^{P}$
  such that  $w' \weq{P\setminus\{p_j\}}w$. Since
  $N_\stateb$ is well-formed w.r.t.\ $j-1$, $w'\in \words(N_\stateb)$
  as well. Let $N'_\statec$ be a $j$-segment of $N_\stateb$ such that
  $w'\in \words(N'_\statec)$. By definition of $\psi_j$,
  $N'_\statec\models \psi_j$. Thus, since $N_\stateb\models
  \translate{\varphi}_j$, we deduce that $N'_\statec\models
  \translate{\varphi}_{j+1}$. There are two cases:
      \begin{itemize}
      \item $j=n$ (\emph{base step}): by construction,
        $\words(N'_\statec)$ is a singleton,
        $\translate{\varphi}_{n+1}=\Upsilon(\varphi_{n+1})$, and
        $\varphi_{n+1}$ is a quantification-free $\QLTL$-formula in
        positive normal form. Since $w'\in \words(N'_\statec)$ and
        $N'_\statec\models \translate{\varphi}_{n+1}$, by
        Lemma~\ref{lemma:basicNonElem}, it follows that $w'\models
        \varphi_{n+1}$.
      \item $j\leq n-1 $ (\emph{induction step}): since $w'\in
        \words(N'_\statec)$ and $N'_\statec\models
        \translate{\varphi}_{j+1}$, by the induction hypothesis (note
        that since $N'_\statec$ is minimal, for each $j$-segment
        $N''_\stated$ of $N'_\statec$,
        $\words(N''_\stated)=\words(N'_\statec)$), it follows that
        $w'\models \varphi_{j+1}$.
      \end{itemize}
      Thus, in both cases $w'\models \varphi_{j+1}$.  Since $w'$ is an
      arbitrary infinite word over $2^{P}$ such that
      $w'\weq{P\setminus\{p_j\}}w$, we
      obtain that $w\models \qall \,p_j.\varphi_{j+1}=\varphi_j$,
      and the result follows.\\

\noindent (2) $\varphi_j=\qis \,p_j.\varphi_{j+1}$, $\translate{\varphi}_j=\F
      (\psi_j\,\wedge\,\translate{\varphi}_{j+1})$, and $j\leq n-1$
      (\emph{induction step}): since $M_\state\models
      \translate{\varphi}_j$, there is a refinement $N_{\stateb}$ of
      $M_\state$ satisfying both $ \psi_j$ and
      $\translate{\varphi}_{j+1}$. By definition of $\psi_j$,
      $N_{\stateb}$ is well-formed w.r.t.\ $j$. Thus, since
      $N_{\stateb}\models\translate{\varphi}_{j+1}$ and $j\leq n-1$,
      by the induction hypothesis, there is a $j$-segment
      $N'_{\statec}$ of $N_{\stateb}$ such that
      $N'_{\statec}\models\psi_j$,
      $N'_{\statec}\models\translate{\varphi}_{j+1}$, and for each
      $w\in \words(N'_{\statec})$, $w\models \varphi_{j+1}$. Since
      $N_{\stateb}$ is a refinement of $M_\state$, it easily follows
      that $N'_{\statec}$ is the refinement of some $(j-1)$-segment
      $M'_{\stated}$ of $M_\state$.  Since
      $N'_{\statec}\models\psi_j\wedge \translate{\varphi}_{j+1}$, it
      holds that $M'_{\stated}\models \translate{\varphi}_j$. Hence,
      it suffices to show that for each $w\in \words(M'_{\stated})$,
      $w\models \varphi_j$. Let $w\in \words(M'_{\stated})$. Then,
      since $M'_{\stated}$ (resp., $N'_{\statec}$) is \emph{minimal}
      and well-formed w.r.t. $j-1$ (resp., $j$) and $N'_{\statec}$ is
      a refinement of $M'_{\stated}$, it follows that there is $w'\in
      \words(N'_{\statec})$ such that
      $w'\weq{P\setminus\{p_j\}}w$. Since $w'\models \varphi_{j+1}$,
      we obtain that $w\models \qis \,p_j.\varphi_{j+1}=\varphi_j$,
      and the result follows.
\end{proof}

Now, we can prove the correctness of the construction.

\begin{theorem}\label{theo:correcNonElem} $\varphi$ is
  satisfiable if, and only if, $\translate{\varphi}$ is satisfiable.
\end{theorem}
\begin{proof}
  First, assume that $\translate{\varphi}=\psi_0\wedge
  \translate{\varphi}_1$ is satisfiable. Hence, there is a pointed
  model $M_\state$ which satisfies both $\psi_0$ and
  $\translate{\varphi}_1$. By definition of formula $\psi_0$, it
  follows that $M_\state$ is well-formed w.r.t.\ $0$. Since
  $M_\state\models \translate{\varphi}_1$, by
  Lemma~\ref{lemma:Aux1NonElem}, we deduce that there is an infinite
  word $w$ over $2^{P}$ such that $w\models \varphi_1$. Since
  $\varphi=\varphi_1$, it follows that $\varphi$ is satisfiable.

  Now, assume that $\varphi$ is satisfiable. Since any proposition in
  $P$ does not occur free in $\varphi$, it follows that for each
  infinite word $w$ over $2^{P}$, $w\models \varphi$. Let $M_\state$
  be any pointed model which is well-formed w.r.t.\ $0$. By definition
  of formula $\psi_0$, it holds that $M_\state\models
  \psi_0$. Moreover, since $w\models \varphi$ for each $w\in
  \words(M_\state)$, and $\varphi=\varphi_1$, by
  Lemma~\ref{lemma:Aux2NonElem} it follows that $M_\state\models
  \translate{\varphi}_1$. Therefore, $M_\state\models \psi_0\wedge
  \translate{\varphi}_1=\translate{\varphi}$. Hence,
  $\translate{\varphi}$ is satisfiable.
\end{proof}


By using 
Theorem~\ref{theo:lowerBoundRefCTL} and the fact that there exists a
linear time translation of $\RCTL$ into
$\lang^\mu_\G$ (see page~\pageref{page-CTLtoLmu}) we now obtain the required proof of 
Theorem~\ref{theo:RLmuNonElem}.

\weg{
\subsection{Upper bound for $\lang_\G$}
\label{sec-tableau}
Given any $\phi \in \lang_\G$, we describe a tableau
construction that either constructs a model for $\phi$, or reports
that $\phi$ is not satisfiable. This construction is provided for formulas in \emph{positive normal form}.
\begin{definition}
A formula is in \emph{positive normal form} if it is built from the following abstract syntax. 
$$\phi::=\top\ |\ \bot\\ |\ p\ |\ \lnot p\ |\ \phi\land\phi\ |\ \phi\lor\phi\ | \know\phi\ |\ \susp\phi\ |\ \G\phi\ |\ \F\phi$$
\end{definition}
However, we note that every $\lang_\G$ formula may be converted into positive normal
form with linear change to the size of formula.
\begin{center}
{\huge\fbox{TIM @@previous source after end{document} }}
\end{center}
\paragraph{Tableau Definition:}~

\begin{lemma}\label{tableauSoundness}
If the tableau reports that $\phi$ is satisfiable, then $\phi$ has a model.
\end{lemma}

\begin{proof}
~
\end{proof}

\begin{lemma}\label{tableauCompleteness} 
If $\phi$ is satisfiable, then the tableau reports that $\phi$ is satisfiable.
\end{lemma}
\begin{proof}
~
\end{proof}

\begin{corollary}
The satisfiability problem for $\lang_\G$ is in $\TWOEXPTIME$.
\end{corollary}
\begin{proof}
~
\end{proof}
}

\subsection{Succinctness}\label{succinctness}
\label{sec-succinctness}

In this section we establish the following result.

\begin{theorem}\label{theorem:mainSuccinctness}
$\logicRML$ is doubly exponentially more succinct than $\logicK$, and $\logicRML^{\mu}$ is doubly exponentially more succinct than modal $\mu$-calculus.
\end{theorem}

Theorem~\ref{theorem:mainSuccinctness} directly follows from the following result whose proof is given in the rest of this section.

\begin{proposition}\label{proposition:mainSuccinctness}There is a finite set $P$ of propositional variables and a family 
  $(\varphi_n)_{n\in\N}$ of one-agent $\lang_\G$ formulas over $P$
  such that for each $n\in\N$, $\varphi_n$ has size $O(n^{2})$ and
  refining quantifier alternation depth $2$, and each equivalent one-agent
  $\lang^{\mu}$ formula has size at least
  $2^{2^{\Omega(n)}}$.\footnote{Recall that $f(n) \in \Omega(g(n))$ iff
    $g(n) \in O(f(n))$.}
\end{proposition}

\noindent \textbf{Construction of the $\lang_\G$ formulas $\varphi_n$
  in Proposition~\ref{proposition:mainSuccinctness}:} let
$P=\{l,r,\#,0,1,a,b\}$. An \emph{$n$-configuration} is a string on
$\{a,b\}$ of length exactly $2^{2^{n}}$. We define a class
$\mathcal{C}_n$ of pointed models, where each pointed model in the
class encodes, in a suitable way, a pair of $n$-configurations. Then, we
construct the $\lang_\G$ formula $\varphi_n$ in such a way that the
following holds: a pointed model $M_\state\in \mathcal{C}_n$ satisfies
$\varphi_n$ iff the two $n$-configurations encoded by $M_\state$
coincide. In order to formally define the class $\mathcal{C}_n$, we
need additional definitions.  An $n$-block is a pair $bl=(c,i)$ such
that $c\in\{a,b\}$ and $1\leq i\leq 2^{2^{n}}$. We say that $c$ is the
\emph{content} of $bl$ and $i$ is the \emph{position} of
$bl$. Intuitively, $bl$ represents the $i$th symbol of some
$n$-configuration. First, we define an encoding of $(c,i)$ by a set
$\code(c,i)$ of strings over $2^{P}$ of length $n+3$. Since $1\leq
i\leq 2^{2^{n}}$, $i$ can be encoded by a binary string over $\{0,1\}$
of length exactly $2^{n}$. Moreover, we keep track, for each $1\leq
j\leq 2^{n}$, of the binary encoding (a string over $\{0,1\}$ of
length $n$)\footnote{Here, it is not relevant to specify the form of
  the binary encoding which is used.} of the position $j$ of the $j$th
bit in the binary encoding of $i$. This leads to the following
definition. An $n$-\emph{sub-block} is a string over $2^{P}$ of length
$n+2$ of the form $sbl=\{\#\},\{b_1\},\ldots,\{b_n\},\{B\}$, where
$b_1,\ldots,b_n,B\in\{0,1\}$. The \emph{content} of $sbl$ is $B$ and
the \emph{position} of $sbl$ is the integer $1\leq j\leq 2^{n}$ whose
binary encoding is $b_1,\ldots,b_n$. Intuitively, $sbl$ encodes the
position and the content $B$ of a bit along the binary encoding of an
integer $1\leq i\leq 2^{2^{n}}$. Then, $\code(c,i)$ is the set of
strings over $2^{P}$ of length $n+3$ such that
\begin{itemize}
\item for each $u\in \code(c,i)$, $u=sbl\cdot \{c\}$, where $sbl$ is an
  $n$-sub-block whose position $j$ and content $b$ satisfy the
  following: $b$ is the $j$th bit in the binary encoding of $i$.
\item for each $1\leq j\leq 2^{n}$, let $B_j$ be the $j$th bit in the
  binary encoding of $i$ and $sbl_j$ be the $n$-sub-block whose
  position is $j$ and whose content is $B_j$. Then, $sbl_j\cdot
  \{c\}\in \code(c,i)$.
\end{itemize}

\noindent Let $M_\state$ be a pointed model over $P$. We denote by
$\Traces(M_\state)$ the set of finite or infinite strings over $2^{P}$
of the form $(V^{M})^{-1}(\state_0),(V^{M})^{-1}(\state_1),\ldots$
such that $\state_0,\state_1,\ldots$ is a maximal path of $M$ starting
from $\state$.  A pointed model $M_\state$ \emph{encodes an $n$-block
  $(c,i)$} if 
$$\Traces(M_\state)=\code(c,i) \text{ and } M_\state \models \displaystyle{\bigwedge_{d=0}^{n-1}\know^{d}(\M
  1\wedge\M 0)} \in \lang$$
Note that the set of pointed models encoding $(c,i)$
is nonempty. Let $(w_l,w_r)$ be a pair of $n$-configurations.  A
pointed model $M_\state$ \emph{encodes the pair} $(w_l,w_r)$ if it
holds that:
\begin{itemize}
\item $s$ has two successors $\state_l$ and $\state_r$ (called the
  left successor and right successor of $s$, respectively). Moreover,
  $(V^{M})^{-1}(\state)=\emptyset$, $(V^{M})^{-1}(\state_l)=\{l\}$ and
  $(V^{M})^{-1}(\state_r)=\{r\}$;
\item for each $dir\in\{l,r\}$, $s_{dir}$ has $2^{2^{n}}$ successors
  $\state_{1,dir},\ldots,\state_{2^{2^{n}},dir}$. Moreover, for each
  $1\leq i\leq 2^{2^{n}}$, $M_{\state_{i,dir}}$ encodes the $n$-block
  $(c_{i,dir},i)$, where $c_{i,dir}$ is the $i$th symbol of the
  $n$-configuration $w_{dir}$.
\end{itemize}
If additionally $w_l=w_r$, then we say that $M_\state$ is
\emph{balanced}. The class $\mathcal{C}_n$ is the class of pointed
models $M_\state$ such that $M_\state$ encodes some pair $(w_l,w_r)$
of $n$-configurations. In order to define the $\lang_\G$ formula
$\varphi_n$ (for each $n\geq 0$), we first show
Lemma~\ref{lemma:constructionSuccinctness1}. This lemma asserts that there is an
$\lang_\G$ formula $\psi_n$ of size $O(n^{2})$ which allows one to
select, for a given pointed model $M_\state\in\mathcal{C}_n$, only the
$n$-blocks encoded by $M_\state$ having the same position.

\begin{lemma}\label{lemma:constructionSuccinctness1} For each $n\geq
  0$, one can construct a one-agent $\lang_\G$ formula $\psi_n$ of
  size $O(n^{2})$ and refinement quantifier alternation depth $1$ satisfying the
  following for all pairs $(w_l,w_r)$ of $n$-configurations: for each
  $M_\state\in \mathcal{C}_n$ encoding the pair $(w_l,w_r)$ and each
  refinement $M'_{\state'}$ of $M_{\state}$,
\begin{itemize}
\item $M'_{\state'}$ satisfies $\psi_n$ \emph{iff} there is $1\leq
  i\leq 2^{2^{n}}$ such that the set of $\#$-states (i.e.\ states whose label is $\{\#\}$) $\state'_\#$
  reachable from $\state'$ is nonempty and for each of such states
  $\state'_\#$, $M'_{\state'_\#}$ encodes an $n$-block whose position
  is $i$ and whose content is either the $i$th symbol of $w_l$ or the
  $i$th symbol of $w_r$.
\end{itemize}
\end{lemma}
\begin{proof}
The $\lang_\G$ formula $\psi_n$ is defined as follows:
$$\psi_n:=\xi_n\,\wedge\,\G(\theta_n\rightarrow \bigvee_{b\in\{0,1\}}\know^{n+3}b)$$
where $\xi_n$ and $\theta_n$ are $\lang$ formulas defined as follows:
$$\xi_n:=\M\top\wedge\know\M\top\wedge\bigwedge_{d=0}^{n-1}\know^{d+2}(\M 1\wedge\M 0)\wedge \know^{n+2}\M\top\wedge\know^{n+3}\M\top$$
$$\theta_n:=\M\top\wedge\know\M\top\wedge\know^{2}\M\top\wedge\bigwedge_{d=1}^{n}\bigvee_{b\in\{0,1\}}\know^{d+2}(b\wedge \M \top)\wedge \know^{n+3}\M\top$$
Note that $\psi_n$ has size $O(n^{2})$ and that $\ad(\psi_n) =1$
(refinement alternation depth). Thus, it remains to prove the second
part of the lemma.  Fix $M_\state\in \mathcal{C}_n$ encoding some pair
$(w_l,w_r)$ of $n$-configurations, and let $M'_{\state'}$ be a
refinement of $M_\state$. By construction, for each $\#$-state
$\state'_\#$ reachable from $\state'$ in $M'$, there is a $\#$-state
$\state_\#$ reachable from $\state$ in $M$ such that $M'_{\state'_\#}$
is a refinement of $M_{\state_\#}$. Moreover, $M_{\state_\#}$ encodes
some $n$-block $(c,i)$, where the content $c$ is either the $i$th
symbol of $w_l$ or the $i$th symbol of $w_r$. Thus, by definition of
$\xi_n$, we obtain the following.

\paragraph{Fact 1:} $M'_{\state'}$ satisfies $\xi_n$ \emph{iff}
the set of $\#$-states $\state'_\#$ reachable from $\state'$ is
nonempty and for each of such states $\state'_\#$, $M'_{\state'_\#}$
encodes some $n$-block $(c,i)$, where the content $c$ is either the
$i$th symbol of $w_l$ or the $i$th symbol of $w_r$.\vspace{0.2cm}

In the second conjunct $\G(\theta_n\rightarrow
\bigvee_{b\in\{0,1\}}\know^{n+3}b)$ of the definition of $\psi_n$, the
formula $\theta_n$ intuitively enforces one to select the refinements
$M'_{\state'}$ of $M_\state$ encoding only $n$-blocks having the same
position. Formally, by definition of $\theta_n$, we obtain the
following.

\paragraph{Fact 2:} Let $M''_{\state''}$ be a refinement of
$M'_{\state'}$. Then, $M''_{\state''}$ satisfies $\theta_n$ \emph{iff}
for all $u,u'\in \Traces(M''_{\state''})$, $u,u'\in
\Traces(M_{\state})$ and the $n$-sub-block in $u$ and the
$n$-sub-block in $u'$ have the same position.\vspace{0.2cm}

Thus, by Fact 2 it follows that the second conjunct
$\G(\theta_n\rightarrow \bigvee_{b\in\{0,1\}}\know^{n+3}b)$ of
definition of $\psi_n$ requires that all the $n$-sub-blocks in
$\Traces(M'_{\state'})$ having the same position have also the same
content, i.e., all the $n$-blocks encoded by $M'_{\state'}$ have the
same position. Thus, by Fact 1~the result follows. 
\end{proof}
For each $n\geq 0$, let $\psi_n$ be the $\lang_\G$ formula  satisfying the statement of Lemma~\ref{lemma:constructionSuccinctness1}. Then, the
one-agent $\lang_\G$ formula $\varphi_n$ is defined as follows:
$$\varphi_n=\G(\psi_n\rightarrow \bigvee_{c\in\{a,b\}}\know^{n+4}c)$$

\noindent By construction and
Lemma~\ref{lemma:constructionSuccinctness1}, we easily obtain the
following result.

\begin{lemma}\label{lemma:constructionSuccinctness2} For each $n\geq
  0$, the $\lang_\G$ formula $\varphi_n$ has size $O(n^{2})$ and
  $\ad(\varphi_n)=2$ (refinement alternation depth). Moreover, for
  each $M_\state\in \mathcal{C}_n$, $M_\state$ satisfies $\varphi_n$
  \emph{iff} $M_\state$ is balanced.
\end{lemma}

\paragraph{Proof of Proposition~\ref{proposition:mainSuccinctness}:}
by Lemma~\ref{lemma:constructionSuccinctness2}, in order to complete
the proof of Proposition~\ref{proposition:mainSuccinctness}, we need
to show that for each $n\geq 0$, each one-agent $\lang^{\mu}$ formula
equivalent to $\varphi_n$ has size at least $2^{2^{\Omega(n)}}$. For
this, we use a well-known automata-characterization of (one-agent)
$\lang^{\mu}$ in terms of \emph{parity symmetric alternating
  $($finite-state$)$ automata} ($\PSAA$) which operate on pointed
models~\cite{Wil99}.  First, we recall the class of $\PSAA$. We need
additional definitions.

A \emph{tree}
$T$ is a prefix closed subset of $\mathbb{N}^{*}$. The elements of
$T$ are called \emph{nodes} and the empty word $\varepsilon$ is
the \emph{root} of $T$. For $x\in T$, the set of \emph{children}
 of $x$ (in $T$) is $\{x\cdot i \in T\mid i\in \mathbb{N}\}$. A \emph{path}
 of  $T$ is a maximal sequence $\pi=x_0 x_1\ldots$ of $T$-nodes such that $x_0=\varepsilon$ and for any $i$, $x_{i+1}$ is a child of $x_i$.
For an alphabet
$\Sigma$, a $\Sigma$-labeled tree is a pair $\tpl{T,r}$ where $T$
is a tree and $r:T \rightarrow \Sigma$.
For a  set $X$,
$\B_+(X)$   denotes the set
of \emph{positive} boolean formulas over $X$,  built from elements in $X$
using $\vee$ and $\wedge$ (we also allow the formulas $\true$ and
$\false$). A subset $Y$ of $X$  \emph{satisfies}
$\theta\in\B_+(X)$ iff the truth assignment that assigns $\true$
to the elements in $Y$ and $\false$ to the elements of $X\setminus
Y$ satisfies $\theta$.

A \emph{parity symmetric alternating automaton ($\PSAA$)} over $P$ is a tuple
 $\mathcal{A}=\tpl{P,Q,q_0,\delta,Acc}$, where
$Q$ is a finite set of locations, $q_0\in Q$ is an initial location,
$\delta:Q\times 2^{P}\rightarrow \B_+\bigl(\{\know,\M\}\times Q)$ is
the transition function, and $Acc: Q\rightarrow \N$ is a parity acceptance condition assigning to each location $q\in Q$ an integer (called \emph{priority}).
Intuitively,  a target of a move of
$\mathcal{A}$ is encoded by an element in $\{\know,\M\}\times Q$. 
An atom $(\Diamond,q)$ means that from the current state $s$ (of the  
pointed input model) $A$ moves to some successor of $s$ and the  
location is updated to $q$. On the other hand, an atom $(\Box,q)$  
means that from the current state $s$ the automaton splits in  
multiple copies and, for each successor $s'$ of $s$, one of such copies  
moves to $s'$ and the location is updated to $q$.

Formally,
for a pointed model $M_{\state_0}$ over $P$, a \emph{run} of $\mathcal{A}$ over $M_{\state_0}$ is a
$(Q\times \States^{M})$-labeled tree $\tpl{T,r}$.\footnote{Intuitively, each node  of $T$ labeled by $(q,\state)$ describes a copy of $\mathcal{A}$ that is in  location $q$ and reads the state $\state$ of
$M$.} Moreover, we require that $r(\varepsilon)=(q_0,\state_0)$ (initially, $\mathcal{A}$ is in the initial location $q_0$ reading state $\state_0$), and for each  $y\in T$ with $r(y)=(q,\state)$, there is a (possibly empty) \emph{minimal} set
  $H\subseteq \{\know,\M\}\times Q$  satisfying $\delta(q,(V^{M})^{-1}(\state))$ such that  the set $L(y)$ of labels of children of $y$ in $T$ is the smallest set satisfying the following:   for all atoms $at\in H$,
  \begin{itemize}
    \item if $at=(\M,q')$, then for some successor $\state'$ of $\state$ in $M$, $(q',\state')\in L(y)$;
    \item if $at=(\know,q')$, then for each successor $\state'$ of $\state$ in $M$, $(q',\state')\in L(y)$.
  \end{itemize}

  For an infinite path $\pi=y_0y_1\ldots$ of $T$, let $inf(\pi)$ be
  the set of locations in $Q$ that appear in $r(y_0)r(y_1)\ldots$
  infinitely often. The run $\tpl{T,r}$ is \emph{accepting} if for
  each infinite path $\pi$ of $T$, the smallest priority of the
  locations in $inf(\pi)$ is \emph{even}. The \emph{language} of
  $\mathcal{A}$ is the set of pointed models $M_\state$ over $P$ such
  that $\mathcal{A}$ has an accepting run over $M_\state$. The
  following is a well-known result.

 \begin{proposition}\label{proposition:fromLmutoPSAA}\cite{Wil99}
   Given a one-agent $\lang^\mu$ formula $\varphi$ over $P$, one can
   construct a $\PSAA$ $\mathcal{A}_\varphi$ with $O(|\varphi|)$
   locations whose language is the set of pointed models over $P$
   satisfying $\varphi$.
\end{proposition}

Proposition~\ref{proposition:mainSuccinctness} directly follows from Proposition~\ref{proposition:fromLmutoPSAA} and the following result.

\begin{lemma} Let $n\geq 0$ and $\mathcal{A}_{n}$ be a $\PSAA$ over
  $P$ whose language is the set of pointed models satisfying the
  $\lang_\G$ formula $\varphi_n$. Then, the number of locations of
  $\mathcal{A}_{n}$ is at least $2^{2^{n}}$.
\end{lemma}
\begin{proof} Let $n\geq 0$ and $\mathcal{A}_{n}$ as in the statement
  of the lemma (note that $\mathcal{A}_{n}$ exists by
  Proposition~\ref{proposition:fromLmutoPSAA} together with
  Proposition~\ref{lem:express-mu}), and $Q$ be the set of
  $\mathcal{A}_n$-locations. For each $n$-configuration $w$, let
  $M^{w}_{\state_w}$ be some balanced pointed model encoding the pair
  $(w,w)$, and $H(w)$ be the set of sets $Q_l\subseteq Q$ such that
  there is an accepting run $\tpl{T,r}$ of $\mathcal{A}_{n}$ over the
  pointed model $M^{w}_{\state_w}$ so that:
\begin{itemize}
\item $Q_l$ is the set of locations associated with the replicas of
  $\mathcal{A}_n$ in the run $\tpl{T,r}$ which read the left successor
  $\state_l$ of $\state_w$ in $M^{w}$, i.e., $Q_{l}=\{q\in Q\mid$ for
  some $x\in T,\,r(x)=(q,\state_{l})\}$. (Note that $Q_l=\emptyset$ if $\tpl{T,r}$ does not visit the left  successor $s_l$ of $s_w$.)
\end{itemize}
First, we show that $H(w)\neq \emptyset$. By hypothesis and Lemma \ref{lemma:constructionSuccinctness2},  
there must exist some accepting run of $\mathcal{A}_n$ over the input
$M^{w}_{s_w}$. Now, by construction,  $H(w)$ is a set of subsets of  
$Q$, and $H(w)$ is non-empty if and only if there is some accepting  
run of $\mathcal{A}_n$ over $M^{w}_{s_w}$. (If no  
accepting run of $\mathcal{A}_n$ visits the left successor $s_l$ of  
$s_w$ in $M^{w}$, then
$H(w)$ is a singleton containing just the empty set.) Hence,  
non-emptiness of $H(w)$ follows. Next, we prove the following.

\vspace{0.2cm}

\noindent \textbf{Claim:} for all $n$-configurations $w$ and $w'$ such
that $w\neq w'$, $H(w)\cap H(w')=\emptyset$.\vspace{0.1cm}

\noindent \textbf{Proof of the claim:} for a model $M$ and a set
$\States'\subseteq \States^{M}$, the \emph{restriction} of $M$ to $S'$
is defined in the obvious way. For $s\in\States^{M}$, let $[M_\state]$
denote the restriction of $M$ to the set of states reachable from
$\state$ in $M$. For all $n$-configurations $w$ and $dir\in\{l,r\}$,
let $\state_{w,dir}$ be the $dir$-successor of $\state_w$ in
$M^{w}$. We prove the claim by contradiction. So, assume that there
are two distinct $n$-configurations $w$ and $w'$ such that $H(w)\cap
H(w')\neq \emptyset$. Without loss of generality we can assume that
$M^{w}$ and $M^{w'}$ have no states in common.  Let
$M^{w,w'}_{\state_w}$ be any pointed model satisfying the following:
the successors of $\state_w$ in $M^{w,w'}$ are $\state_{w',l}$ and
$\state_{w,r}$, and
$[M^{w,w'}_{\state_{w',l}}]=[M^{w'}_{\state_{w',l}}]$ and
$[M^{w,w'}_{\state_{w,r}}]=[M^{w}_{\state_{w,r}}]$. Evidently,
$M^{w,w'}_{\state_w}$ is a pointed model encoding the pair
$(w',w)$. Since $w\neq w'$, by hypothesis and
Lemma~\ref{lemma:constructionSuccinctness2}, $\mathcal{A}_n$ does not
accept $M^{w,w'}_{\state_w}$. On the other hand, since there is $Q\in
H(w)\cap H(w')$, by definition of the sets $H(w)$ and $H(w')$ and the
semantics of $\PSAA$, it easily follows that there is an accepting run
of $\mathcal{A}_n$ over $M^{w,w'}_{\state_w}$, which is a
contradiction. Hence, the claim holds.
\vspace{0.2cm}

By the claim above, it follows that for each $n$-configuration $w$, there is $Q_w\in H(w)$ (recall that $H(w)\neq \emptyset$) such that
for all $n$-configurations $w'$ distinct from $w$, $Q_w\notin H(w')$. Since the number of distinct $n$-configurations is $2^{2^{2^{n}}}$ and the number of subsets of $Q$ is $2^{|Q|}$, we obtain that $|Q|\geq 2^{2^{n}}$, and the result holds. 
\end{proof}

\section{Conclusions and perspectives} \label{sec.conclusions}

\paragraph{Conclusions}
We conclude that we hope to have established a platform for structural refinement in various modal logics. We established results on axiomatization, complexity, expressivity, and we gave applications to software verification and design, and to dynamic epistemic logics. We clearly established the relation to bisimulation quantified logics: refinement quantification is bisimulation followed by relativization. The multi-agent refinement modal logic and the furthest generalization in the form of refinement $\mu$-calculus are only the beginning. One could think of refinement CTL, refinement PDL, (yet other) refinement epistemic logics, refinement with further structural restrictions or with protocol restrictions, and so on. Each of these logics may have different axiomatizations and complexities, and equal expressivity as the logic without refinement is certainly not to be expected; e.g., we estimate that refinement modal logic is more expressive than the base modal logic on the ${\mathcal KT}$ model class.


\paragraph{Recent results}
Following the initial submission of the paper, some further results have been obtained in this area, typically involving one of the authors. In \cite{DBLP:conf/jelia/BozzelliDP12} it was established that the complexity of refinement modal
logic for a single agent is \LINAEXPTIME-complete, which means that
the satisfiability of an $\logicRML$ formula can be decided by an
exponential-time bounded Alternating Turing Machine with a
linearly-bounded number of alternations. In \cite{hales.aiml:2012} an axiomatization of the multi-agent refinement modal logic of knowledge is given, among other results. As a generalization of quantifying over announcements (arbitrary announcements), in \cite{balbianietal:2008} semantics were also given for quantifying over action models and the question was posed how to axiomatize this logic: in \cite{hales2013arbitrary} it is shown that quantifying over action models is equally expressive as the refinement quantifier, i.e., `there is a refinement after which $\phi$ is true' means the same as `there is an action model such that after its execution $\phi$ is true'. This answers one of the open questions on logics with quantification over information change, posed in the recent survey \cite{hvd.wollic:2012}. That survey also puts various other proposals on propositional quantification in perspective, such as \cite{jfak.sabotage:2005}, \cite{economou:2010}, \cite{aucher.planning:2012} (going back to \cite{aucher:2010}), and \cite{wenetal:2011}---for details, see \cite{hvd.wollic:2012}. It should not be forgotten to mention that many of these, including our own proposal, go back to the original publication \cite{Fine:1970}.

\paragraph{Further research}
We wish to determine the complexity of model checking in the various refinement modal logics. On the further horizon loom the detailed investigation of {\em other} refinement logics, mainly refinement PDL and refinement CTL, and the exploration of their applications. The relation of refinement quantification to other forms of propositional quantification over information change also needs further investigation.

\section*{Acknowledgements}


We acknowledge a very insightful and detailed review from a journal referee. Hans van Ditmarsch is also affiliated to IMSc (Institute of Mathematical Sciences), Chennai, India. We acknowledge support from ERC project EPS 313360, and from EU 7th Framework Programme under grant agreement no.\ 295261 (MEALS).

\bibliographystyle{plain}
\bibliography{biblio2013,Sophie}

\begin{thebibliography}{10}

\bibitem{aczel:1988}
P.~Aczel.
\newblock {\em Non-Well-Founded Sets}.
\newblock CSLI Publications, Stanford, CA, 1988.
\newblock CSLI Lecture Notes 14.

\bibitem{alur98}
R.~Alur, T.~A. Henzinger, and O.~Kupferman.
\newblock Alternating-time temporal logic.
\newblock {\em Lecture Notes in Computer Science}, 1536:23--60, 1998.

\bibitem{alur98alternating}
Rajeev Alur, Thomas~A. Henzinger, Orna Kupferman, and Moshe~Y. Vardi.
\newblock Alternating refinement relations.
\newblock In {\em International Conference on Concurrency Theory}, pages
  163--178, 1998.

\bibitem{AHLNW-08}
Adam Antonik, Michael Huth, Kim~G. Larsen, Ulrik Nyman, and Andrzej Wasowski.
\newblock 20 years of modal and mixed specifications.
\newblock {\em Bulletin of European Association of Theoretical Computer
  Science}, 1(94), 2008.

\bibitem{arnoldetal:2001}
A.~Arnold and D.~Niwinski.
\newblock {\em Rudiments of $\mu$-calculus}.
\newblock North Holland, 2001.

\bibitem{aucher:2010}
G.~Aucher.
\newblock Characterizing updates in dynamic epistemic logic.
\newblock In {\em Proceedings of Twelfth KR}. AAAI Press, 2010.

\bibitem{aucher.planning:2012}
G.~Aucher.
\newblock {DEL}-sequents for regression and epistemic planning.
\newblock {\em Journal of Applied Non-Classical Logics}, 22(4):337--367, 2012.

\bibitem{balbianietal:2008}
P.~Balbiani, A.~Baltag, H.~van Ditmarsch, A.~Herzig, T.~Hoshi, and T.~De Lima.
\newblock `{K}nowable' as `known after an announcement'.
\newblock {\em Review of Symbolic Logic}, 1(3):305--334, 2008.

\bibitem{baltagetal:1998}
A.~Baltag, L.S. Moss, and S.~Solecki.
\newblock The logic of public announcements, common knowledge, and private
  suspicions.
\newblock In {\em Proc.\ of 7th TARK}, pages 43--56. Morgan Kaufmann, 1998.

\bibitem{bilkovaetal:2008}
M.~Bilkova, A.~Palmigiano, and Y.~Venema.
\newblock Proof systems for the coalgebraic cover modality.
\newblock In Carlos Areces and Robert Goldblatt, editors, {\em Advances in
  Modal Logic}, pages 1--21. College Publications, 2008.

\bibitem{blackburnetal:2001}
P.~Blackburn, M.~de~Rijke, and Y.~Venema.
\newblock {\em Modal Logic}.
\newblock Cambridge University Press, Cambridge, 2001.
\newblock Cambridge Tracts in Theoretical Computer Science 53.

\bibitem{DBLP:conf/jelia/BozzelliDP12}
Laura Bozzelli, Hans~P. van Ditmarsch, and Sophie Pinchinat.
\newblock The complexity of one-agent refinement modal logic.
\newblock In Luis~Fari{\~n}as del Cerro, Andreas Herzig, and J{\'e}r{\^o}me
  Mengin, editors, {\em JELIA}, volume 7519 of {\em Lecture Notes in Computer
  Science}, pages 120--133. Springer, 2012.

\bibitem{browneetal:1987}
M.~Browne, E.~Clarke, and O.~Gr\"umberg.
\newblock Characterizing {K}ripke structures in temporal logic.
\newblock In H.~Ehrig, R.~Kowalski, G.~Levi, and U.~Montanari, editors, {\em
  TAPSOFT '87}, LNCS 249, pages 256--270. Springer, 1987.

\bibitem{CE81}
E.M. Clarke and E.A. Emerson.
\newblock Design and {V}erification of {S}ynchronization {S}keletons using
  {B}ranching {T}ime {T}emporal {L}ogic.
\newblock In {\em Proceedings of Workshop on Logic of Programs}, LNCS 131,
  pages 52--71. Springer-Verlag, 1981.

\bibitem{dagostinoetal:2000}
G.~d'Agostino and M.~Hollenberg.
\newblock Logical questions concerning the $\mu$-calculus: Interpolation,
  {L}yndon and {L}os-{T}arski.
\newblock {\em Journal of Symbolic Logic}, 65(1):310--332, 2000.

\bibitem{dagostinoetal:2005}
G.~d'Agostino and G.~Lenzi.
\newblock An axiomatization of bisimulation quantifiers via the
  {$\mu$}-calculus.
\newblock {\em Theor. Comput. Sci.}, 338(1-3):64--95, 2005.

\bibitem{dagostinoetal:2008}
G.~d'Agostino and G.~Lenzi.
\newblock A note on bisimulation quantifiers and fixed points over transitive
  frames.
\newblock {\em J. Log. Comput.}, 18(4):601--614, 2008.

\bibitem{economou:2010}
P.~Economou.
\newblock {\em Extensions and Applications of Dynamic Epistemic Logic}.
\newblock PhD thesis, Oxford University, 2010.

\bibitem{faginetal:1995}
R.~Fagin, J.Y. Halpern, Y.~Moses, and M.Y. Vardi.
\newblock {\em Reasoning about Knowledge}.
\newblock MIT Press, Cambridge MA, 1995.

\bibitem{jdeds-feuillade-pinchinat07}
Guillaume Feuillade and Sophie Pinchinat.
\newblock Modal specifications for the control theory of discrete-event
  systems.
\newblock {\em Discrete Event Dynamic Systems}, 17(2):181--205, 2007.

\bibitem{Fine:1970}
K.~Fine.
\newblock Propositional quantifiers in modal logic.
\newblock {\em Theoria}, 36(3):336--346, 1970.

\bibitem{french:2006}
T.~French.
\newblock {\em Bisimulation quantifiers for modal logic}.
\newblock PhD thesis, University of Western Australia, 2006.

\bibitem{frenchetal:2008}
T.~French and H.~van Ditmarsch.
\newblock Undecidability for arbitrary public announcement logic.
\newblock In C.~Areces and R.~Goldblatt, editors, {\em Advances in Modal Logic
  7}, pages 23--42, London, 2008. College Publications.
\newblock Proc.\ of the seventh conference ``Advances in Modal Logic''.

\bibitem{gerbrandyetal:1997}
J.D. Gerbrandy and W.~Groeneveld.
\newblock Reasoning about information change.
\newblock {\em Journal of Logic, Language, and Information}, 6:147--169, 1997.

\bibitem{hales:2011}
J.~Hales.
\newblock Refinement quantifiers for logics of belief and knowledge.
\newblock Honours Thesis, University of Western Australia, 2011.

\bibitem{hales2013arbitrary}
J.~Hales.
\newblock Arbitrary action model logic and action model synthesis.
\newblock In {\em Proc.\ of 28th {LICS}}, pages 253--262. IEEE, 2013.

\bibitem{halesetal:2011}
J.~Hales, T.~French, and R.~Davies.
\newblock Refinement quantified logics of knowledge.
\newblock {\em Electr. Notes Theor. Comput. Sci.}, 278:85--98, 2011.

\bibitem{hales.aiml:2012}
J.~Hales, T.~French, and R.~Davies.
\newblock Refinement quantified logics of knowledge and belief for multiple
  agents.
\newblock In {\em Advances in Modal Logic 9}, pages 317--338. College
  Publications, 2012.

\bibitem{hareletal:2000}
D.~Harel, D.~Kozen, and J.~Tiuryn.
\newblock {\em Dynamic Logic}.
\newblock MIT Press, Cambridge MA, 2000.
\newblock Foundations of Computing Series.

\bibitem{janinetal:1995}
D.~Janin and I.~Walukiewicz.
\newblock Automata for the modal mu-calculus and related results.
\newblock In {\em Proc.\ of 20th MFCS}, LNCS 969, pages 552--562. Springer,
  1995.

\bibitem{janin96}
D.~Janin and I.~Walukiewicz.
\newblock On the expressive completeness of the propositional mu-calculus with
  respect to monadic second order logic.
\newblock In {\em Concurrency Theory, 7th International Conference}, volume
  1119 of {\em {LNCS}}, pages 263--277. Springer, 1996.

\bibitem{kooi.jancl:2007}
B.~Kooi.
\newblock Expressivity and completeness for public update logics via reduction
  axioms.
\newblock {\em Journal of Applied Non-Classical Logics}, 17(2):231--254, 2007.

\bibitem{kupferman01}
O.~Kupferman, M.~Vardi, and P.~Wolper.
\newblock Module checking.
\newblock {\em Information and Computation}, 164(2):322--344, 2001.

\bibitem{kupkeetal:2008}
C.~Kupke, A.~Kurz, and Y.~Venema.
\newblock Completeness of the finitary moss logic.
\newblock In C.~Areces and R.~Goldblatt, editors, {\em Advances in Modal Logic
  7}, pages 193--217. College Publications, 2008.

\bibitem{kusters2002memoryless}
Ralf K{\"u}sters.
\newblock Memoryless determinacy of parity games.
\newblock In {\em Automata logics, and infinite games}, pages 95--106.
  Springer, 2002.

\bibitem{LarsenNW07a}
Kim~G. Larsen, Ulrik Nyman, and Andrzej Wasowski.
\newblock Modal {I}/{O} automata for interface and product line theories.
\newblock In {\em Proceedings of the 16th European Symposium on Programming
  (ESOP'07)}, volume 4421 of {\em Lecture Notes in Computer Science}, pages
  64--79. Springer, 2007.

\bibitem{lomuscioetal:1998b}
A.R. Lomuscio and M.D. Ryan.
\newblock An algorithmic approach to knowledge evolution.
\newblock {\em Artificial Intelligence for Engineering Design, Analysis and
  Manufacturing (AIEDAM)}, 13(2), 1998.
\newblock Special issue on Temporal Logic in Engineering.

\bibitem{mazala2002infinite}
Ren{\'e} Mazala.
\newblock Infinite games.
\newblock {\em Automata logics, and infinite games}, pages 197--204, 2002.

\bibitem{milleretal:2005}
J.S. Miller and L.S. Moss.
\newblock The undecidability of iterated modal relativization.
\newblock {\em Studia Logica}, 79(3):373--407, 2005.

\bibitem{Morgan94}
C.~Morgan.
\newblock {\em Programming from Specifications: Second Edition}.
\newblock Prentice Hall International, Hempstead, UK, 1994.

\bibitem{parikhetal:2007}
R.~Parikh, L.S. Moss, and C.~Steinsvold.
\newblock Topology and epistemic logic.
\newblock In M.~Aiello, I.~Pratt-Hartmann, and J.~van Benthem, editors, {\em
  Handbook of Spatial Logics}, pages 299--341. Springer Verlag, 2007.

\bibitem{pauly:2001}
M.~Pauly.
\newblock {\em Logic for social software}.
\newblock PhD thesis, University of Amsterdam, 2001.
\newblock ILLC Dissertation Series DS-2001-10.

\bibitem{plaza:1989}
J.A. Plaza.
\newblock Logics of public communications.
\newblock In {\em Proc.\ of the 4th ISMIS}, pages 201--216. Oak Ridge National
  Laboratory, 1989.

\bibitem{TheseJBR07}
Jean-Baptiste Raclet.
\newblock {\em Quotient de sp\'ecifications pour la r\'eutilisation de
  composants}.
\newblock PhD thesis, Universit\'e de Rennes I, December 2007.
\newblock (In French).

\bibitem{Raclet2007b}
Jean-Baptiste Raclet.
\newblock Residual for component specifications.
\newblock In {\em Proc. of the 4th International Workshop on Formal Aspects of
  Component Software (FACS'07)}, volume 215 of {\em Electr. Notes Theor.
  Comput. Sci.}, pages 93--110, 2008.

\bibitem{RB-acsd09}
Jean-Baptiste Raclet, Eric Badouel, Albert Benveniste, Benoit Caillaud, and
  Roberto Passerone.
\newblock Why are modalities good for interface theories?
\newblock In {\em {P}roceedings of the 9th International Conference on
  Application of Concurrency to System Design ({ACSD}'09)}, pages 199--127.
  IEEE Computer Society Press, 2009.

\bibitem{ramadge89}
P.~Ramadge and W.~Wonham.
\newblock On the supervisory control of discrete event systems.
\newblock In {\em Proc.\ of the IEEE}, pages 81--98, 1989.

\bibitem{DBLP:conf/mfcs/RiedwegP03}
St{\'e}phane Riedweg and Sophie Pinchinat.
\newblock Quantified mu-calculus for control synthesis.
\newblock In Branislav Rovan and Peter Vojt{\'a}s, editors, {\em MFCS}, volume
  2747 of {\em Lecture Notes in Computer Science}, pages 642--651. Springer,
  2003.

\bibitem{ryanetal:2002}
M.~Ryan and P.-Y. Schobbens.
\newblock Agents and roles: Refinement in alternating-time temporal logic.
\newblock In {\em Revised Papers from the 8th International Workshop on
  Intelligent Agents VIII (ATAL '01)}, pages 100--114. Springer, 2002.

\bibitem{SVW87}
A.P. Sistla, M.Y. Vardi, and P.~Wolper.
\newblock The {C}omplementation {P}roblem for {B}uchi {A}utomata with
  {A}ppplications to {T}emporal {L}ogic.
\newblock {\em Theoretical Computer Science}, 49:217--237, 1987.

\bibitem{tsitsiklis89}
J.~Tsitsoklis.
\newblock On the control of discrete event dynamical systems.
\newblock {\em Mathematics of Control Signals and Systems}, 2(2):95--107, 1989.

\bibitem{jfak.sabotage:2005}
J.~van Benthem.
\newblock An essay on sabotage and obstruction.
\newblock In {\em Mechanizing Mathematical Reasoning}, volume 2605 of {\em LNCS
  2605}, pages 268--276. Springer, 2005.

\bibitem{jfak.lonely:2006}
J.~van Benthem.
\newblock One is a lonely number: on the logic of communication.
\newblock In {\em Logic colloquium 2002. Lecture Notes in Logic, Vol. 27},
  pages 96--129. A.K. Peters, 2006.

\bibitem{hvd.wollic:2012}
H.~van Ditmarsch.
\newblock Quantifying notes.
\newblock In {\em Proc.\ of 19th {WoLLIC}}, LNCS 7456, pages 89--109. Springer,
  2012.

\bibitem{hvdetal.loft:2009}
H.~van Ditmarsch and T.~French.
\newblock Simulation and information.
\newblock In J.~Broersen and J.-J. Meyer, editors, {\em Knowledge
  Representation for Agents and Multi-Agent Systems}, LNAI 5605. Presented at
  LOFT 2008 and KRAMAS 2008, pages 51--65. Springer, 2009.

\bibitem{hvdetal.felax:2010}
H.~van Ditmarsch, T.~French, and S.~Pinchinat.
\newblock Future event logic - axioms and complexity.
\newblock In L.~Beklemishev, V.~Goranko, and V.~Shehtman, editors, {\em
  Advances in Modal Logic, Moscow}, volume~8, pages 77--99. College
  Publications, 2010.

\bibitem{hvdetal.del:2007}
H.~van Ditmarsch, W.~van~der Hoek, and B.~Kooi.
\newblock {\em Dynamic Epistemic Logic}, volume 337 of {\em Synthese Library}.
\newblock Springer, 2007.

\bibitem{venema:2012}
Y.~Venema.
\newblock Lecture notes on the modal $\mu$-calculus.
\newblock (Draft), 2012.

\bibitem{Walukiewicz00}
I.~Walukiewicz.
\newblock Completeness of {K}ozen's axiomatisation of the propositional
  mu-calculus.
\newblock {\em INFCTRL: Information and Computation (formerly Information and
  Control)}, 157, 2000.

\bibitem{wenetal:2011}
X.~Wen, H.~Liu, and F.~Huang.
\newblock An alternative logic for knowability.
\newblock In {\em Logic, Rationality, and Interaction (Proceedings of LORI-3)},
  LNCS 6953, pages 342--355. Springer, 2011.

\bibitem{Wil99}
T.~Wilke.
\newblock {CTL}$^+$ is exponentially more succinct than {CTL}.
\newblock In {\em Proc. 19th FSTTCS}, LNCS 1738, pages 110--121. Springer,
  1999.

\bibitem{woodcocketal:1996}
J.~Woodcock and J.~Davies.
\newblock {\em Using {Z} --- Specification, Refinement and Proof}.
\newblock Prentice Hall, 1996.

\end{thebibliography}

\end{document}